\newif\iffullversion
\newif\ifnotfullversion
\fullversiontrue
\notfullversionfalse

\iffullversion
\documentclass{article}
\usepackage{kpfonts,baskervald}
\else
\documentclass[conference]{IEEEtran}
\fi

\usepackage{amsmath, amssymb, amsthm}
\usepackage{comment}
\usepackage{varwidth}
\usepackage{rotating}
\usepackage{relsize}
\usepackage{tikz}
\usetikzlibrary{positioning}

\iffullversion
  \usepackage{authblk}
\fi

\theoremstyle{definition}
\newtheorem{theorem}{Theorem}[section]
\newtheorem{lemma}[theorem]{Lemma}

\newtheorem{corollary}[theorem]{Corollary}
\newtheorem{example}[theorem]{Example}
\newcommand{\DCESHi}{DCESH$_1$}
\newcommand{\DCESHn}{DCESH}

\newcommand{\removecodespace}{\vspace{-0.8cm}}

\makeatletter
\@ifundefined{lhs2tex.lhs2tex.sty.read}%
  {\@namedef{lhs2tex.lhs2tex.sty.read}{}%
   \newcommand\SkipToFmtEnd{}%
   \newcommand\EndFmtInput{}%
   \long\def\SkipToFmtEnd#1\EndFmtInput{}%
  }\SkipToFmtEnd

\newcommand\ReadOnlyOnce[1]{\@ifundefined{#1}{\@namedef{#1}{}}\SkipToFmtEnd}
\usepackage{amstext}
\usepackage{amssymb}
\usepackage{stmaryrd}
\DeclareFontFamily{OT1}{cmtex}{}
\DeclareFontShape{OT1}{cmtex}{m}{n}
  {<5><6><7><8>cmtex8
   <9>cmtex9
   <10><10.95><12><14.4><17.28><20.74><24.88>cmtex10}{}
\DeclareFontShape{OT1}{cmtex}{m}{it}
  {<-> ssub * cmtt/m/it}{}

\DeclareFontShape{OT1}{cmtt}{bx}{n}
  {<5><6><7><8>cmtt8
   <9>cmbtt9
   <10><10.95><12><14.4><17.28><20.74><24.88>cmbtt10}{}
\DeclareFontShape{OT1}{cmtex}{bx}{n}
  {<-> ssub * cmtt/bx/n}{}

\newcommand{\Conid}[1]{\mathit{#1}}
\newcommand{\Varid}[1]{\mathit{#1}}
\newcommand{\anonymous}{\kern0.06em \vbox{\hrule\@width.5em}}
\newcommand{\plus}{\mathbin{+\!\!\!+}}

\usepackage{polytable}

\@ifundefined{mathindent}%
  {\newdimen\mathindent\mathindent\leftmargini}%
  {}%

\def\resethooks{%
  \global\let\SaveRestoreHook\empty
  \global\let\ColumnHook\empty}
\newcommand*{\savecolumns}[1][default]%
  {\g@addto@macro\SaveRestoreHook{\savecolumns[#1]}}
\newcommand*{\restorecolumns}[1][default]%
  {\g@addto@macro\SaveRestoreHook{\restorecolumns[#1]}}
\newcommand*{\aligncolumn}[2]%
  {\g@addto@macro\ColumnHook{\column{#1}{#2}}}

\resethooks

\newcommand{\onelinecommentchars}{\quad-{}- }
\newcommand{\commentbeginchars}{\enskip\{-}
\newcommand{\commentendchars}{-\}\enskip}

\newcommand{\visiblecomments}{%
  \let\onelinecomment=\onelinecommentchars
  \let\commentbegin=\commentbeginchars
  \let\commentend=\commentendchars}

\newcommand{\invisiblecomments}{%
  \let\onelinecomment=\empty
  \let\commentbegin=\empty
  \let\commentend=\empty}

\visiblecomments

\newlength{\blanklineskip}
\newcommand{\hsindent}[1]{\quad}
\let\hspre\empty
\let\hspost\empty

\EndFmtInput
\makeatother
\ReadOnlyOnce{polycode.fmt}%
\makeatletter

\newcommand{\hsnewpar}[1]%
  {{\parskip=0pt\parindent=0pt\par\vskip #1\noindent}}

\newcommand{\hscodestyle}{}

\newcommand{\sethscode}[1]%
  {\expandafter\let\expandafter\hscode\csname #1\endcsname
   \expandafter\let\expandafter\endhscode\csname end#1\endcsname}

  {\par\noindent
   \advance\leftskip\mathindent
   \hscodestyle
   \let\\=\@normalcr
   \let\hspre\(\let\hspost\)%
   \pboxed}%
  {\endpboxed\)%
   \par\noindent
   \ignorespacesafterend}

  {\hsnewpar\abovedisplayskip
   \advance\leftskip\mathindent
   \hscodestyle
   \let\hspre\(\let\hspost\)%
   \pboxed}%
  {\endpboxed%
   \hsnewpar\belowdisplayskip
   \ignorespacesafterend}

  {\hsnewpar\abovedisplayskip
   \advance\leftskip\mathindent
   \hscodestyle
   \let\\=\@normalcr
   \(\pboxed}%
  {\endpboxed\)%
   \hsnewpar\belowdisplayskip
   \ignorespacesafterend}

\newcommand{\plainhs}{\sethscode{plainhscode}}

\plainhs

  {\hsnewpar\abovedisplayskip
   \advance\leftskip\mathindent
   \hscodestyle
   \let\\=\@normalcr
   \(\parray}%
  {\endparray\)%
   \hsnewpar\belowdisplayskip
   \ignorespacesafterend}

  {\parray}{\endparray}

  {\(\parray}{\endparray\)}

\def\codeframewidth{\arrayrulewidth}
\RequirePackage{calc}

  {\parskip=\abovedisplayskip\par\noindent
   \hscodestyle
   \arrayrulewidth=\codeframewidth
   \tabular{@{}|p{\linewidth-2\arraycolsep-2\arrayrulewidth-2pt}|@{}}%
   \hline\framedhslinecorrect\\{-1.5ex}%
   \let\endoflinesave=\\
   \let\\=\@normalcr
   \(\pboxed}%
  {\endpboxed\)%
   \framedhslinecorrect\endoflinesave{.5ex}\hline
   \endtabular
   \parskip=\belowdisplayskip\par\noindent
   \ignorespacesafterend}

\newcommand{\framedhslinecorrect}[2]%
  {#1[#2]}

  {\(\def\column##1##2{}%
   \let\>\undefined\let\<\undefined\let\\\undefined
   \newcommand\>[1][]{}\newcommand\<[1][]{}\newcommand\\[1][]{}%
   \def\fromto##1##2##3{##3}%
   }{\) }%

  {\let\orighscode=\hscode
   \let\origendhscode=\endhscode
   \def\endhscode{\def\hscode{\endgroup\def\@currenvir{hscode}\\}\begingroup}
   \orighscode\def\hscode{\endgroup\def\@currenvir{hscode}}}%
  {\origendhscode
   \global\let\hscode=\orighscode
   \global\let\endhscode=\origendhscode}%

\makeatother
\EndFmtInput
\ReadOnlyOnce{agda.fmt}%

\RequirePackage[T1]{fontenc}
\RequirePackage[utf8x]{inputenc}
\RequirePackage{ucs}
\RequirePackage{amsfonts}

\DeclareUnicodeCharacter{737}{\textsuperscript{l}}
\DeclareUnicodeCharacter{8718}{\ensuremath{\blacksquare}}
\DeclareUnicodeCharacter{8759}{::}
\DeclareUnicodeCharacter{9669}{\ensuremath{\triangleleft}}
\DeclareUnicodeCharacter{8799}{\ensuremath{\stackrel{\scriptscriptstyle ?}{=}}}
\DeclareUnicodeCharacter{10214}{\ensuremath{\llbracket}}
\DeclareUnicodeCharacter{10215}{\ensuremath{\rrbracket}}
\DeclareUnicodeCharacter{10230}{\ensuremath{\longrightarrow}}
\DeclareUnicodeCharacter{8614}{\ensuremath{\mapsto}}

\renewcommand\Varid[1]{\mathord{\textsf{#1}}}
\let\Conid\Varid
\newcommand\Keyword[1]{\textsf{\textbf{#1}}}
\EndFmtInput

\renewcommand{\hscodestyle}{\smaller[.5]}

\newcommand{\indentcolumn}[1]{\aligncolumn{#1}{@{}>{\hspre \quad}l<{\hspost}@{}}}
\newcommand{\rightaligncolumn}[1]{\aligncolumn{#1}{@{}>{\hspre}r<{\hspost}@{}}}
\newcommand{\centeraligncolumn}[1]{\aligncolumn{#1}{@{}>{\hspre}c<{\hspost}@{}}}
\newcommand\MyConid[1]{\mathord{\textsf{\textbf{#1}}}}
\renewcommand\Keyword[1]{\textsf{\underline{#1}}}
\renewcommand\Varid[1]{\textsf{#1}}

\DeclareUnicodeCharacter{411}{\ensuremath{\lambda}}
\DeclareUnicodeCharacter{8343}{\ensuremath{\lambda}}
\DeclareUnicodeCharacter{9656}{\ensuremath{\blacktriangleright}}

\newcommand{\dummy}{\Varid{\_}}

\iffullversion
\title{Distributed call-by-value machines}
\else
\title{Distributed call-by-value machines \\ {\large Extended abstract}}
\fi

\newcommand\theauthor{Olle Fredriksson}
\newcommand\theaffiliation{University of Birmingham, UK}

\iffullversion
  \author\theauthor
  \affil\theaffiliation
\else
  \author{\IEEEauthorblockN\theauthor
          \IEEEauthorblockA\theaffiliation
  }
\fi

\begin{document}
\iffullversion
\date{}
\fi
\maketitle


\begin{abstract}
We present a new abstract machine, called \DCESHn{}, which describes the
execution of higher-order programs running in distributed architectures.
\DCESHn{} implements a generalised form of Remote Procedure Call that supports
calling higher-order functions across node boundaries, without sending actual
code.  Our starting point is a variant of the SECD machine that we call the CES
machine, which implements reduction for untyped call-by-value PCF.  We
successively add the features that we need for distributed execution and show
the correctness of each addition.  First we add heaps, forming the CESH
machine, which provides features necessary for more efficient execution, and
show that there is a bisimulation between the CES and the CESH machine.  Then
we construct a two-level operational semantics, where the high level is a
network of communicating machines, and the low level is given by local machine
transitions.  Using these networks, we arrive at our final system, the
distributed CESH machine (\DCESHn{}).  We show that there is a bisimulation
relation also between the CESH machine and the \DCESHn{} machine.
All the technical results have been formalised and proved correct
in Agda, and a prototype compiler has been developed.

\end{abstract}

\iffullversion
  \newpage
  \tableofcontents
  \newpage
\fi

\section{Seamless computing}
Suppose we need to program a system in which the function {\textsmaller[.5]{\ensuremath{\Conid{F}}}} runs on node {\textsmaller[.5]{\ensuremath{\Conid{A}}}}
in a distributed system, for instance because {\textsmaller[.5]{\ensuremath{\Conid{F}}}} depends on a local resource
residing on node {\textsmaller[.5]{\ensuremath{\Conid{A}}}}. Suppose further that we need to write a program
{\textsmaller[.5]{\ensuremath{\Conid{G}}}}, running on node {\textsmaller[.5]{\ensuremath{\Conid{B}}}}, that uses {\textsmaller[.5]{\ensuremath{\Conid{F}}}}.  How to achieve this depends on what
programming language or library for distributed computing we choose. One of the
most prominent ways to do it is using message passing, for instance with the
Message-Passing Interface \cite{gropp1999using}.  This involves writing {\textsmaller[.5]{\ensuremath{\Conid{F}}}} and
{\textsmaller[.5]{\ensuremath{\Conid{G}}}} as separate processes, and explicitly constructing messages that are sent
between them.

Suppose now that our specification changes: A part {\textsmaller[.5]{\ensuremath{\Conid{F'}}}} of {\textsmaller[.5]{\ensuremath{\Conid{F}}}} actually needs
to run on a \emph{third} node {\textsmaller[.5]{\ensuremath{\Conid{C}}}}. Using conventional languages or libraries,
this means that we have to rewrite big parts of the program since a substantial
part of it deals with the architecture-specific details of the problem.
Languages with support for Remote Procedure
Calls~\cite{DBLP:journals/tocs/BirrelN84} can help mitigate this, since
such a call has the same syntax as a local procedure call, but will
not work if {\textsmaller[.5]{\ensuremath{\Conid{F'}}}} is a higher-order function that is invoked with a function
as its argument.
In previous
papers~\cite{DBLP:conf/lics/FredrikssonG13,DBLP:conf/tgc/FredrikssonG12} we
suggest the following alternative way to express the two programs above:
\begin{hscode}\SaveRestoreHook
\column{B}{@{}>{\hspre}l<{\hspost}@{}}%
\column{17}{@{}>{\hspre}l<{\hspost}@{}}%
\column{28}{@{}>{\hspre}l<{\hspost}@{}}%
\column{E}{@{}>{\hspre}l<{\hspost}@{}}%
\>[B]{}\Keyword{let}\;\Conid{F}\;\mathrel{=}\;\{\mskip1.5mu \Varid{...}\;{}\<[17]%
\>[17]{}\Conid{F'}\;{}\<[28]%
\>[28]{}\Varid{...}\mskip1.5mu\}\;\MyConid{$\boldsymbol{@}$}\;\Conid{A}\;\Keyword{in}\;\{\mskip1.5mu \Conid{G}\mskip1.5mu\}\;\MyConid{$\boldsymbol{@}$}\;\Conid{B}{}\<[E]%
\\
\>[B]{}\Keyword{let}\;\Conid{F}\;\mathrel{=}\;\{\mskip1.5mu \Varid{...}\;\{\mskip1.5mu {}\<[17]%
\>[17]{}\Conid{F'}\mskip1.5mu\}\;\MyConid{$\boldsymbol{@}$}\;\Conid{C}\;{}\<[28]%
\>[28]{}\Varid{...}\mskip1.5mu\}\;\MyConid{$\boldsymbol{@}$}\;\Conid{A}\;\Keyword{in}\;\{\mskip1.5mu \Conid{G}\mskip1.5mu\}\;\MyConid{$\boldsymbol{@}$}\;\Conid{B}{}\<[E]%
\ColumnHook
\end{hscode}\resethooks
Here we write the whole program as if it was running on a single computer, and
use pragma-like annotations, written {\textsmaller[.5]{\ensuremath{\{\mskip1.5mu \Varid{x}\mskip1.5mu\}\;\MyConid{$\boldsymbol{@}$}\;\Conid{A}}}}, to indicate the node of
execution. We call such annotations \emph{locus specifiers}. The compiler uses
the annotations to automatically handle architecture-specific details like
communication. We call this \emph{seamless computing}.  A key feature is full
support for higher-order functions, even across node boundaries, without
sending actual code (in contrast to e.g. Remote
Evaluation~\cite{DBLP:journals/toplas/StamosG90}). This is important for full
generality, since it is not always the case that all code is meaningful on all
nodes (for example because of resource locality or platform differences).

Our previous work enables writing these programs but uses an execution model
based on game semantics that is vastly different from conventional compilation
techniques.  In this paper we instead develop an approach which is a conservative
extension of existing abstract machines. This means that the vast literature on
compiler optimisation more readily applies, and makes it possible to interface
with legacy code.  The key idea in this work, like in our previous work, is
that computational phenomena like function calls can be subsumed by simple
communication protocols. We assume that a run-time infrastructure can handle
system-level aspects associated with distribution such as failure, load
balancing, global reset, and so on.

\paragraph*{Technical outline}
To achieve the goal of an abstract machine for seamless computing, we make
gradual refinements to a machine, based on Landin's SECD
machine~\cite{Landin64}, that we call the \emph{CES machine}
(Sec.~\ref{section:CES}). The first change is to add heaps 
\iffullversion 
(Sec.~\ref{section:Heaps})
\fi
for dynamically allocating closures, forming the \emph{CESH machine}
(Sec.~\ref{section:CESH}), which provides features necessary for more efficient
execution, and we show the CES and CESH machines to be
\iffullversion
bisimilar (Sec.~\ref{section:CESH-bisim}).
\fi
We then add communication primitives (synchronous and asynchronous) by defining
a general form of networks of nodes that run an instance of an underlying abstract
machine (Sec.~\ref{section:Networks}).  Using these networks, we illustrate the
idea of subsuming function calls by communication protocols by constructing a
degenerate distributed machine, \DCESHi{} (Sec.~\ref{section:ADCESH}), that decomposes
some machine instructions into message passing, but only runs on one node.
Finally, the main contribution is the fully distributed CESH machine (\DCESHn{}, Sec.~\ref{section:DCESH}), which is shown to be bisimilar to the CESH
\iffullversion
machine (Sec.~\ref{section:DCESH-bisim}).
\fi

\paragraph*{Formalisation in Agda}
The theorems that we present in this paper have been proved correct in
Agda~\cite{norell:thesis}, an interactive proof assistant and programming
language based on intuitionistic type theory.  The definitions and proofs in
this paper are intricate and often consist of many cases, so carrying them out manually would
be error-prone and arduous.  Agda has been a helpful tool in producing these
proofs, and also allows us to easily play with alternative definitions (even wrong
ones). To eliminate another source of error, we do not adopt the usual practice
of writing up the results in informal mathematics; in fact, the paper is built
from a collection of literate Agda source files and the code blocks come
directly from the
\ifnotfullversion
formalisation \cite{SourceCode}.
\else
formalisation.
\fi
Although our work is not
about Agda \emph{per se}, we believe that this presentation is beneficial also
to you, the reader, since you can trust that the propositions do not contain
mistakes.
Since Agda builds on a constructive foundation, it also means that the
formalisation of an abstract machine in Agda can act as a verified prototype
implementation.
\iffullversion
\paragraph*{Syntax and notation for code}
We assume a certain familiarity with the syntax of Agda, but since it is close
to that of several popular functional programming languages we believe that
this will not cause much difficulty for the audience.
We will use {\textsmaller[.5]{\ensuremath{\star}}} for the type of types.
We will use \emph{implicit parameters}, written e.g. {\textsmaller[.5]{\ensuremath{\Varid{f}\;\mathbin{:}\;\{\mskip1.5mu \Conid{A}\;\mathbin{:}\;\star\mskip1.5mu\}\;\Varid{→}\;\Varid{...}}}}
which means that {\textsmaller[.5]{\ensuremath{\Varid{f}}}} takes, as its first argument, a type {\textsmaller[.5]{\ensuremath{\Conid{A}}}} that does not
need to be explicitly spelled out when it can be inferred from other
arguments.
We will sometimes use the same name for constructors of different types, and
rely on context for disambiguation.
Constructors will be written in $\MyConid{bold face}$ and keywords $\Keyword{underlined}$.
We make liberal use of Agda's ability to define \emph{mixfix} operators like
{\textsmaller[.5]{\ensuremath{{\MyConid{if0}\dummy\MyConid{then}\dummy\MyConid{else}\dummy}}}} which is a constructor that accepts arguments
in the positions of the underscores, as in {\textsmaller[.5]{\ensuremath{\MyConid{if0}\;\Varid{b}\;\MyConid{then}\;\Varid{t}\;\MyConid{else}\;\Varid{f}}}}.
\fi

This paper is organised as
follows, where the arrows denote dependence, the lines with
$\sim$ symbols bisimulations, and the parenthesised numerals section numbers:
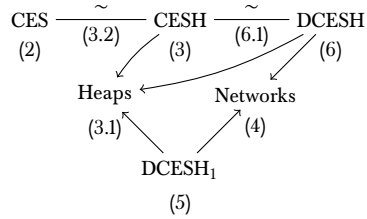
\begin{center}
\begin{tikzpicture}[node distance=2.0cm, on grid,every node/.style={scale=0.8}]
  \node[label=below:(\ref{section:CES})]   (CES) {CES};
  \node[label=below:(\ref{section:CESH})]  (CESH) [right=of CES] {CESH};
  \node[label=below:(\ref{section:DCESH})] (DCESH) [right=of CESH] {\DCESHn};
\iffullversion
  \node[label=below:(\ref{section:Heaps})] (Heaps) [below left=1.0cm and 1.00cm of CESH] {Heaps};
\fi
  \node[label=below:(\ref{section:Networks})] (Networks) [below right=1.0cm and 1.00cm of CESH] {Networks};
  \node[label=below:(\ref{section:ADCESH})] (ADCESH) [below=of CESH] {\DCESHi};
  \draw (CES) to node[above] {$\sim$}
                 node[below] {(\ref{section:CESH-bisim})}
                 (CESH);
  \draw (CESH) to node[above] {$\sim$}
                  node[below] {(\ref{section:DCESH-bisim})}
                  (DCESH);
  \draw[->] (CESH) to[bend right=10] (Heaps);
  \draw[->] (DCESH) to[bend left=10] (Heaps);
  \draw[->] (DCESH) to (Networks);
  \draw[->] (ADCESH) to (Heaps);
  \draw[->] (ADCESH) to (Networks);
\end{tikzpicture}
\end{center}

\section{The CES machine} \label{section:CES}
Our goal is to make a compiler for a programming language with
locus specifiers that is based on conventional compilation techniques.
A very common technique is the usage of \emph{abstract machines} to describe
the evaluation at a level low enough to be used as a basis for compilation.
The starting point for our work is based on a variation of Landin's
well-studied SECD machine~\cite{Landin64} called Modern SECD~\cite{ModernSECD}.
Modern SECD itself can be traced back to the SECD machine of
Henderson~\cite{DBLP:books/daglib/0068837}, in that both use bytecode for the
control component of the 
\iffullversion
machine
(and so use explicit return instructions);
\fi
and
to the CEK machine of Felleisen~\cite{Felleisen:1986:CEK}, in that they both
place the continuations that originally resided in the dump
\iffullversion
(the D component)
\fi
directly on the
\iffullversion
stack (the S component),
\fi
simplifying the machine configurations.

We choose to call this variation the \emph{CES machine} because of its three
configuration constituents. This machine is important for us since it will be
used as the \emph{specification} for the elaborated machines that we later
construct. We will show that their termination and divergence behaviour is the
same as that of CES by constructing bisimulation relations.

A CES configuration ({\textsmaller[.5]{\ensuremath{\Conid{Config}}}}) is a tuple consisting of a fragment of code
({\textsmaller[.5]{\ensuremath{\Conid{Code}}}}), an environment ({\textsmaller[.5]{\ensuremath{\Conid{Env}}}}), and a stack ({\textsmaller[.5]{\ensuremath{\Conid{Stack}}}}).  Evaluation begins
with an empty stack and environment, and then follows a \emph{stack
discipline}. Sub-terms push their result on the stack so that their super-terms
can consume them. When (and if) the evaluation terminates, the program's result
is the sole stack element.
\paragraph*{Source language}
We show how to compile untyped call-by-value PCF \cite{DBLP:journals/tcs/Plotkin77}.
The source language has constructors for lambda abstractions ({\textsmaller[.5]{\ensuremath{\MyConid{$\boldsymbol{\lambda}$}\;\Varid{t}}}}), applications ({\textsmaller[.5]{\ensuremath{\Varid{t}\;\MyConid{\$}\;\Varid{t'}}}}), and variables ({\textsmaller[.5]{\ensuremath{\MyConid{var}\;\Varid{n}}}}).  Our representation uses De Bruijn
indices \cite{DeBruijn}, so a variable is simply a natural number.
\ifnotfullversion
Additionally, we have natural number literals ({\textsmaller[.5]{\ensuremath{\MyConid{lit}\;\Varid{n}}}}), binary operations on
them ({\textsmaller[.5]{\ensuremath{\MyConid{op}\;\Varid{f}\;\Varid{t}\;\Varid{t'}}}}), and conditionals ({\textsmaller[.5]{\ensuremath{\MyConid{if0}\;\Varid{t}\;\MyConid{then}\;\Varid{t₀}\;\MyConid{else}\;\Varid{t₁}}}}).
\fi
\iffullversion
\savecolumns
\begin{hscode}\SaveRestoreHook
\column{B}{@{}>{\hspre}l<{\hspost}@{}}%
\column{3}{@{}>{\hspre}l<{\hspost}@{}}%
\column{9}{@{}>{\hspre}l<{\hspost}@{}}%
\column{E}{@{}>{\hspre}l<{\hspost}@{}}%
\>[B]{}\Keyword{data}\;\Conid{Term}\;\mathbin{:}\;\star\;\Keyword{where}{}\<[E]%
\\
\>[B]{}\hsindent{3}{}\<[3]%
\>[3]{}{\MyConid{$\boldsymbol{\lambda}$}\dummy}\;{}\<[9]%
\>[9]{}\mathbin{:}\;\Conid{Term}\;\Varid{→}\;\Conid{Term}{}\<[E]%
\\
\>[B]{}\hsindent{3}{}\<[3]%
\>[3]{}{\dummy\MyConid{\$}\dummy}\;{}\<[9]%
\>[9]{}\mathbin{:}\;(\Varid{t}\;\Varid{t'}\;\mathbin{:}\;\Conid{Term})\;\Varid{→}\;\Conid{Term}{}\<[E]%
\\
\>[B]{}\hsindent{3}{}\<[3]%
\>[3]{}\MyConid{var}\;{}\<[9]%
\>[9]{}\mathbin{:}\;\Conid{ℕ}\;\Varid{→}\;\Conid{Term}{}\<[E]%
\ColumnHook
\end{hscode}\resethooks
We also have natural number literals, binary operations on them, and
conditionals:
\restorecolumns
\begin{hscode}\SaveRestoreHook
\column{B}{@{}>{\hspre}l<{\hspost}@{}}%
\column{3}{@{}>{\hspre}l<{\hspost}@{}}%
\column{9}{@{}>{\hspre}l<{\hspost}@{}}%
\column{19}{@{}>{\hspre}l<{\hspost}@{}}%
\column{E}{@{}>{\hspre}l<{\hspost}@{}}%
\>[3]{}\MyConid{lit}\;{}\<[9]%
\>[9]{}\mathbin{:}\;\Conid{ℕ}\;\Varid{→}\;\Conid{Term}{}\<[E]%
\\
\>[3]{}\MyConid{op}\;{}\<[9]%
\>[9]{}\mathbin{:}\;(\Varid{f}\;\mathbin{:}\;\Conid{ℕ}\;\Varid{→}\;\Conid{ℕ}\;\Varid{→}\;\Conid{ℕ})\;(\Varid{t}\;\Varid{t'}\;\mathbin{:}\;\Conid{Term})\;\Varid{→}\;\Conid{Term}{}\<[E]%
\\
\>[3]{}{\MyConid{if0}\dummy\MyConid{then}\dummy\MyConid{else}\dummy}\;{}\<[19]%
\>[19]{}\mathbin{:}\;(\Varid{b}\;\Varid{t}\;\Varid{f}\;\mathbin{:}\;\Conid{Term})\;\Varid{→}\;\Conid{Term}{}\<[E]%
\ColumnHook
\end{hscode}\resethooks
\fi
\iffullversion
The language can be thought of as an intermediate representation for a
compiler which may expose a more sugary front-end language.
\fi
Because
\iffullversion
it
\fi
is untyped, we can express fixed-point combinators without adding additional
constructors.

\iffullversion
We define the bytecode, {\textsmaller[.5]{\ensuremath{\Conid{Code}}}}, that the machine will operate on. A
fragment of {\textsmaller[.5]{\ensuremath{\Conid{Code}}}} is a list of instructions, {\textsmaller[.5]{\ensuremath{\Conid{Instr}}}}, terminated by
{\textsmaller[.5]{\ensuremath{\MyConid{END}}}}, {\textsmaller[.5]{\ensuremath{\MyConid{RET}}}}, or a conditional {\textsmaller[.5]{\ensuremath{\MyConid{COND}}}} which has code fragments for its
two branches:
\savecolumns
\begin{hscode}\SaveRestoreHook
\column{B}{@{}>{\hspre}l<{\hspost}@{}}%
\column{3}{@{}>{\hspre}l<{\hspost}@{}}%
\column{5}{@{}>{\hspre}l<{\hspost}@{}}%
\column{14}{@{}>{\hspre}l<{\hspost}@{}}%
\column{22}{@{}>{\hspre}l<{\hspost}@{}}%
\column{E}{@{}>{\hspre}l<{\hspost}@{}}%
\>[B]{}\Keyword{mutual}{}\<[E]%
\\
\>[B]{}\hsindent{3}{}\<[3]%
\>[3]{}\Keyword{data}\;\Conid{Instr}\;\mathbin{:}\;\star\;\Keyword{where}{}\<[E]%
\\
\>[3]{}\hsindent{2}{}\<[5]%
\>[5]{}\MyConid{VAR}\;{}\<[14]%
\>[14]{}\mathbin{:}\;\Conid{ℕ}\;{}\<[22]%
\>[22]{}\Varid{→}\;\Conid{Instr}{}\<[E]%
\\
\>[3]{}\hsindent{2}{}\<[5]%
\>[5]{}\MyConid{CLOS}\;{}\<[14]%
\>[14]{}\mathbin{:}\;\Conid{Code}\;{}\<[22]%
\>[22]{}\Varid{→}\;\Conid{Instr}{}\<[E]%
\\
\>[3]{}\hsindent{2}{}\<[5]%
\>[5]{}\MyConid{APPL}\;{}\<[14]%
\>[14]{}\mathbin{:}\;\Conid{Instr}{}\<[E]%
\\
\>[3]{}\hsindent{2}{}\<[5]%
\>[5]{}\MyConid{LIT}\;{}\<[14]%
\>[14]{}\mathbin{:}\;\Conid{ℕ}\;{}\<[22]%
\>[22]{}\Varid{→}\;\Conid{Instr}{}\<[E]%
\\
\>[3]{}\hsindent{2}{}\<[5]%
\>[5]{}\MyConid{OP}\;{}\<[14]%
\>[14]{}\mathbin{:}\;(\Conid{ℕ}\;\Varid{→}\;\Conid{ℕ}\;\Varid{→}\;\Conid{ℕ})\;\Varid{→}\;\Conid{Instr}{}\<[E]%
\ColumnHook
\end{hscode}\resethooks
\removecodespace
\restorecolumns
\begin{hscode}\SaveRestoreHook
\column{B}{@{}>{\hspre}l<{\hspost}@{}}%
\column{3}{@{}>{\hspre}l<{\hspost}@{}}%
\column{5}{@{}>{\hspre}l<{\hspost}@{}}%
\column{14}{@{}>{\hspre}l<{\hspost}@{}}%
\column{23}{@{}>{\hspre}l<{\hspost}@{}}%
\column{E}{@{}>{\hspre}l<{\hspost}@{}}%
\>[3]{}\Keyword{data}\;\Conid{Code}\;\mathbin{:}\;\star\;\Keyword{where}{}\<[E]%
\\
\>[3]{}\hsindent{2}{}\<[5]%
\>[5]{}{\dummy\MyConid{;}\dummy}\;{}\<[14]%
\>[14]{}\mathbin{:}\;\Conid{Instr}\;{}\<[23]%
\>[23]{}\Varid{→}\;\Conid{Code}\;\Varid{→}\;\Conid{Code}{}\<[E]%
\\
\>[3]{}\hsindent{2}{}\<[5]%
\>[5]{}\MyConid{COND}\;{}\<[14]%
\>[14]{}\mathbin{:}\;\Conid{Code}\;{}\<[23]%
\>[23]{}\Varid{→}\;\Conid{Code}\;\Varid{→}\;\Conid{Code}{}\<[E]%
\\
\>[3]{}\hsindent{2}{}\<[5]%
\>[5]{}\MyConid{END}\;{}\<[14]%
\>[14]{}\mathbin{:}\;\Conid{Code}{}\<[E]%
\\
\>[3]{}\hsindent{2}{}\<[5]%
\>[5]{}\MyConid{RET}\;{}\<[14]%
\>[14]{}\mathbin{:}\;\Conid{Code}{}\<[E]%
\ColumnHook
\end{hscode}\resethooks
\fi
The main work of compilation is done by the function
{\textsmaller[.5]{\ensuremath{\Varid{compile'}}}}, which takes a term {\textsmaller[.5]{\ensuremath{\Varid{t}}}} to be compiled and a fragment of
code {\textsmaller[.5]{\ensuremath{\Varid{c}}}} that is placed after the instructions that the compilation
emits.
\ifnotfullversion
The bold upper-case names ({\textsmaller[.5]{\ensuremath{\MyConid{CLOS}}}}, {\textsmaller[.5]{\ensuremath{\MyConid{VAR}}}}, and so on) are the
bytecode instructions, which are sequenced using {\textsmaller[.5]{\ensuremath{{\dummy\MyConid{;}\dummy}}}}:
\fi
\begin{hscode}\SaveRestoreHook
\column{B}{@{}>{\hspre}l<{\hspost}@{}}%
\column{3}{@{}>{\hspre}l<{\hspost}@{}}%
\column{23}{@{}>{\hspre}l<{\hspost}@{}}%
\column{26}{@{}>{\hspre}l<{\hspost}@{}}%
\column{38}{@{}>{\hspre}l<{\hspost}@{}}%
\column{E}{@{}>{\hspre}l<{\hspost}@{}}%
\>[B]{}\Varid{compile'}\;\mathbin{:}\;\Conid{Term}\;\Varid{→}\;\Conid{Code}\;\Varid{→}\;\Conid{Code}{}\<[E]%
\\
\>[B]{}\Varid{compile'}\;(\MyConid{$\boldsymbol{\lambda}$}\;\Varid{t})\;{}\<[23]%
\>[23]{}\Varid{c}\;{}\<[26]%
\>[26]{}\mathrel{=}\;\MyConid{CLOS}\;(\Varid{compile'}\;\Varid{t}\;\MyConid{RET})\;\MyConid{;}\;\Varid{c}{}\<[E]%
\\
\>[B]{}\Varid{compile'}\;(\Varid{t}\;\MyConid{\$}\;\Varid{t'})\;{}\<[23]%
\>[23]{}\Varid{c}\;{}\<[26]%
\>[26]{}\mathrel{=}\;\Varid{compile'}\;{}\<[38]%
\>[38]{}\Varid{t}\;(\Varid{compile'}\;\Varid{t'}\;(\MyConid{APPL}\;\MyConid{;}\;\Varid{c})){}\<[E]%
\\
\>[B]{}\Varid{compile'}\;(\MyConid{var}\;\Varid{x})\;{}\<[23]%
\>[23]{}\Varid{c}\;{}\<[26]%
\>[26]{}\mathrel{=}\;\MyConid{VAR}\;\Varid{x}\;\MyConid{;}\;\Varid{c}{}\<[E]%
\\
\>[B]{}\Varid{compile'}\;(\MyConid{lit}\;\Varid{n})\;{}\<[23]%
\>[23]{}\Varid{c}\;{}\<[26]%
\>[26]{}\mathrel{=}\;\MyConid{LIT}\;\Varid{n}\;\MyConid{;}\;\Varid{c}{}\<[E]%
\\
\>[B]{}\Varid{compile'}\;(\MyConid{op}\;\Varid{f}\;\Varid{t}\;\Varid{t'})\;{}\<[23]%
\>[23]{}\Varid{c}\;{}\<[26]%
\>[26]{}\mathrel{=}\;\Varid{compile'}\;{}\<[38]%
\>[38]{}\Varid{t'}\;(\Varid{compile'}\;\Varid{t}\;(\MyConid{OP}\;\Varid{f}\;\MyConid{;}\;\Varid{c})){}\<[E]%
\\
\>[B]{}\Varid{compile'}\;(\MyConid{if0}\;\Varid{b}\;\MyConid{then}\;\Varid{t}\;\MyConid{else}\;\Varid{f})\;\Varid{c}\;\mathrel{=}\;{}\<[E]%
\\
\>[B]{}\hsindent{3}{}\<[3]%
\>[3]{}\Varid{compile'}\;\Varid{b}\;(\MyConid{COND}\;(\Varid{compile'}\;\Varid{t}\;\Varid{c})\;(\Varid{compile'}\;\Varid{f}\;\Varid{c})){}\<[E]%
\ColumnHook
\end{hscode}\resethooks
It should be apparent that the instructions correspond closely
to the constructs of the source language but are sequentialised.
Compilation of a term is simply a call to {\textsmaller[.5]{\ensuremath{\Varid{compile'}}}}, terminated by
{\textsmaller[.5]{\ensuremath{\MyConid{END}}}}:
\iffullversion
\begin{hscode}\SaveRestoreHook
\column{B}{@{}>{\hspre}l<{\hspost}@{}}%
\column{E}{@{}>{\hspre}l<{\hspost}@{}}%
\>[B]{}\Varid{compile}\;\mathbin{:}\;\Conid{Term}\;\Varid{→}\;\Conid{Code}{}\<[E]%
\\
\>[B]{}\Varid{compile}\;\Varid{t}\;\mathrel{=}\;\Varid{compile'}\;\Varid{t}\;\MyConid{END}{}\<[E]%
\ColumnHook
\end{hscode}\resethooks
\fi
\begin{example}[{\textsmaller[.5]{\ensuremath{\Varid{codeExample}}}}]
The term {\textsmaller[.5]{\ensuremath{(\Varid{λx.}\;\Varid{x})\;(\Varid{λx}\;\Varid{y.}\;\Varid{x})}}} is compiled as follows:
\begin{hscode}\SaveRestoreHook
\column{B}{@{}>{\hspre}l<{\hspost}@{}}%
\column{3}{@{}>{\hspre}l<{\hspost}@{}}%
\column{E}{@{}>{\hspre}l<{\hspost}@{}}%
\>[B]{}\Varid{compile}\;((\MyConid{$\boldsymbol{\lambda}$}\;\MyConid{var}\;\Varid{0})\;\MyConid{\$}\;(\MyConid{$\boldsymbol{\lambda}$}\;(\MyConid{$\boldsymbol{\lambda}$}\;\MyConid{var}\;\Varid{1})))\;\mathrel{=}\;{}\<[E]%
\\
\>[B]{}\hsindent{3}{}\<[3]%
\>[3]{}\MyConid{CLOS}\;(\MyConid{VAR}\;\Varid{0}\;\MyConid{;}\;\MyConid{RET})\;\MyConid{;}\;{}\<[E]%
\\
\>[B]{}\hsindent{3}{}\<[3]%
\>[3]{}\MyConid{CLOS}\;(\MyConid{CLOS}\;(\MyConid{VAR}\;\Varid{1}\;\MyConid{;}\;\MyConid{RET})\;\MyConid{;}\;\MyConid{RET})\;\MyConid{;}\;\MyConid{APPL}\;\MyConid{;}\;\MyConid{END}{}\<[E]%
\ColumnHook
\end{hscode}\resethooks
Compilation first emits two {\textsmaller[.5]{\ensuremath{\MyConid{CLOS}}}} instructions containing the code
of the function and its argument. The {\textsmaller[.5]{\ensuremath{\MyConid{APPL}}}} instruction is then used to
perform the actual application.
\end{example}

\ifnotfullversion
Environments ({\textsmaller[.5]{\ensuremath{\Conid{Env}}}}) are lists of values ({\textsmaller[.5]{\ensuremath{\Conid{List}\;\Conid{Value}}}}).  A value is
either a natural number ({\textsmaller[.5]{\ensuremath{\MyConid{nat}\;\Varid{n}}}}) or a closure ({\textsmaller[.5]{\ensuremath{\MyConid{clos}\;\Varid{cl}}}}).  A closure ({\textsmaller[.5]{\ensuremath{\Conid{Closure}}}})
is a fragment of code paired with an environment ({\textsmaller[.5]{\ensuremath{\Conid{Code}\;\Varid{×}\;\Conid{Env}}}}).
\else
We mutually define values, closures and environments. A closure is a
code fragment paired with an environment. A value is either a natural
number literal or a closure.  Since we are working in a call-by-value
setting an environment is a list of values.
\begin{hscode}\SaveRestoreHook
\column{B}{@{}>{\hspre}l<{\hspost}@{}}%
\column{3}{@{}>{\hspre}l<{\hspost}@{}}%
\column{5}{@{}>{\hspre}l<{\hspost}@{}}%
\column{11}{@{}>{\hspre}l<{\hspost}@{}}%
\column{22}{@{}>{\hspre}l<{\hspost}@{}}%
\column{E}{@{}>{\hspre}l<{\hspost}@{}}%
\>[B]{}\Keyword{mutual}{}\<[E]%
\\
\>[B]{}\hsindent{3}{}\<[3]%
\>[3]{}\Conid{Closure}\;\mathrel{=}\;\Conid{Code}\;\Varid{×}\;\Conid{Env}{}\<[E]%
\\
\>[B]{}\hsindent{3}{}\<[3]%
\>[3]{}\Keyword{data}\;\Conid{Value}\;\mathbin{:}\;\star\;\Keyword{where}{}\<[E]%
\\
\>[3]{}\hsindent{2}{}\<[5]%
\>[5]{}\MyConid{nat}\;{}\<[11]%
\>[11]{}\mathbin{:}\;\Conid{ℕ}\;{}\<[22]%
\>[22]{}\Varid{→}\;\Conid{Value}{}\<[E]%
\\
\>[3]{}\hsindent{2}{}\<[5]%
\>[5]{}\MyConid{clos}\;{}\<[11]%
\>[11]{}\mathbin{:}\;\Conid{Closure}\;{}\<[22]%
\>[22]{}\Varid{→}\;\Conid{Value}{}\<[E]%
\\
\>[B]{}\hsindent{3}{}\<[3]%
\>[3]{}\Conid{Env}\;\mathrel{=}\;\Conid{List}\;\Conid{Value}{}\<[E]%
\ColumnHook
\end{hscode}\resethooks
\fi
\ifnotfullversion
Stacks ({\textsmaller[.5]{\ensuremath{\Conid{Stack}}}}) are lists of stack elements ({\textsmaller[.5]{\ensuremath{\Conid{List}\;\Conid{StackElem}}}}), where
stack elements are either values ({\textsmaller[.5]{\ensuremath{\MyConid{val}\;\Varid{v}}}}) or continuations ({\textsmaller[.5]{\ensuremath{\MyConid{cont}\;\Varid{cl}}}}), represented by closures.
\else
A stack is a list of stack elements, defined to be either values
or continuations (represented by closures):
\begin{hscode}\SaveRestoreHook
\column{B}{@{}>{\hspre}l<{\hspost}@{}}%
\column{3}{@{}>{\hspre}l<{\hspost}@{}}%
\column{9}{@{}>{\hspre}l<{\hspost}@{}}%
\column{20}{@{}>{\hspre}l<{\hspost}@{}}%
\column{E}{@{}>{\hspre}l<{\hspost}@{}}%
\>[B]{}\Keyword{data}\;\Conid{StackElem}\;\mathbin{:}\;\star\;\Keyword{where}{}\<[E]%
\\
\>[B]{}\hsindent{3}{}\<[3]%
\>[3]{}\MyConid{val}\;{}\<[9]%
\>[9]{}\mathbin{:}\;\Conid{Value}\;{}\<[20]%
\>[20]{}\Varid{→}\;\Conid{StackElem}{}\<[E]%
\\
\>[B]{}\hsindent{3}{}\<[3]%
\>[3]{}\MyConid{cont}\;{}\<[9]%
\>[9]{}\mathbin{:}\;\Conid{Closure}\;{}\<[20]%
\>[20]{}\Varid{→}\;\Conid{StackElem}{}\<[E]%
\\
\>[B]{}\Conid{Stack}\;\mathrel{=}\;\Conid{List}\;\Conid{StackElem}{}\<[E]%
\ColumnHook
\end{hscode}\resethooks
A configuration is, as stated, a tuple consisting of a code fragment, an
environment and a stack:
\begin{hscode}\SaveRestoreHook
\column{B}{@{}>{\hspre}l<{\hspost}@{}}%
\column{E}{@{}>{\hspre}l<{\hspost}@{}}%
\>[B]{}\Conid{Config}\;\mathrel{=}\;\Conid{Code}\;\Varid{×}\;\Conid{Env}\;\Varid{×}\;\Conid{Stack}{}\<[E]%
\ColumnHook
\end{hscode}\resethooks
\fi

\iffullversion
\begin{sidewaysfigure}
\else
\begin{figure*}[!t]
\fi
\centering
\begin{varwidth}[t]{\textwidth}
\rightaligncolumn{36}
\rightaligncolumn{54}
\indentcolumn{3}
\iffullversion
\begin{hscode}\SaveRestoreHook
\column{B}{@{}>{\hspre}l<{\hspost}@{}}%
\column{E}{@{}>{\hspre}l<{\hspost}@{}}%
\>[B]{}\Keyword{data}\;{\dummy\xrightarrow[\Varid{CES}]{}\dummy}\;\mathbin{:}\;\Conid{Rel}\;\Conid{Config}\;\Conid{Config}\;\Keyword{where}{}\<[E]%
\ColumnHook
\end{hscode}\resethooks
\removecodespace
\indentcolumn{3}
\fi
\begin{hscode}\SaveRestoreHook
\column{B}{@{}>{\hspre}l<{\hspost}@{}}%
\column{3}{@{}>{\hspre}l<{\hspost}@{}}%
\column{13}{@{}>{\hspre}l<{\hspost}@{}}%
\column{36}{@{}>{\hspre}l<{\hspost}@{}}%
\column{54}{@{}>{\hspre}l<{\hspost}@{}}%
\column{92}{@{}>{\hspre}l<{\hspost}@{}}%
\column{103}{@{}>{\hspre}l<{\hspost}@{}}%
\column{105}{@{}>{\hspre}l<{\hspost}@{}}%
\column{E}{@{}>{\hspre}l<{\hspost}@{}}%
\>[3]{}\MyConid{VAR}\;{}\<[13]%
\>[13]{}\mathbin{:}\;\Varid{∀}\;\{\mskip1.5mu \Varid{n}\;\Varid{c}\;\Varid{e}\;\Varid{s}\;\Varid{v}\mskip1.5mu\}\;\Varid{→}\;\Varid{lookup}\;\Varid{n}\;\Varid{e}\;\Varid{≡}\;\MyConid{just}\;\Varid{v}\;\Varid{→}\;{}\<[54]%
\>[54]{}(\MyConid{VAR}\;\Varid{n}\;\MyConid{;}\;\Varid{c},\Varid{e},\Varid{s})\;{}\<[92]%
\>[92]{}\xrightarrow[\Varid{CES}]{}\;(\Varid{c},{}\<[105]%
\>[105]{}\Varid{e},\MyConid{val}\;\Varid{v}\;\Varid{∷}\;\Varid{s}){}\<[E]%
\\
\>[3]{}\MyConid{CLOS}\;{}\<[13]%
\>[13]{}\mathbin{:}\;\Varid{∀}\;\{\mskip1.5mu \Varid{c'}\;\Varid{c}\;\Varid{e}\;\Varid{s}\mskip1.5mu\}\;\Varid{→}\;{}\<[36]%
\>[36]{}(\MyConid{CLOS}\;\Varid{c'}\;\MyConid{;}\;\Varid{c},\Varid{e},\Varid{s})\;{}\<[92]%
\>[92]{}\xrightarrow[\Varid{CES}]{}\;(\Varid{c},{}\<[103]%
\>[103]{}\Varid{e},\MyConid{val}\;(\MyConid{clos}\;(\Varid{c'},\Varid{e}))\;\Varid{∷}\;\Varid{s}){}\<[E]%
\\
\>[3]{}\MyConid{APPL}\;{}\<[13]%
\>[13]{}\mathbin{:}\;\Varid{∀}\;\{\mskip1.5mu \Varid{c}\;\Varid{e}\;\Varid{v}\;\Varid{c'}\;\Varid{e'}\;\Varid{s}\mskip1.5mu\}\;\Varid{→}\;{}\<[36]%
\>[36]{}(\MyConid{APPL}\;\MyConid{;}\;\Varid{c},\Varid{e},\MyConid{val}\;\Varid{v}\;\Varid{∷}\;\MyConid{val}\;(\MyConid{clos}\;(\Varid{c'},\Varid{e'}))\;\Varid{∷}\;\Varid{s})\;{}\<[92]%
\>[92]{}\xrightarrow[\Varid{CES}]{}\;(\Varid{c'},\Varid{v}\;\Varid{∷}\;\Varid{e'},\MyConid{cont}\;(\Varid{c},\Varid{e})\;\Varid{∷}\;\Varid{s}){}\<[E]%
\\
\>[3]{}\MyConid{RET}\;{}\<[13]%
\>[13]{}\mathbin{:}\;\Varid{∀}\;\{\mskip1.5mu \Varid{e}\;\Varid{v}\;\Varid{c}\;\Varid{e'}\;\Varid{s}\mskip1.5mu\}\;\Varid{→}\;{}\<[36]%
\>[36]{}(\MyConid{RET},\Varid{e},\MyConid{val}\;\Varid{v}\;\Varid{∷}\;\MyConid{cont}\;(\Varid{c},\Varid{e'})\;\Varid{∷}\;\Varid{s})\;{}\<[92]%
\>[92]{}\xrightarrow[\Varid{CES}]{}\;(\Varid{c},\Varid{e'},\MyConid{val}\;\Varid{v}\;\Varid{∷}\;\Varid{s}){}\<[E]%
\\
\>[3]{}\MyConid{LIT}\;{}\<[13]%
\>[13]{}\mathbin{:}\;\Varid{∀}\;\{\mskip1.5mu \Varid{n}\;\Varid{c}\;\Varid{e}\;\Varid{s}\mskip1.5mu\}\;\Varid{→}\;{}\<[36]%
\>[36]{}(\MyConid{LIT}\;\Varid{n}\;\MyConid{;}\;\Varid{c},\Varid{e},\Varid{s})\;{}\<[92]%
\>[92]{}\xrightarrow[\Varid{CES}]{}\;(\Varid{c},\Varid{e},\MyConid{val}\;(\MyConid{nat}\;\Varid{n})\;\Varid{∷}\;\Varid{s}){}\<[E]%
\\
\>[3]{}\MyConid{OP}\;{}\<[13]%
\>[13]{}\mathbin{:}\;\Varid{∀}\;\{\mskip1.5mu \Varid{f}\;\Varid{c}\;\Varid{e}\;\Varid{n₁}\;\Varid{n₂}\;\Varid{s}\mskip1.5mu\}\;\Varid{→}\;{}\<[36]%
\>[36]{}(\MyConid{OP}\;\Varid{f}\;\MyConid{;}\;\Varid{c},\Varid{e},\MyConid{val}\;(\MyConid{nat}\;\Varid{n₁})\;\Varid{∷}\;\MyConid{val}\;(\MyConid{nat}\;\Varid{n₂})\;\Varid{∷}\;\Varid{s})\;{}\<[92]%
\>[92]{}\xrightarrow[\Varid{CES}]{}\;(\Varid{c},\Varid{e},\MyConid{val}\;(\MyConid{nat}\;(\Varid{f}\;\Varid{n₁}\;\Varid{n₂}))\;\Varid{∷}\;\Varid{s}){}\<[E]%
\\
\>[3]{}\MyConid{COND-0}\;{}\<[13]%
\>[13]{}\mathbin{:}\;\Varid{∀}\;\{\mskip1.5mu \Varid{c}\;\Varid{c'}\;\Varid{e}\;\Varid{s}\mskip1.5mu\}\;\Varid{→}\;{}\<[36]%
\>[36]{}(\MyConid{COND}\;\Varid{c}\;\Varid{c'},\Varid{e},\MyConid{val}\;(\MyConid{nat}\;\Varid{0})\;\Varid{∷}\;\Varid{s})\;{}\<[92]%
\>[92]{}\xrightarrow[\Varid{CES}]{}\;(\Varid{c},\Varid{e},\Varid{s}){}\<[E]%
\\
\>[3]{}\MyConid{COND-1+n}\;{}\<[13]%
\>[13]{}\mathbin{:}\;\Varid{∀}\;\{\mskip1.5mu \Varid{c}\;\Varid{c'}\;\Varid{e}\;\Varid{n}\;\Varid{s}\mskip1.5mu\}\;\Varid{→}\;{}\<[36]%
\>[36]{}(\MyConid{COND}\;\Varid{c}\;\Varid{c'},\Varid{e},\MyConid{val}\;(\MyConid{nat}\;(1+\;\Varid{n}))\;\Varid{∷}\;\Varid{s})\;{}\<[92]%
\>[92]{}\xrightarrow[\Varid{CES}]{}\;(\Varid{c'},\Varid{e},\Varid{s}){}\<[E]%
\ColumnHook
\end{hscode}\resethooks
\end{varwidth}
\caption{The definition of the transition relation of the CES machine.}
\label{figure:CES-step}
\iffullversion
\end{sidewaysfigure}
\else
\end{figure*}
\fi

Fig.~\ref{figure:CES-step} shows the definition of the transition
relation for configurations of the CES machine. The Agda syntax may
require some further explanation: The instructions' constructor names
are overloaded to also act as constructors for the relation; their
usage will be disambiguated by context. We use \emph{implicit
arguments}, written in curly braces, for arguments that can automatically
be inferred and do not need to be spelled out explicitly. The type of
propositional equality is written {\textsmaller[.5]{\ensuremath{\Varid{\char95 ≡\char95 }}}}.

The stack discipline becomes apparent in the definition of the
transition relation.  When e.g. {\textsmaller[.5]{\ensuremath{\MyConid{VAR}}}} is executed, the CES machine looks up
the value of the variable in the environment and pushes it on the
stack.  A somewhat subtle part of the relation is the interplay
between the {\textsmaller[.5]{\ensuremath{\MyConid{APPL}}}} instruction and the {\textsmaller[.5]{\ensuremath{\MyConid{RET}}}} instruction. When
performing an application, two values are required on the stack, one of
which has to be a closure.  The machine enters the closure, adding the
value to the environment, and pushes a return continuation on the
stack.  Looking at the {\textsmaller[.5]{\ensuremath{\Varid{compile}}}} function, we see that the code inside
a closure will be terminated by a {\textsmaller[.5]{\ensuremath{\MyConid{RET}}}} instruction, so once the
machine has finished executing the closure (and thus produced a value
on the stack), that value is returned to the continuation.

\begin{example} \label{example:CES}
We trace the execution of {\textsmaller[.5]{\ensuremath{\Varid{codeExample}}}} defined above, which exemplifies
how returning from an application works. Here we write {\textsmaller[.5]{\ensuremath{\Varid{a}\;\xrightarrow[\Varid{CES}]{}\Varid{⟨}\;\Varid{x}\;\Varid{⟩}\;\Varid{b}}}}
meaning that the machine uses rule {\textsmaller[.5]{\ensuremath{\Varid{x}}}} to transition from {\textsmaller[.5]{\ensuremath{\Varid{a}}}} to {\textsmaller[.5]{\ensuremath{\Varid{b}}}}.
\begin{hscode}\SaveRestoreHook
\column{B}{@{}>{\hspre}l<{\hspost}@{}}%
\column{6}{@{}>{\hspre}l<{\hspost}@{}}%
\column{11}{@{}>{\hspre}l<{\hspost}@{}}%
\column{16}{@{}>{\hspre}l<{\hspost}@{}}%
\column{36}{@{}>{\hspre}l<{\hspost}@{}}%
\column{41}{@{}>{\hspre}l<{\hspost}@{}}%
\column{E}{@{}>{\hspre}l<{\hspost}@{}}%
\>[B]{}\Keyword{let}\;{}\<[6]%
\>[6]{}\Varid{c₁}\;{}\<[11]%
\>[11]{}\mathrel{=}\;\MyConid{VAR}\;\Varid{0}\;\MyConid{;}\;\MyConid{RET}{}\<[E]%
\\
\>[6]{}\Varid{c₂}\;{}\<[11]%
\>[11]{}\mathrel{=}\;\MyConid{CLOS}\;(\MyConid{VAR}\;\Varid{1}\;\MyConid{;}\;\MyConid{RET})\;\MyConid{;}\;\MyConid{RET}{}\<[E]%
\\
\>[6]{}\Varid{cl₁}\;{}\<[11]%
\>[11]{}\mathrel{=}\;\MyConid{val}\;(\MyConid{clos}\;(\Varid{c₁},[\mskip1.5mu \mskip1.5mu]));{}\<[36]%
\>[36]{}\Varid{cl₂}\;{}\<[41]%
\>[41]{}\mathrel{=}\;\MyConid{val}\;(\MyConid{clos}\;(\Varid{c₂},[\mskip1.5mu \mskip1.5mu])){}\<[E]%
\\
\>[B]{}\Keyword{in}\;(\MyConid{CLOS}\;\Varid{c₁}\;\MyConid{;}\;\MyConid{CLOS}\;\Varid{c₂}\;\MyConid{;}\;\MyConid{APPL}\;\MyConid{;}\;\MyConid{END},[\mskip1.5mu \mskip1.5mu],[\mskip1.5mu \mskip1.5mu]){}\<[E]%
\\
\>[B]{}\xrightarrow[\Varid{CES}]{}\Varid{⟨}\;\MyConid{CLOS}\;\Varid{⟩}\;{}\<[16]%
\>[16]{}(\MyConid{CLOS}\;\Varid{c₂}\;\MyConid{;}\;\MyConid{APPL}\;\MyConid{;}\;\MyConid{END},[\mskip1.5mu \mskip1.5mu],[\mskip1.5mu \Varid{cl₁}\mskip1.5mu]){}\<[E]%
\\
\>[B]{}\xrightarrow[\Varid{CES}]{}\Varid{⟨}\;\MyConid{CLOS}\;\Varid{⟩}\;{}\<[16]%
\>[16]{}(\MyConid{APPL}\;\MyConid{;}\;\MyConid{END},[\mskip1.5mu \mskip1.5mu],[\mskip1.5mu \Varid{cl₂},\Varid{cl₁}\mskip1.5mu]){}\<[E]%
\\
\>[B]{}\xrightarrow[\Varid{CES}]{}\Varid{⟨}\;\MyConid{APPL}\;\Varid{⟩}\;{}\<[16]%
\>[16]{}(\MyConid{VAR}\;\Varid{0}\;\MyConid{;}\;\MyConid{RET},[\mskip1.5mu \Varid{cl₂}\mskip1.5mu],[\mskip1.5mu \MyConid{cont}\;(\MyConid{END},[\mskip1.5mu \mskip1.5mu])\mskip1.5mu]){}\<[E]%
\\
\>[B]{}\xrightarrow[\Varid{CES}]{}\Varid{⟨}\;\MyConid{VAR}\;\Varid{refl}\;\Varid{⟩}\;{}\<[16]%
\>[16]{}(\MyConid{RET},[\mskip1.5mu \Varid{cl₂}\mskip1.5mu],[\mskip1.5mu \Varid{cl₂},\MyConid{cont}\;(\MyConid{END},[\mskip1.5mu \mskip1.5mu])\mskip1.5mu]){}\<[E]%
\\
\>[B]{}\xrightarrow[\Varid{CES}]{}\Varid{⟨}\;\MyConid{RET}\;\Varid{⟩}\;{}\<[16]%
\>[16]{}(\MyConid{END},[\mskip1.5mu \mskip1.5mu],[\mskip1.5mu \Varid{cl₂}\mskip1.5mu]){}\<[E]%
\ColumnHook
\end{hscode}\resethooks
The final result is therefore the second closure, {\textsmaller[.5]{\ensuremath{\Varid{cl₂}}}}.
\end{example}

\begin{lemma}[{\textsmaller[.5]{\ensuremath{\Varid{determinism}_{\Varid{CES}}}}}] {\textsmaller[.5]{\ensuremath{\xrightarrow[\Varid{CES}]{}}}} is deterministic.
\iffullversion
In the formalisation this means that we can construct the following term:
\begin{hscode}\SaveRestoreHook
\column{B}{@{}>{\hspre}l<{\hspost}@{}}%
\column{E}{@{}>{\hspre}l<{\hspost}@{}}%
\>[B]{}\Varid{determinism}_{\Varid{CES}}\;\mathbin{:}\;{\dummy\xrightarrow[\Varid{CES}]{}\dummy}\;\Varid{is-deterministic}{}\<[E]%
\ColumnHook
\end{hscode}\resethooks
where the type {\textsmaller[.5]{\ensuremath{\Varid{\char95 is-deterministic}}}} is defined as follows:
\begin{hscode}\SaveRestoreHook
\column{B}{@{}>{\hspre}l<{\hspost}@{}}%
\column{3}{@{}>{\hspre}l<{\hspost}@{}}%
\column{E}{@{}>{\hspre}l<{\hspost}@{}}%
\>[B]{}\Varid{\char95 is-deterministic}\;\mathbin{:}\;\{\mskip1.5mu \Conid{A}\;\Conid{B}\;\mathbin{:}\;\star\mskip1.5mu\}\;\Varid{→}\;\Conid{Rel}\;\Conid{A}\;\Conid{B}\;\Varid{→}\;\star{}\<[E]%
\\
\>[B]{}\Conid{R}\;\Varid{is-deterministic}\;\mathrel{=}\;{}\<[E]%
\\
\>[B]{}\hsindent{3}{}\<[3]%
\>[3]{}\Varid{∀}\;\{\mskip1.5mu \Varid{a}\;\Varid{b}\;\Varid{b'}\mskip1.5mu\}\;\Varid{→}\;\Conid{R}\;\Varid{a}\;\Varid{b}\;\Varid{→}\;\Conid{R}\;\Varid{a}\;\Varid{b'}\;\Varid{→}\;\Varid{b}\;\Varid{≡}\;\Varid{b'}{}\<[E]%
\ColumnHook
\end{hscode}\resethooks
\fi
\end{lemma}
This is a key property that is useful when constructing a compiler
implementation of the CES machine. Note that we write the name
of the definition containing this proof in parentheses.

The CES machine \emph{terminates with a value {\textsmaller[.5]{\ensuremath{\Varid{v}}}}}, written {\textsmaller[.5]{\ensuremath{\Varid{cfg}\;\downarrow_{\Varid{CES}}\;\Varid{v}}}} if it, through the reflexive transitive closure of {\textsmaller[.5]{\ensuremath{\xrightarrow[\Varid{CES}]{}}}},
reaches the end of its code fragment with an empty environment, and
{\textsmaller[.5]{\ensuremath{\Varid{v}}}} as its sole stack element.
\iffullversion
\begin{hscode}\SaveRestoreHook
\column{B}{@{}>{\hspre}l<{\hspost}@{}}%
\column{E}{@{}>{\hspre}l<{\hspost}@{}}%
\>[B]{}{\dummy\downarrow_{\Varid{CES}}\dummy}\;\mathbin{:}\;\Conid{Config}\;\Varid{→}\;\Conid{Value}\;\Varid{→}\;\star{}\<[E]%
\\
\>[B]{}\Varid{cfg}\;\downarrow_{\Varid{CES}}\;\Varid{v}\;\mathrel{=}\;\Varid{cfg}\;\xrightarrow[\Varid{CES}]{}^*\;(\MyConid{END},[\mskip1.5mu \mskip1.5mu],\MyConid{val}\;\Varid{v}\;\Varid{∷}\;[\mskip1.5mu \mskip1.5mu]){}\<[E]%
\ColumnHook
\end{hscode}\resethooks
where the reflexive transitive closure of a relation is defined as follows:
\begin{hscode}\SaveRestoreHook
\column{B}{@{}>{\hspre}l<{\hspost}@{}}%
\column{3}{@{}>{\hspre}l<{\hspost}@{}}%
\column{8}{@{}>{\hspre}l<{\hspost}@{}}%
\column{E}{@{}>{\hspre}l<{\hspost}@{}}%
\>[B]{}\Keyword{data}\;\Varid{\char95 *}\;\{\mskip1.5mu \Conid{A}\;\mathbin{:}\;\star\mskip1.5mu\}\;(\Conid{R}\;\mathbin{:}\;\Conid{Rel}\;\Conid{A}\;\Conid{A})\;(\Varid{a}\;\mathbin{:}\;\Conid{A})\;\mathbin{:}\;\Conid{A}\;\Varid{→}\;\star\;\Keyword{where}{}\<[E]%
\\
\>[B]{}\hsindent{3}{}\<[3]%
\>[3]{}[\mskip1.5mu \mskip1.5mu]\;{}\<[8]%
\>[8]{}\mathbin{:}\;(\Conid{R}\;\Varid{*})\;\Varid{a}\;\Varid{a}{}\<[E]%
\\
\>[B]{}\hsindent{3}{}\<[3]%
\>[3]{}\Varid{\char95 ∷\char95 }\;{}\<[8]%
\>[8]{}\mathbin{:}\;\{\mskip1.5mu \Varid{b}\;\Varid{c}\;\mathbin{:}\;\Conid{A}\mskip1.5mu\}\;\Varid{→}\;\Conid{R}\;\Varid{a}\;\Varid{b}\;\Varid{→}\;(\Conid{R}\;\Varid{*})\;\Varid{b}\;\Varid{c}\;\Varid{→}\;(\Conid{R}\;\Varid{*})\;\Varid{a}\;\Varid{c}{}\<[E]%
\ColumnHook
\end{hscode}\resethooks
\fi
It \emph{terminates}, written {\textsmaller[.5]{\ensuremath{\Varid{cfg}\;\downarrow_{\Varid{CES}}}}} if there exists a value {\textsmaller[.5]{\ensuremath{\Varid{v}}}} such
that it terminates with the value {\textsmaller[.5]{\ensuremath{\Varid{v}}}}.
\iffullversion
The Agda syntax for the existential quantifier normally
written as $\exists x. P(x)$ is {\textsmaller[.5]{\ensuremath{\Varid{∃}\;\Varid{λ}\;\Varid{x}\;\Varid{→}\;\Conid{P}\;\Varid{x}}}}. Using this syntax, the
definition of termination with value {\textsmaller[.5]{\ensuremath{\Varid{v}}}} is:
\begin{hscode}\SaveRestoreHook
\column{B}{@{}>{\hspre}l<{\hspost}@{}}%
\column{E}{@{}>{\hspre}l<{\hspost}@{}}%
\>[B]{}{\dummy\downarrow_{\Varid{CES}}}\;\mathbin{:}\;\Conid{Config}\;\Varid{→}\;\star{}\<[E]%
\\
\>[B]{}\Varid{cfg}\;\downarrow_{\Varid{CES}}\;\mathrel{=}\;\Varid{∃}\;\Varid{λ}\;\Varid{v}\;\Varid{→}\;\Varid{cfg}\;\downarrow_{\Varid{CES}}\;\Varid{v}{}\<[E]%
\ColumnHook
\end{hscode}\resethooks
\fi
It \emph{diverges}, written {\textsmaller[.5]{\ensuremath{\Varid{cfg}\;\uparrow_{\Varid{CES}}}}} if it is possible to take
another step from any configuration reachable from the reflexive transitive
closure of {\textsmaller[.5]{\ensuremath{\xrightarrow[\Varid{CES}]{}}}}.
\iffullversion
\begin{hscode}\SaveRestoreHook
\column{B}{@{}>{\hspre}l<{\hspost}@{}}%
\column{E}{@{}>{\hspre}l<{\hspost}@{}}%
\>[B]{}{\dummy\uparrow_{\Varid{CES}}}\;\mathbin{:}\;\Conid{Config}\;\Varid{→}\;\star{}\<[E]%
\\
\>[B]{}{\dummy\uparrow_{\Varid{CES}}}\;\mathrel{=}\;\Varid{↑}\;{\dummy\xrightarrow[\Varid{CES}]{}\dummy}{}\<[E]%
\ColumnHook
\end{hscode}\resethooks
where
\begin{hscode}\SaveRestoreHook
\column{B}{@{}>{\hspre}l<{\hspost}@{}}%
\column{E}{@{}>{\hspre}l<{\hspost}@{}}%
\>[B]{}\Varid{↑}\;\mathbin{:}\;\{\mskip1.5mu \Conid{A}\;\mathbin{:}\;\star\mskip1.5mu\}\;(\Conid{R}\;\mathbin{:}\;\Conid{Rel}\;\Conid{A}\;\Conid{A})\;\Varid{→}\;\Conid{A}\;\Varid{→}\;\star{}\<[E]%
\\
\>[B]{}\Varid{↑}\;\Conid{R}\;\Varid{a}\;\mathrel{=}\;\Varid{∀}\;\Varid{b}\;\Varid{→}\;(\Conid{R}\;\Varid{*})\;\Varid{a}\;\Varid{b}\;\Varid{→}\;\Varid{∃}\;\Varid{λ}\;\Varid{c}\;\Varid{→}\;\Conid{R}\;\Varid{b}\;\Varid{c}{}\<[E]%
\ColumnHook
\end{hscode}\resethooks
\fi

\section{CESH: A heap machine} \label{section:CESH}
In a compiler implementation of the CES machine targeting a low-level
language, closures have to be dynamically allocated in a
heap. However, the CES machine does not make this dynamic allocation
explicit. In this section, we try to make it explicit by
constructing the CESH machine, which is a CES machine with an extra heap
component in its configuration.

While heaps are not strictly necessary for a \emph{presentation} of
the CES machine, they are of great importance to us. The distributed
machine that we will later define needs heaps for persistent
storage of data, and the CESH machine forms an intermediate step
between that and the CES machine.
\iffullversion
Another thing that can be done with
heaps is to implement general recursion, without using fix-point
combinators, by constructing circular closures.
\fi

A CESH configuration is defined as {\textsmaller[.5]{\ensuremath{\Conid{Config}\;\mathrel{=}\;\Conid{Code}\;\Varid{×}\;\Conid{Env}\;\Varid{×}\;\Conid{Stack}\;\Varid{×}\;\Conid{Heap}\;\Conid{Closure}}}}, where {\textsmaller[.5]{\ensuremath{\Conid{Heap}}}} is a type constructor for heaps
parameterised by the type of its content.
\ifnotfullversion
The only difference in the definition of the configuration constituents,
compared to the CES machine, is that a closure value (the {\textsmaller[.5]{\ensuremath{\MyConid{clos}}}} constructor of the {\textsmaller[.5]{\ensuremath{\Conid{Value}}}} type)
does not contain an actual closure, but just a pointer ({\textsmaller[.5]{\ensuremath{\Conid{Ptr}}}}).
\else
Closures, values and environments are again mutually defined. Now a
closure value is simply represented by a pointer:
\begin{hscode}\SaveRestoreHook
\column{B}{@{}>{\hspre}l<{\hspost}@{}}%
\column{3}{@{}>{\hspre}l<{\hspost}@{}}%
\column{5}{@{}>{\hspre}l<{\hspost}@{}}%
\column{11}{@{}>{\hspre}l<{\hspost}@{}}%
\column{22}{@{}>{\hspre}l<{\hspost}@{}}%
\column{E}{@{}>{\hspre}l<{\hspost}@{}}%
\>[B]{}\Conid{ClosPtr}\;\mathrel{=}\;\Conid{Ptr}{}\<[E]%
\\
\>[B]{}\Keyword{mutual}{}\<[E]%
\\
\>[B]{}\hsindent{3}{}\<[3]%
\>[3]{}\Conid{Closure}\;\mathrel{=}\;\Conid{Code}\;\Varid{×}\;\Conid{Env}{}\<[E]%
\\
\>[B]{}\hsindent{3}{}\<[3]%
\>[3]{}\Keyword{data}\;\Conid{Value}\;\mathbin{:}\;\star\;\Keyword{where}{}\<[E]%
\\
\>[3]{}\hsindent{2}{}\<[5]%
\>[5]{}\MyConid{nat}\;{}\<[11]%
\>[11]{}\mathbin{:}\;\Conid{ℕ}\;{}\<[22]%
\>[22]{}\Varid{→}\;\Conid{Value}{}\<[E]%
\\
\>[3]{}\hsindent{2}{}\<[5]%
\>[5]{}\MyConid{clos}\;{}\<[11]%
\>[11]{}\mathbin{:}\;\Conid{ClosPtr}\;{}\<[22]%
\>[22]{}\Varid{→}\;\Conid{Value}{}\<[E]%
\\
\>[B]{}\hsindent{3}{}\<[3]%
\>[3]{}\Conid{Env}\;\mathrel{=}\;\Conid{List}\;\Conid{Value}{}\<[E]%
\ColumnHook
\end{hscode}\resethooks
\fi
We leave the stack as in the CES
\ifnotfullversion
machine.
\else
machine (even though we could, in principle, change the continuations,
currently represented by closures, to pointers as well -- we do not do
this since it is not necessary for our purposes).
\begin{hscode}\SaveRestoreHook
\column{B}{@{}>{\hspre}l<{\hspost}@{}}%
\column{3}{@{}>{\hspre}l<{\hspost}@{}}%
\column{9}{@{}>{\hspre}l<{\hspost}@{}}%
\column{20}{@{}>{\hspre}l<{\hspost}@{}}%
\column{E}{@{}>{\hspre}l<{\hspost}@{}}%
\>[B]{}\Keyword{data}\;\Conid{StackElem}\;\mathbin{:}\;\star\;\Keyword{where}{}\<[E]%
\\
\>[B]{}\hsindent{3}{}\<[3]%
\>[3]{}\MyConid{val}\;{}\<[9]%
\>[9]{}\mathbin{:}\;\Conid{Value}\;{}\<[20]%
\>[20]{}\Varid{→}\;\Conid{StackElem}{}\<[E]%
\\
\>[B]{}\hsindent{3}{}\<[3]%
\>[3]{}\MyConid{cont}\;{}\<[9]%
\>[9]{}\mathbin{:}\;\Conid{Closure}\;{}\<[20]%
\>[20]{}\Varid{→}\;\Conid{StackElem}{}\<[E]%
\\[\blanklineskip]%
\>[B]{}\Conid{Stack}\;\mathrel{=}\;\Conid{List}\;\Conid{StackElem}{}\<[E]%
\ColumnHook
\end{hscode}\resethooks
\fi

\iffullversion
\begin{sidewaysfigure}
\else
\begin{figure*}[!t]
\fi
\centering
\begin{varwidth}[t]{\textwidth}
\iffullversion
\begin{hscode}\SaveRestoreHook
\column{B}{@{}>{\hspre}l<{\hspost}@{}}%
\column{E}{@{}>{\hspre}l<{\hspost}@{}}%
\>[B]{}\Keyword{data}\;{\dummy\xrightarrow[\Varid{CESH}]{}\dummy}\;\mathbin{:}\;\Conid{Rel}\;\Conid{Config}\;\Conid{Config}\;\Keyword{where}{}\<[E]%
\ColumnHook
\end{hscode}\resethooks
\removecodespace
\fi
\savecolumns
\rightaligncolumn{40}
\indentcolumn{3}
\begin{hscode}\SaveRestoreHook
\column{B}{@{}>{\hspre}l<{\hspost}@{}}%
\column{3}{@{}>{\hspre}l<{\hspost}@{}}%
\column{13}{@{}>{\hspre}l<{\hspost}@{}}%
\column{40}{@{}>{\hspre}l<{\hspost}@{}}%
\column{100}{@{}>{\hspre}l<{\hspost}@{}}%
\column{E}{@{}>{\hspre}l<{\hspost}@{}}%
\>[3]{}\MyConid{CLOS}\;{}\<[13]%
\>[13]{}\mathbin{:}\;\Varid{∀}\;\{\mskip1.5mu \Varid{c'}\;\Varid{c}\;\Varid{e}\;\Varid{s}\;\Varid{h}\mskip1.5mu\}\;\Varid{→}\;\Keyword{let}\;(\Varid{h'},\Varid{ptr}_{\Varid{cl}})\;\mathrel{=}\;\Varid{h}\;\Varid{▸}\;(\Varid{c'},\Varid{e})\;\Keyword{in}{}\<[E]%
\\
\>[13]{}\hsindent{27}{}\<[40]%
\>[40]{}(\MyConid{CLOS}\;\Varid{c'}\;\MyConid{;}\;\Varid{c},\Varid{e},\Varid{s},\Varid{h})\;{}\<[100]%
\>[100]{}\xrightarrow[\Varid{CESH}]{}\;(\Varid{c},\Varid{e},\MyConid{val}\;(\MyConid{clos}\;\Varid{ptr}_{\Varid{cl}})\;\Varid{∷}\;\Varid{s},\Varid{h'}){}\<[E]%
\\
\>[3]{}\MyConid{APPL}\;{}\<[13]%
\>[13]{}\mathbin{:}\;\Varid{∀}\;\{\mskip1.5mu \Varid{c}\;\Varid{e}\;\Varid{v}\;\Varid{ptr}_{\Varid{cl}}\;\Varid{c'}\;\Varid{e'}\;\Varid{s}\;\Varid{h}\mskip1.5mu\}\;\Varid{→}\;\Varid{h}\;\mathbin{!}\;\Varid{ptr}_{\Varid{cl}}\;\Varid{≡}\;\MyConid{just}\;(\Varid{c'},\Varid{e'})\;\Varid{→}\;{}\<[E]%
\\
\>[13]{}\hsindent{27}{}\<[40]%
\>[40]{}(\MyConid{APPL}\;\MyConid{;}\;\Varid{c},\Varid{e},\MyConid{val}\;\Varid{v}\;\Varid{∷}\;\MyConid{val}\;(\MyConid{clos}\;\Varid{ptr}_{\Varid{cl}})\;\Varid{∷}\;\Varid{s},\Varid{h})\;{}\<[100]%
\>[100]{}\xrightarrow[\Varid{CESH}]{}\;(\Varid{c'},\Varid{v}\;\Varid{∷}\;\Varid{e'},\MyConid{cont}\;(\Varid{c},\Varid{e})\;\Varid{∷}\;\Varid{s},\Varid{h}){}\<[E]%
\ColumnHook
\end{hscode}\resethooks
\iffullversion
\removecodespace
\restorecolumns
\indentcolumn{3}
\rightaligncolumn{40}
\begin{hscode}\SaveRestoreHook
\column{B}{@{}>{\hspre}l<{\hspost}@{}}%
\column{3}{@{}>{\hspre}l<{\hspost}@{}}%
\column{13}{@{}>{\hspre}l<{\hspost}@{}}%
\column{40}{@{}>{\hspre}l<{\hspost}@{}}%
\column{100}{@{}>{\hspre}l<{\hspost}@{}}%
\column{E}{@{}>{\hspre}l<{\hspost}@{}}%
\>[3]{}\MyConid{VAR}\;{}\<[13]%
\>[13]{}\mathbin{:}\;\Varid{∀}\;\{\mskip1.5mu \Varid{n}\;\Varid{c}\;\Varid{e}\;\Varid{s}\;\Varid{h}\;\Varid{v}\mskip1.5mu\}\;\Varid{→}\;\Varid{lookup}\;\Varid{n}\;\Varid{e}\;\Varid{≡}\;\MyConid{just}\;\Varid{v}\;\Varid{→}\;{}\<[E]%
\\
\>[13]{}\hsindent{27}{}\<[40]%
\>[40]{}(\MyConid{VAR}\;\Varid{n}\;\MyConid{;}\;\Varid{c},\Varid{e},\Varid{s},\Varid{h})\;{}\<[100]%
\>[100]{}\xrightarrow[\Varid{CESH}]{}\;(\Varid{c},\Varid{e},\MyConid{val}\;\Varid{v}\;\Varid{∷}\;\Varid{s},\Varid{h}){}\<[E]%
\\
\>[3]{}\MyConid{RET}\;{}\<[13]%
\>[13]{}\mathbin{:}\;\Varid{∀}\;\{\mskip1.5mu \Varid{e}\;\Varid{v}\;\Varid{c}\;\Varid{e'}\;\Varid{s}\;\Varid{h}\mskip1.5mu\}\;\Varid{→}\;{}\<[40]%
\>[40]{}(\MyConid{RET},\Varid{e},\MyConid{val}\;\Varid{v}\;\Varid{∷}\;\MyConid{cont}\;(\Varid{c},\Varid{e'})\;\Varid{∷}\;\Varid{s},\Varid{h})\;{}\<[100]%
\>[100]{}\xrightarrow[\Varid{CESH}]{}\;(\Varid{c},\Varid{e'},\MyConid{val}\;\Varid{v}\;\Varid{∷}\;\Varid{s},\Varid{h}){}\<[E]%
\\
\>[3]{}\MyConid{LIT}\;{}\<[13]%
\>[13]{}\mathbin{:}\;\Varid{∀}\;\{\mskip1.5mu \Varid{l}\;\Varid{c}\;\Varid{e}\;\Varid{s}\;\Varid{h}\mskip1.5mu\}\;\Varid{→}\;{}\<[40]%
\>[40]{}(\MyConid{LIT}\;\Varid{l}\;\MyConid{;}\;\Varid{c},\Varid{e},\Varid{s},\Varid{h})\;{}\<[100]%
\>[100]{}\xrightarrow[\Varid{CESH}]{}\;(\Varid{c},\Varid{e},\MyConid{val}\;(\MyConid{nat}\;\Varid{l})\;\Varid{∷}\;\Varid{s},\Varid{h}){}\<[E]%
\\
\>[3]{}\MyConid{OP}\;{}\<[13]%
\>[13]{}\mathbin{:}\;\Varid{∀}\;\{\mskip1.5mu \Varid{f}\;\Varid{c}\;\Varid{e}\;\Varid{l₁}\;\Varid{l₂}\;\Varid{s}\;\Varid{h}\mskip1.5mu\}\;\Varid{→}\;{}\<[40]%
\>[40]{}(\MyConid{OP}\;\Varid{f}\;\MyConid{;}\;\Varid{c},\Varid{e},\MyConid{val}\;(\MyConid{nat}\;\Varid{l₁})\;\Varid{∷}\;\MyConid{val}\;(\MyConid{nat}\;\Varid{l₂})\;\Varid{∷}\;\Varid{s},\Varid{h})\;{}\<[100]%
\>[100]{}\xrightarrow[\Varid{CESH}]{}\;(\Varid{c},\Varid{e},\MyConid{val}\;(\MyConid{nat}\;(\Varid{f}\;\Varid{l₁}\;\Varid{l₂}))\;\Varid{∷}\;\Varid{s},\Varid{h}){}\<[E]%
\\
\>[3]{}\MyConid{COND-0}\;{}\<[13]%
\>[13]{}\mathbin{:}\;\Varid{∀}\;\{\mskip1.5mu \Varid{c}\;\Varid{c'}\;\Varid{e}\;\Varid{s}\;\Varid{h}\mskip1.5mu\}\;\Varid{→}\;{}\<[40]%
\>[40]{}(\MyConid{COND}\;\Varid{c}\;\Varid{c'},\Varid{e},\MyConid{val}\;(\MyConid{nat}\;\Varid{0})\;\Varid{∷}\;\Varid{s},\Varid{h})\;{}\<[100]%
\>[100]{}\xrightarrow[\Varid{CESH}]{}\;(\Varid{c},\Varid{e},\Varid{s},\Varid{h}){}\<[E]%
\\
\>[3]{}\MyConid{COND-1+n}\;{}\<[13]%
\>[13]{}\mathbin{:}\;\Varid{∀}\;\{\mskip1.5mu \Varid{c}\;\Varid{c'}\;\Varid{e}\;\Varid{n}\;\Varid{s}\;\Varid{h}\mskip1.5mu\}\;\Varid{→}\;{}\<[40]%
\>[40]{}(\MyConid{COND}\;\Varid{c}\;\Varid{c'},\Varid{e},\MyConid{val}\;(\MyConid{nat}\;(1+\;\Varid{n}))\;\Varid{∷}\;\Varid{s},\Varid{h})\;{}\<[100]%
\>[100]{}\xrightarrow[\Varid{CESH}]{}\;(\Varid{c'},\Varid{e},\Varid{s},\Varid{h}){}\<[E]%
\ColumnHook
\end{hscode}\resethooks
\fi
\end{varwidth}
\iffullversion
\caption{The definition of the transition relation of the CESH machine.}
\fi
\label{figure:CESH-step}
\iffullversion
\end{sidewaysfigure}
\else
\end{figure*}
\fi

Fig.~\ref{figure:CESH-step} shows the definition of the transition
relation of the CESH machine.
\iffullversion
It is largely the same as that of the CES
machine, but with the added heap component.
\fi
The difference appears in the {\textsmaller[.5]{\ensuremath{\MyConid{CLOS}}}} and {\textsmaller[.5]{\ensuremath{\MyConid{APPL}}}} instructions. To
build a closure, the machine allocates it in the heap using the {\textsmaller[.5]{\ensuremath{\Varid{\char95 ▸\char95 }}}}
function, which, given a heap and an element, gives back an
updated heap and a pointer to the element.  When performing an
application, the machine has a \emph{pointer} to a closure, so it
looks it up in the heap using the {\textsmaller[.5]{\ensuremath{\Varid{\char95 !\char95 }}}} function, which, given a heap
and a pointer, gives back the element that the pointer points to (if
it exists).
\begin{lemma}[{\textsmaller[.5]{\ensuremath{\Varid{determinism}_{\Varid{CESH}}}}}]
\iffullversion
{\textsmaller[.5]{\ensuremath{\xrightarrow[\Varid{CESH}]{}}}} is deterministic, i.e. we can define the following term:
\begin{hscode}\SaveRestoreHook
\column{B}{@{}>{\hspre}l<{\hspost}@{}}%
\column{E}{@{}>{\hspre}l<{\hspost}@{}}%
\>[B]{}\Varid{determinism}_{\Varid{CESH}}\;\mathbin{:}\;{\dummy\xrightarrow[\Varid{CESH}]{}\dummy}\;\Varid{is-deterministic}{}\<[E]%
\ColumnHook
\end{hscode}\resethooks
\fi
\end{lemma}

We define what it means for a CESH configuration {\textsmaller[.5]{\ensuremath{\Varid{cfg}}}} to
\emph{terminate with a value {\textsmaller[.5]{\ensuremath{\Varid{v}}}}} ({\textsmaller[.5]{\ensuremath{\Varid{cfg}\;\downarrow_{\Varid{CESH}}\;\Varid{v}}}}), \emph{terminate}
({\textsmaller[.5]{\ensuremath{\Varid{cfg}\;\downarrow_{\Varid{CESH}}}}}), and \emph{diverge} ({\textsmaller[.5]{\ensuremath{\Varid{cfg}\;\uparrow_{\Varid{CESH}}}}}). These are analogous to
the definitions for the CES machine, with the difference that the CESH
machine is allowed to terminate with \emph{any} heap since it never
deallocates
\iffullversion
anything:
\fi
\iffullversion
\begin{hscode}\SaveRestoreHook
\column{B}{@{}>{\hspre}l<{\hspost}@{}}%
\column{E}{@{}>{\hspre}l<{\hspost}@{}}%
\>[B]{}{\dummy\downarrow_{\Varid{CESH}}\dummy}\;\mathbin{:}\;\Conid{Config}\;\Varid{→}\;\Conid{Value}\;\Varid{→}\;\star{}\<[E]%
\ColumnHook
\end{hscode}\resethooks
\removecodespace
\fi
\begin{hscode}\SaveRestoreHook
\column{B}{@{}>{\hspre}l<{\hspost}@{}}%
\column{E}{@{}>{\hspre}l<{\hspost}@{}}%
\>[B]{}\Varid{cfg}\;\downarrow_{\Varid{CESH}}\;\Varid{v}\;\mathrel{=}\;\Varid{∃}\;\Varid{λ}\;\Varid{h}\;\Varid{→}\;\Varid{cfg}\;\xrightarrow[\Varid{CESH}]{}^*\;(\MyConid{END},[\mskip1.5mu \mskip1.5mu],[\mskip1.5mu \MyConid{val}\;\Varid{v}\mskip1.5mu],\Varid{h}){}\<[E]%
\ColumnHook
\end{hscode}\resethooks
\iffullversion
\removecodespace
\begin{hscode}\SaveRestoreHook
\column{B}{@{}>{\hspre}l<{\hspost}@{}}%
\column{E}{@{}>{\hspre}l<{\hspost}@{}}%
\>[B]{}{\dummy\downarrow_{\Varid{CESH}}}\;\mathbin{:}\;\Conid{Config}\;\Varid{→}\;\star{}\<[E]%
\\
\>[B]{}\Varid{cfg}\;\downarrow_{\Varid{CESH}}\;\mathrel{=}\;\Varid{∃}\;\Varid{λ}\;\Varid{v}\;\Varid{→}\;\Varid{cfg}\;\downarrow_{\Varid{CESH}}\;\Varid{v}{}\<[E]%
\\
\>[B]{}{\dummy\uparrow_{\Varid{CESH}}}\;\mathbin{:}\;\Conid{Config}\;\Varid{→}\;\star{}\<[E]%
\\
\>[B]{}{\dummy\uparrow_{\Varid{CESH}}}\;\mathrel{=}\;\Varid{↑}\;{\dummy\xrightarrow[\Varid{CESH}]{}\dummy}{}\<[E]%
\ColumnHook
\end{hscode}\resethooks
\fi
\iffullversion
\subsection{An interface for heaps} \label{section:Heaps}

In this section we formally define an abstract interface for the type
constructor {\textsmaller[.5]{\ensuremath{\Conid{Heap}}}} and its associated functions that we use to model
the dynamic memory allocation that we will need in our system.  We
will leave the details unspecified, requiring instead that it captures
certain algebraic properties that should be fulfilled by any
reasonable implementation.

The type {\textsmaller[.5]{\ensuremath{\Conid{Heap}\;\Conid{A}}}} models heaps with memory cells of type {\textsmaller[.5]{\ensuremath{\Conid{A}}}}, and
{\textsmaller[.5]{\ensuremath{\Conid{Ptr}}}} pointers into some heap. We require the existence of a heap {\textsmaller[.5]{\ensuremath{\Varid{∅}}}}
without any other requirements.
\savecolumns
\begin{hscode}\SaveRestoreHook
\column{B}{@{}>{\hspre}l<{\hspost}@{}}%
\column{3}{@{}>{\hspre}l<{\hspost}@{}}%
\column{13}{@{}>{\hspre}l<{\hspost}@{}}%
\column{E}{@{}>{\hspre}l<{\hspost}@{}}%
\>[B]{}\Keyword{abstract}{}\<[E]%
\\
\>[B]{}\hsindent{3}{}\<[3]%
\>[3]{}\Conid{Heap}\;{}\<[13]%
\>[13]{}\mathbin{:}\;\star\;\Varid{→}\;\star{}\<[E]%
\\
\>[B]{}\hsindent{3}{}\<[3]%
\>[3]{}\Conid{Ptr}\;{}\<[13]%
\>[13]{}\mathbin{:}\;\star{}\<[E]%
\\
\>[B]{}\hsindent{3}{}\<[3]%
\>[3]{}\Varid{∅}\;{}\<[13]%
\>[13]{}\mathbin{:}\;\{\mskip1.5mu \Conid{A}\;\mathbin{:}\;\star\mskip1.5mu\}\;\Varid{→}\;\Conid{Heap}\;\Conid{A}{}\<[E]%
\ColumnHook
\end{hscode}\resethooks
\restorecolumns
We need to be able to attempt to lookup (or dereference) pointers,
and allocate new items. Allocating gives back a new heap and a pointer.
\restorecolumns
\begin{hscode}\SaveRestoreHook
\column{B}{@{}>{\hspre}l<{\hspost}@{}}%
\column{3}{@{}>{\hspre}l<{\hspost}@{}}%
\column{13}{@{}>{\hspre}l<{\hspost}@{}}%
\column{16}{@{}>{\hspre}l<{\hspost}@{}}%
\column{E}{@{}>{\hspre}l<{\hspost}@{}}%
\>[3]{}\Varid{\char95 !\char95 }\;{}\<[13]%
\>[13]{}\mathbin{:}\;\{\mskip1.5mu \Conid{A}\;\mathbin{:}\;\star\mskip1.5mu\}\;\Varid{→}\;\Conid{Heap}\;\Conid{A}\;\Varid{→}\;\Conid{Ptr}\;\Varid{→}\;\Conid{Maybe}\;\Conid{A}{}\<[E]%
\\
\>[3]{}\Varid{\char95 ▸\char95 }\;{}\<[13]%
\>[13]{}\mathbin{:}\;{}\<[16]%
\>[16]{}\{\mskip1.5mu \Conid{A}\;\mathbin{:}\;\star\mskip1.5mu\}\;\Varid{→}\;\Conid{Heap}\;\Conid{A}\;\Varid{→}\;\Conid{A}\;\Varid{→}\;\Conid{Heap}\;\Conid{A}\;\Varid{×}\;\Conid{Ptr}{}\<[E]%
\ColumnHook
\end{hscode}\resethooks
We require the following relationship between dereferencing and
allocation: if we dereference a pointer that was obtained from
allocating a memory cell with value {\textsmaller[.5]{\ensuremath{\Varid{x}}}}, we get {\textsmaller[.5]{\ensuremath{\Varid{x}}}} back:
\restorecolumns
\begin{hscode}\SaveRestoreHook
\column{B}{@{}>{\hspre}l<{\hspost}@{}}%
\column{3}{@{}>{\hspre}l<{\hspost}@{}}%
\column{13}{@{}>{\hspre}l<{\hspost}@{}}%
\column{16}{@{}>{\hspre}l<{\hspost}@{}}%
\column{E}{@{}>{\hspre}l<{\hspost}@{}}%
\>[3]{}\Varid{!-▸}\;{}\<[13]%
\>[13]{}\mathbin{:}\;\{\mskip1.5mu \Conid{A}\;\mathbin{:}\;\star\mskip1.5mu\}\;(\Varid{h}\;\mathbin{:}\;\Conid{Heap}\;\Conid{A})\;(\Varid{x}\;\mathbin{:}\;\Conid{A})\;\Varid{→}{}\<[E]%
\\
\>[13]{}\hsindent{3}{}\<[16]%
\>[16]{}\Keyword{let}\;(\Varid{h'},\Varid{ptr})\;\mathrel{=}\;\Varid{h}\;\Varid{▸}\;\Varid{x}\;\Keyword{in}\;\Varid{h'}\;\mathbin{!}\;\Varid{ptr}\;\Varid{≡}\;\MyConid{just}\;\Varid{x}{}\<[E]%
\ColumnHook
\end{hscode}\resethooks
We define a preorder {\textsmaller[.5]{\ensuremath{\Varid{⊆}}}} for \emph{sub-heaps}. The intuitive reading
for {\textsmaller[.5]{\ensuremath{\Varid{h}\;\Varid{⊆}\;\Varid{h'}}}} is that {\textsmaller[.5]{\ensuremath{\Varid{h'}}}} can be used where {\textsmaller[.5]{\ensuremath{\Varid{h}}}} can, i.e. that {\textsmaller[.5]{\ensuremath{\Varid{h'}}}}
contains at least the allocations of {\textsmaller[.5]{\ensuremath{\Varid{h}}}}.
\begin{hscode}\SaveRestoreHook
\column{B}{@{}>{\hspre}l<{\hspost}@{}}%
\column{10}{@{}>{\hspre}l<{\hspost}@{}}%
\column{E}{@{}>{\hspre}l<{\hspost}@{}}%
\>[B]{}\Varid{\char95 ⊆\char95 }\;{}\<[10]%
\>[10]{}\mathbin{:}\;\{\mskip1.5mu \Conid{A}\;\mathbin{:}\;\star\mskip1.5mu\}\;\Varid{→}\;\Conid{Heap}\;\Conid{A}\;\Varid{→}\;\Conid{Heap}\;\Conid{A}\;\Varid{→}\;\star{}\<[E]%
\\
\>[B]{}\Varid{h₁}\;\Varid{⊆}\;\Varid{h₂}\;\mathrel{=}\;\Varid{∀}\;\Varid{ptr}\;\{\mskip1.5mu \Varid{x}\mskip1.5mu\}\;\Varid{→}\;\Varid{h₁}\;\mathbin{!}\;\Varid{ptr}\;\Varid{≡}\;\MyConid{just}\;\Varid{x}\;\Varid{→}\;\Varid{h₂}\;\mathbin{!}\;\Varid{ptr}\;\Varid{≡}\;\MyConid{just}\;\Varid{x}{}\<[E]%
\\[\blanklineskip]%
\>[B]{}\Varid{⊆-refl}\;{}\<[10]%
\>[10]{}\mathbin{:}\;\{\mskip1.5mu \Conid{A}\;\mathbin{:}\;\star\mskip1.5mu\}\;(\Varid{h}\;\mathbin{:}\;\Conid{Heap}\;\Conid{A})\;\Varid{→}\;\Varid{h}\;\Varid{⊆}\;\Varid{h}{}\<[E]%
\\
\>[B]{}\Varid{⊆-refl}\;\Varid{h}\;\Varid{ptr}\;\Varid{eq}\;\mathrel{=}\;\Varid{eq}{}\<[E]%
\\[\blanklineskip]%
\>[B]{}\Varid{⊆-trans}\;{}\<[10]%
\>[10]{}\mathbin{:}\;\{\mskip1.5mu \Conid{A}\;\mathbin{:}\;\star\mskip1.5mu\}\;\{\mskip1.5mu \Varid{h₁}\;\Varid{h₂}\;\Varid{h₃}\;\mathbin{:}\;\Conid{Heap}\;\Conid{A}\mskip1.5mu\}\;{}\<[E]%
\\
\>[10]{}\Varid{→}\;\Varid{h₁}\;\Varid{⊆}\;\Varid{h₂}\;\Varid{→}\;\Varid{h₂}\;\Varid{⊆}\;\Varid{h₃}\;\Varid{→}\;\Varid{h₁}\;\Varid{⊆}\;\Varid{h₃}{}\<[E]%
\\
\>[B]{}\Varid{⊆-trans}\;\Varid{h₁⊆h₂}\;\Varid{h₂⊆h₃}\;\Varid{ptr}\;\Varid{eq}\;\mathrel{=}\;\Varid{h₂⊆h₃}\;\Varid{ptr}\;(\Varid{h₁⊆h₂}\;\Varid{ptr}\;\Varid{eq}){}\<[E]%
\ColumnHook
\end{hscode}\resethooks
Our last requirement is that allocation does not overwrite any memory
cells that were previously allocated ({\textsmaller[.5]{\ensuremath{\Varid{proj₁}}}} means first projection):
\restorecolumns
\begin{hscode}\SaveRestoreHook
\column{B}{@{}>{\hspre}l<{\hspost}@{}}%
\column{3}{@{}>{\hspre}l<{\hspost}@{}}%
\column{13}{@{}>{\hspre}l<{\hspost}@{}}%
\column{E}{@{}>{\hspre}l<{\hspost}@{}}%
\>[B]{}\Keyword{abstract}{}\<[E]%
\\
\>[B]{}\hsindent{3}{}\<[3]%
\>[3]{}\Varid{h⊆h▸x}\;{}\<[13]%
\>[13]{}\mathbin{:}\;\{\mskip1.5mu \Conid{A}\;\mathbin{:}\;\star\mskip1.5mu\}\;(\Varid{h}\;\mathbin{:}\;\Conid{Heap}\;\Conid{A})\;\{\mskip1.5mu \Varid{x}\;\mathbin{:}\;\Conid{A}\mskip1.5mu\}\;\Varid{→}\;\Varid{h}\;\Varid{⊆}\;\Varid{proj₁}\;(\Varid{h}\;\Varid{▸}\;\Varid{x}){}\<[E]%
\ColumnHook
\end{hscode}\resethooks
\fi
\subsection{Correctness} \label{section:CESH-bisim}
To show that our definition of the machine is correct, we construct
a bisimulation between the CES and CESH machines.
Since the configurations of the machines are very similar, the
intuition for the relation that we construct is simply that it is
almost equality. The only place where it is not equality is for
closure values, where the CESH machine stores pointers instead of
closures directly.  To construct a relation for closure values we need
to to \emph{parameterise} it by the heap of the CESH configuration,
and then make sure that the closure pointer points to a closure
related to the CES closure.

Formally, the relation is constructed separately for the
different components of the machine configurations.  Since they run
the same bytecode, we let the relation for code be equality:
\begin{hscode}\SaveRestoreHook
\column{B}{@{}>{\hspre}l<{\hspost}@{}}%
\column{E}{@{}>{\hspre}l<{\hspost}@{}}%
\>[B]{}\Varid{R}_{\Varid{Code}}\;\mathbin{:}\;\Conid{Rel}\;\Conid{Code}\;\Conid{Code}{}\<[E]%
\\
\>[B]{}\Varid{R}_{\Varid{Code}}\;\Varid{c₁}\;\Varid{c₂}\;\mathrel{=}\;\Varid{c₁}\;\Varid{≡}\;\Varid{c₂}{}\<[E]%
\ColumnHook
\end{hscode}\resethooks
We forward declare the relation for environments and define the
relation for closures, which is simply that the components of the
closures are related. Since we have used the same names for some of
the components of the CES and CESH machines, we qualify them, using
Agda's qualified imports, by prepending {\textsmaller[.5]{\ensuremath{\Conid{CES.}}}} and {\textsmaller[.5]{\ensuremath{\Conid{CESH.}}}} to their
names. These components may contain values, so we have to
parameterise the relations by a closure heap (here {\textsmaller[.5]{\ensuremath{\Conid{ClosHeap}\;\mathrel{=}\;\Conid{Heap}\;\Conid{CESH.Closure}}}}).
\begin{hscode}\SaveRestoreHook
\column{B}{@{}>{\hspre}l<{\hspost}@{}}%
\column{12}{@{}>{\hspre}l<{\hspost}@{}}%
\column{E}{@{}>{\hspre}l<{\hspost}@{}}%
\>[B]{}\Varid{R}_{\Varid{Env}}\;{}\<[12]%
\>[12]{}\mathbin{:}\;\Conid{ClosHeap}\;\Varid{→}\;\Conid{Rel}\;\Conid{CES.Env}\;\Conid{CESH.Env}{}\<[E]%
\\
\>[B]{}\Varid{R}_{\Varid{Clos}}\;{}\<[12]%
\>[12]{}\mathbin{:}\;\Conid{ClosHeap}\;\Varid{→}\;\Conid{Rel}\;\Conid{CES.Closure}\;\Conid{CESH.Closure}{}\<[E]%
\\
\>[B]{}\Varid{R}_{\Varid{Clos}}\;\Varid{h}\;(\Varid{c₁},\Varid{e₁})\;(\Varid{c₂},\Varid{e₂})\;\mathrel{=}\;\Varid{R}_{\Varid{Code}}\;\Varid{c₁}\;\Varid{c₂}\;\Varid{×}\;\Varid{R}_{\Varid{Env}}\;\Varid{h}\;\Varid{e₁}\;\Varid{e₂}{}\<[E]%
\ColumnHook
\end{hscode}\resethooks
Two values are unrelated  if
they do not start with the same constructor.  When they do start with
the same constructor, there are two cases: If the two values are
number literals, they are related if they are equal. If the two values
are a CES closure and a pointer, we require that the pointer points to
a CESH closure that is related to the CES closure.
\begin{hscode}\SaveRestoreHook
\column{B}{@{}>{\hspre}l<{\hspost}@{}}%
\column{3}{@{}>{\hspre}l<{\hspost}@{}}%
\column{20}{@{}>{\hspre}l<{\hspost}@{}}%
\column{32}{@{}>{\hspre}l<{\hspost}@{}}%
\column{E}{@{}>{\hspre}l<{\hspost}@{}}%
\>[B]{}\Varid{R}_{\Varid{Val}}\;\mathbin{:}\;\Conid{ClosHeap}\;\Varid{→}\;\Conid{Rel}\;\Conid{CES.Value}\;\Conid{CESH.Value}{}\<[E]%
\\
\>[B]{}\Varid{R}_{\Varid{Val}}\;\Varid{h}\;(\MyConid{nat}\;\Varid{n₁})\;{}\<[20]%
\>[20]{}(\MyConid{nat}\;\Varid{n₂})\;{}\<[32]%
\>[32]{}\mathrel{=}\;\Varid{n₁}\;\Varid{≡}\;\Varid{n₂}{}\<[E]%
\\
\>[B]{}\Varid{R}_{\Varid{Val}}\;\Varid{h}\;(\MyConid{nat}\;\anonymous )\;{}\<[20]%
\>[20]{}(\MyConid{clos}\;\anonymous )\;{}\<[32]%
\>[32]{}\mathrel{=}\;\Varid{⊥}{}\<[E]%
\\
\>[B]{}\Varid{R}_{\Varid{Val}}\;\Varid{h}\;(\MyConid{clos}\;\anonymous )\;{}\<[20]%
\>[20]{}(\MyConid{nat}\;\anonymous )\;{}\<[32]%
\>[32]{}\mathrel{=}\;\Varid{⊥}{}\<[E]%
\\
\>[B]{}\Varid{R}_{\Varid{Val}}\;\Varid{h}\;(\MyConid{clos}\;\Varid{c₁})\;{}\<[20]%
\>[20]{}(\MyConid{clos}\;\Varid{ptr})\;{}\<[32]%
\>[32]{}\mathrel{=}\;\Varid{∃}\;\Varid{λ}\;\Varid{c₂}\;\Varid{→}\;{}\<[E]%
\\
\>[B]{}\hsindent{3}{}\<[3]%
\>[3]{}\Varid{h}\;\mathbin{!}\;\Varid{ptr}\;\Varid{≡}\;\MyConid{just}\;\Varid{c₂}\;\Varid{×}\;\Varid{R}_{\Varid{Clos}}\;\Varid{h}\;\Varid{c₁}\;\Varid{c₂}{}\<[E]%
\ColumnHook
\end{hscode}\resethooks
Two environments are related if they have the same list spine
and their values are pointwise related.
\begin{hscode}\SaveRestoreHook
\column{B}{@{}>{\hspre}l<{\hspost}@{}}%
\column{20}{@{}>{\hspre}l<{\hspost}@{}}%
\column{31}{@{}>{\hspre}l<{\hspost}@{}}%
\column{E}{@{}>{\hspre}l<{\hspost}@{}}%
\>[B]{}\Varid{R}_{\Varid{Env}}\;\Varid{h}\;[\mskip1.5mu \mskip1.5mu]\;{}\<[20]%
\>[20]{}[\mskip1.5mu \mskip1.5mu]\;{}\<[31]%
\>[31]{}\mathrel{=}\;\Varid{⊤}{}\<[E]%
\\
\>[B]{}\Varid{R}_{\Varid{Env}}\;\Varid{h}\;[\mskip1.5mu \mskip1.5mu]\;{}\<[20]%
\>[20]{}(\Varid{x₂}\;\Varid{∷}\;\Varid{e₂})\;{}\<[31]%
\>[31]{}\mathrel{=}\;\Varid{⊥}{}\<[E]%
\\
\>[B]{}\Varid{R}_{\Varid{Env}}\;\Varid{h}\;(\Varid{x₁}\;\Varid{∷}\;\Varid{e₁})\;{}\<[20]%
\>[20]{}[\mskip1.5mu \mskip1.5mu]\;{}\<[31]%
\>[31]{}\mathrel{=}\;\Varid{⊥}{}\<[E]%
\\
\>[B]{}\Varid{R}_{\Varid{Env}}\;\Varid{h}\;(\Varid{x₁}\;\Varid{∷}\;\Varid{e₁})\;{}\<[20]%
\>[20]{}(\Varid{x₂}\;\Varid{∷}\;\Varid{e₂})\;{}\<[31]%
\>[31]{}\mathrel{=}\;\Varid{R}_{\Varid{Val}}\;\Varid{h}\;\Varid{x₁}\;\Varid{x₂}\;\Varid{×}\;\Varid{R}_{\Varid{Env}}\;\Varid{h}\;\Varid{e₁}\;\Varid{e₂}{}\<[E]%
\ColumnHook
\end{hscode}\resethooks
Note that we use {\textsmaller[.5]{\ensuremath{\Varid{⊤}}}} and {\textsmaller[.5]{\ensuremath{\Varid{⊥}}}} to represent true and false, represented in Agda
by the unit type and the uninhabited type.
The relation on stacks, {\textsmaller[.5]{\ensuremath{\Varid{R}_{\Varid{Stack}}}}} is defined similarly, using {\textsmaller[.5]{\ensuremath{\Varid{R}_{\Varid{Val}}}}} and
{\textsmaller[.5]{\ensuremath{\Varid{R}_{\Varid{Clos}}}}} for values and continuations.
\iffullversion
\begin{hscode}\SaveRestoreHook
\column{B}{@{}>{\hspre}l<{\hspost}@{}}%
\column{16}{@{}>{\hspre}l<{\hspost}@{}}%
\column{23}{@{}>{\hspre}l<{\hspost}@{}}%
\column{26}{@{}>{\hspre}l<{\hspost}@{}}%
\column{34}{@{}>{\hspre}l<{\hspost}@{}}%
\column{35}{@{}>{\hspre}l<{\hspost}@{}}%
\column{37}{@{}>{\hspre}l<{\hspost}@{}}%
\column{E}{@{}>{\hspre}l<{\hspost}@{}}%
\>[B]{}\Varid{R}_{\Varid{StackElem}}\;\mathbin{:}\;{}\<[16]%
\>[16]{}\Conid{ClosHeap}\;\Varid{→}\;\Conid{Rel}\;\Conid{CES.StackElem}\;\Conid{CESH.StackElem}{}\<[E]%
\\
\>[B]{}\Varid{R}_{\Varid{StackElem}}\;\Varid{h}\;(\MyConid{val}\;\Varid{v₁})\;{}\<[26]%
\>[26]{}(\MyConid{val}\;\Varid{v₂})\;{}\<[37]%
\>[37]{}\mathrel{=}\;\Varid{R}_{\Varid{Val}}\;\Varid{h}\;\Varid{v₁}\;\Varid{v₂}{}\<[E]%
\\
\>[B]{}\Varid{R}_{\Varid{StackElem}}\;\Varid{h}\;(\MyConid{val}\;\anonymous )\;{}\<[26]%
\>[26]{}(\MyConid{cont}\;\anonymous )\;{}\<[37]%
\>[37]{}\mathrel{=}\;\Varid{⊥}{}\<[E]%
\\
\>[B]{}\Varid{R}_{\Varid{StackElem}}\;\Varid{h}\;(\MyConid{cont}\;\anonymous )\;{}\<[26]%
\>[26]{}(\MyConid{val}\;\anonymous )\;{}\<[37]%
\>[37]{}\mathrel{=}\;\Varid{⊥}{}\<[E]%
\\
\>[B]{}\Varid{R}_{\Varid{StackElem}}\;\Varid{h}\;(\MyConid{cont}\;\Varid{c₁})\;{}\<[26]%
\>[26]{}(\MyConid{cont}\;\Varid{c₂})\;{}\<[37]%
\>[37]{}\mathrel{=}\;\Varid{R}_{\Varid{Clos}}\;\Varid{h}\;\Varid{c₁}\;\Varid{c₂}{}\<[E]%
\\
\>[B]{}\Varid{R}_{\Varid{Stack}}\;\mathbin{:}\;\Conid{ClosHeap}\;\Varid{→}\;\Conid{Rel}\;\Conid{CES.Stack}\;\Conid{CESH.Stack}{}\<[E]%
\\
\>[B]{}\Varid{R}_{\Varid{Stack}}\;\Varid{h}\;[\mskip1.5mu \mskip1.5mu]\;{}\<[23]%
\>[23]{}[\mskip1.5mu \mskip1.5mu]\;{}\<[35]%
\>[35]{}\mathrel{=}\;\Varid{⊤}{}\<[E]%
\\
\>[B]{}\Varid{R}_{\Varid{Stack}}\;\Varid{h}\;[\mskip1.5mu \mskip1.5mu]\;{}\<[23]%
\>[23]{}(\Varid{x₂}\;\Varid{∷}\;\Varid{s₂})\;{}\<[35]%
\>[35]{}\mathrel{=}\;\Varid{⊥}{}\<[E]%
\\
\>[B]{}\Varid{R}_{\Varid{Stack}}\;\Varid{h}\;(\Varid{x₁}\;\Varid{∷}\;\Varid{s₁})\;{}\<[23]%
\>[23]{}[\mskip1.5mu \mskip1.5mu]\;{}\<[35]%
\>[35]{}\mathrel{=}\;\Varid{⊥}{}\<[E]%
\\
\>[B]{}\Varid{R}_{\Varid{Stack}}\;\Varid{h}\;(\Varid{x₁}\;\Varid{∷}\;\Varid{s₁})\;{}\<[23]%
\>[23]{}(\Varid{x₂}\;\Varid{∷}\;\Varid{s₂})\;{}\<[34]%
\>[34]{}\mathrel{=}\;\Varid{R}_{\Varid{StackElem}}\;\Varid{h}\;\Varid{x₁}\;\Varid{x₂}\;\Varid{×}\;\Varid{R}_{\Varid{Stack}}\;\Varid{h}\;\Varid{s₁}\;\Varid{s₂}{}\<[E]%
\ColumnHook
\end{hscode}\resethooks
\fi

Finally, two configurations are related if their components are
related. Here we pass the heap of the CESH configuration as an
argument to the environment and stack relations.
\begin{hscode}\SaveRestoreHook
\column{B}{@{}>{\hspre}l<{\hspost}@{}}%
\column{3}{@{}>{\hspre}l<{\hspost}@{}}%
\column{E}{@{}>{\hspre}l<{\hspost}@{}}%
\>[B]{}\Varid{R}_{\Varid{Cfg}}\;\mathbin{:}\;\Conid{Rel}\;\Conid{CES.Config}\;\Conid{CESH.Config}{}\<[E]%
\\
\>[B]{}\Varid{R}_{\Varid{Cfg}}\;(\Varid{c₁},\Varid{e₁},\Varid{s₁})\;(\Varid{c₂},\Varid{e₂},\Varid{s₂},\Varid{h₂})\;\mathrel{=}\;{}\<[E]%
\\
\>[B]{}\hsindent{3}{}\<[3]%
\>[3]{}\Varid{R}_{\Varid{Code}}\;\Varid{c₁}\;\Varid{c₂}\;\Varid{×}\;\Varid{R}_{\Varid{Env}}\;\Varid{h₂}\;\Varid{e₁}\;\Varid{e₂}\;\Varid{×}\;\Varid{R}_{\Varid{Stack}}\;\Varid{h₂}\;\Varid{s₁}\;\Varid{s₂}{}\<[E]%
\ColumnHook
\end{hscode}\resethooks

\ifnotfullversion
We have not explicitly defined heaps other than mentioning the
allocation and dereferencing operations.  This is deliberate. In the
formalisation we define heaps and their properties \emph{abstractly},
meaning that we do not rely on any specific heap implementation in our
proofs.  One property that any reasonable heap must have is that
dereferencing a pointer in a heap where that pointer was just
allocated with a value gives back the same value:
{\textsmaller[.5]{\ensuremath{\Varid{∀}\;\Varid{h}\;\Varid{x}\;\Varid{→}\;\Keyword{let}\;(\Varid{h'},\Varid{ptr})\;\mathrel{=}\;\Varid{h}\;\Varid{▸}\;\Varid{x}\;\Keyword{in}\;\Varid{h'}\;\mathbin{!}\;\Varid{ptr}\;\Varid{≡}\;\MyConid{just}\;\Varid{x}}}}.

We define a preorder {\textsmaller[.5]{\ensuremath{\Varid{⊆}}}} for \emph{sub-heaps}. The intuitive reading
for {\textsmaller[.5]{\ensuremath{\Varid{h}\;\Varid{⊆}\;\Varid{h'}}}} is that {\textsmaller[.5]{\ensuremath{\Varid{h'}}}} can be used where {\textsmaller[.5]{\ensuremath{\Varid{h}}}} can, i.e. that {\textsmaller[.5]{\ensuremath{\Varid{h'}}}}
contains at least the allocations of {\textsmaller[.5]{\ensuremath{\Varid{h}}}}.
\begin{hscode}\SaveRestoreHook
\column{B}{@{}>{\hspre}l<{\hspost}@{}}%
\column{23}{@{}>{\hspre}l<{\hspost}@{}}%
\column{27}{@{}>{\hspre}l<{\hspost}@{}}%
\column{E}{@{}>{\hspre}l<{\hspost}@{}}%
\>[B]{}\Varid{h}\;\Varid{⊆}\;\Varid{h'}\;\mathrel{=}\;\Varid{∀}\;\Varid{ptr}\;\{\mskip1.5mu \Varid{x}\mskip1.5mu\}\;\Varid{→}\;{}\<[23]%
\>[23]{}\Varid{h}\;{}\<[27]%
\>[27]{}\mathbin{!}\;\Varid{ptr}\;\Varid{≡}\;\MyConid{just}\;\Varid{x}\;\Varid{→}\;{}\<[E]%
\\
\>[23]{}\Varid{h'}\;{}\<[27]%
\>[27]{}\mathbin{!}\;\Varid{ptr}\;\Varid{≡}\;\MyConid{just}\;\Varid{x}{}\<[E]%
\ColumnHook
\end{hscode}\resethooks
The second property that we require of a heap implementation is that
allocation does not overwrite any previously allocated memory cells
({\textsmaller[.5]{\ensuremath{\Varid{proj₁}}}} means first projection):
{\textsmaller[.5]{\ensuremath{\Varid{∀}\;\Varid{h}\;\Varid{x}\;\Varid{→}\;\Varid{h}\;\Varid{⊆}\;\Varid{proj₁}\;(\Varid{h}\;\Varid{▸}\;\Varid{x})}}}.
\fi

\begin{lemma}[{\textsmaller[.5]{\ensuremath{\Conid{HeapUpdate.config}}}}]
Given two heaps {\textsmaller[.5]{\ensuremath{\Varid{h}}}} and {\textsmaller[.5]{\ensuremath{\Varid{h'}}}} such that {\textsmaller[.5]{\ensuremath{\Varid{h}\;\Varid{⊆}\;\Varid{h'}}}},
if {\textsmaller[.5]{\ensuremath{\Varid{R}_{\Varid{Cfg}}\;\Varid{cfg}\;(\Varid{c},\Varid{e},\Varid{s},\Varid{h})}}}, then {\textsmaller[.5]{\ensuremath{\Varid{R}_{\Varid{Cfg}}\;\Varid{cfg}\;(\Varid{c},\Varid{e},\Varid{s},\Varid{h'})}}}.
\iffullversion
\footnote{In the actual implementation,
this is inside a local module {\textsmaller[.5]{\ensuremath{\Conid{HeapUpdate}}}}, parameterised by {\textsmaller[.5]{\ensuremath{\Varid{h}}}} and {\textsmaller[.5]{\ensuremath{\Varid{h'}}}} and their relation,
together with similar lemmas for the constituents of the machine configurations.}
\indentcolumn{2}
\begin{hscode}\SaveRestoreHook
\column{B}{@{}>{\hspre}l<{\hspost}@{}}%
\column{3}{@{}>{\hspre}l<{\hspost}@{}}%
\column{27}{@{}>{\hspre}l<{\hspost}@{}}%
\column{E}{@{}>{\hspre}l<{\hspost}@{}}%
\>[3]{}\Varid{config}\;\mathbin{:}\;\Varid{∀}\;\Varid{cfg}\;\Varid{c}\;\Varid{e}\;\Varid{s}\;\Varid{→}\;{}\<[27]%
\>[27]{}\Varid{R}_{\Varid{Cfg}}\;\Varid{cfg}\;(\Varid{c},\Varid{e},\Varid{s},\Varid{h})\;\Varid{→}\;\Varid{R}_{\Varid{Cfg}}\;\Varid{cfg}\;(\Varid{c},\Varid{e},\Varid{s},\Varid{h'}){}\<[E]%
\ColumnHook
\end{hscode}\resethooks
\fi
\end{lemma}

\begin{theorem}[{\textsmaller[.5]{\ensuremath{\Varid{simulation}}}}]
{\textsmaller[.5]{\ensuremath{\Varid{R}_{\Varid{Cfg}}}}} is a simulation relation.
\iffullversion
\begin{hscode}\SaveRestoreHook
\column{B}{@{}>{\hspre}l<{\hspost}@{}}%
\column{E}{@{}>{\hspre}l<{\hspost}@{}}%
\>[B]{}\Varid{simulation}\;\mathbin{:}\;\Conid{Simulation}\;{\dummy\xrightarrow[\Varid{CES}]{}\dummy}\;{\dummy\xrightarrow[\Varid{CESH}]{}\dummy}\;\Varid{R}_{\Varid{Cfg}}{}\<[E]%
\ColumnHook
\end{hscode}\resethooks
where
\begin{hscode}\SaveRestoreHook
\column{B}{@{}>{\hspre}l<{\hspost}@{}}%
\column{3}{@{}>{\hspre}l<{\hspost}@{}}%
\column{E}{@{}>{\hspre}l<{\hspost}@{}}%
\>[B]{}\Conid{Simulation}\;\Varid{\char95 ⟶\char95 }\;\Varid{\char95 ⟶'\char95 }\;\Varid{\char95 R\char95 }\;\mathrel{=}\;\Varid{∀}\;\Varid{a}\;\Varid{a'}\;\Varid{b}\;\Varid{→}\;{}\<[E]%
\\
\>[B]{}\hsindent{3}{}\<[3]%
\>[3]{}\Varid{a}\;\Varid{⟶}\;\Varid{a'}\;\Varid{→}\;\Varid{a}\;\Conid{R}\;\Varid{b}\;\Varid{→}\;\Varid{∃}\;\Varid{λ}\;\Varid{b'}\;\Varid{→}\;\Varid{b}\;\Varid{⟶'}\;\Varid{b'}\;\Varid{×}\;\Varid{a'}\;\Conid{R}\;\Varid{b'}{}\<[E]%
\ColumnHook
\end{hscode}\resethooks
\fi
\end{theorem}
\begin{proof}
By cases on the {\textsmaller[.5]{\ensuremath{\Conid{CES}}}} transition. In each case, the {\textsmaller[.5]{\ensuremath{\Conid{CESH}}}}
machine can make analogous transitions. Use {\textsmaller[.5]{\ensuremath{\Conid{HeapUpdate.config}}}} to
show that {\textsmaller[.5]{\ensuremath{\Varid{R}_{\Varid{Cfg}}}}} is preserved.
\end{proof}

We call a relation a \emph{presimulation} if it is almost, but not
quite, a simulation:

\begin{hscode}\SaveRestoreHook
\column{B}{@{}>{\hspre}l<{\hspost}@{}}%
\column{3}{@{}>{\hspre}l<{\hspost}@{}}%
\column{E}{@{}>{\hspre}l<{\hspost}@{}}%
\>[B]{}\Conid{Presimulation}\;\Varid{\char95 ⟶\char95 }\;\Varid{\char95 ⟶'\char95 }\;\Varid{\char95 R\char95 }\;\mathrel{=}\;\Varid{∀}\;\Varid{a}\;\Varid{a'}\;\Varid{b}\;\Varid{→}\;{}\<[E]%
\\
\>[B]{}\hsindent{3}{}\<[3]%
\>[3]{}\Varid{a}\;\Varid{⟶}\;\Varid{a'}\;\Varid{→}\;\Varid{a}\;\Conid{R}\;\Varid{b}\;\Varid{→}\;\Varid{∃}\;\Varid{λ}\;\Varid{b'}\;\Varid{→}\;\Varid{b}\;\Varid{⟶'}\;\Varid{b'}{}\<[E]%
\ColumnHook
\end{hscode}\resethooks

\begin{theorem}[{\textsmaller[.5]{\ensuremath{\Varid{presimulation}}}}]
The inverse of {\textsmaller[.5]{\ensuremath{\Varid{R}_{\Varid{Cfg}}}}} is a presimulation.
\iffullversion
\begin{hscode}\SaveRestoreHook
\column{B}{@{}>{\hspre}l<{\hspost}@{}}%
\column{32}{@{}>{\hspre}l<{\hspost}@{}}%
\column{E}{@{}>{\hspre}l<{\hspost}@{}}%
\>[B]{}\Varid{presimulation}\;\mathbin{:}\;\Conid{Presimulation}\;{}\<[32]%
\>[32]{}{\dummy\xrightarrow[\Varid{CESH}]{}\dummy}\;{\dummy\xrightarrow[\Varid{CES}]{}\dummy}\;{}\<[E]%
\\
\>[32]{}(\Varid{R}_{\Varid{Cfg}}\;\Varid{⁻¹}){}\<[E]%
\ColumnHook
\end{hscode}\resethooks
\fi
\end{theorem}

\begin{lemma}[{\textsmaller[.5]{\ensuremath{\Varid{presimulation-to-simulation}}}}]
If {\textsmaller[.5]{\ensuremath{\Conid{R}}}} is a simulation between relations {\textsmaller[.5]{\ensuremath{\Varid{⟶}}}} and {\textsmaller[.5]{\ensuremath{\Varid{⟶'}}}}, {\textsmaller[.5]{\ensuremath{\Conid{R}\;\Varid{⁻¹}}}} is a presimulation, and
{\textsmaller[.5]{\ensuremath{\Varid{⟶'}}}} is deterministic at states {\textsmaller[.5]{\ensuremath{\Varid{b}}}} related to some {\textsmaller[.5]{\ensuremath{\Varid{a}}}}, then
{\textsmaller[.5]{\ensuremath{\Conid{R}\;\Varid{⁻¹}}}} is a simulation.
\iffullversion
\begin{hscode}\SaveRestoreHook
\column{B}{@{}>{\hspre}l<{\hspost}@{}}%
\column{3}{@{}>{\hspre}l<{\hspost}@{}}%
\column{E}{@{}>{\hspre}l<{\hspost}@{}}%
\>[B]{}\Varid{presimulation-to-simulation}\;\mathbin{:}\;(\Varid{\char95 R\char95 }\;\mathbin{:}\;\Conid{Rel}\;\Conid{A}\;\Conid{B})\;{}\<[E]%
\\
\>[B]{}\hsindent{3}{}\<[3]%
\>[3]{}\Varid{→}\;\Conid{Simulation}\;\Varid{⟶}\;\Varid{⟶'}\;\Varid{\char95 R\char95 }\;{}\<[E]%
\\
\>[B]{}\hsindent{3}{}\<[3]%
\>[3]{}\Varid{→}\;\Conid{Presimulation}\;\Varid{⟶'}\;\Varid{⟶}\;(\Varid{\char95 R\char95 }\;\Varid{⁻¹})\;{}\<[E]%
\\
\>[B]{}\hsindent{3}{}\<[3]%
\>[3]{}\Varid{→}\;(\Varid{∀}\;\Varid{a}\;\Varid{b}\;\Varid{→}\;\Varid{a}\;\Conid{R}\;\Varid{b}\;\Varid{→}\;\Varid{⟶'}\;\Varid{is-deterministic-at}\;\Varid{b})\;{}\<[E]%
\\
\>[B]{}\hsindent{3}{}\<[3]%
\>[3]{}\Varid{→}\;\Conid{Simulation}\;\Varid{⟶'}\;\Varid{⟶}\;(\Varid{\char95 R\char95 }\;\Varid{⁻¹}){}\<[E]%
\ColumnHook
\end{hscode}\resethooks
where {\textsmaller[.5]{\ensuremath{\Varid{\char95 is-deterministic-at}}}} is a weaker form of determinism:
\begin{hscode}\SaveRestoreHook
\column{B}{@{}>{\hspre}l<{\hspost}@{}}%
\column{E}{@{}>{\hspre}l<{\hspost}@{}}%
\>[B]{}\Varid{\char95 is-deterministic-at\char95 }\;\mathbin{:}\;(\Conid{R}\;\mathbin{:}\;\Conid{Rel}\;\Conid{A}\;\Conid{B})\;\Varid{→}\;\Conid{A}\;\Varid{→}\;\star{}\<[E]%
\\
\>[B]{}\Varid{\char95 R\char95 }\;\Varid{is-deterministic-at}\;\Varid{a}\;\mathrel{=}\;\Varid{∀}\;\{\mskip1.5mu \Varid{b}\;\Varid{b'}\mskip1.5mu\}\;\Varid{→}\;\Varid{a}\;\Conid{R}\;\Varid{b}\;\Varid{→}\;\Varid{a}\;\Conid{R}\;\Varid{b'}\;\Varid{→}\;\Varid{b}\;\Varid{≡}\;\Varid{b'}{}\<[E]%
\ColumnHook
\end{hscode}\resethooks
\fi
\end{lemma}

\begin{theorem}[{\textsmaller[.5]{\ensuremath{\Varid{bisimulation}}}}]
{\textsmaller[.5]{\ensuremath{\Varid{R}_{\Varid{Cfg}}}}} is a bisimulation.
\iffullversion
\begin{hscode}\SaveRestoreHook
\column{B}{@{}>{\hspre}l<{\hspost}@{}}%
\column{E}{@{}>{\hspre}l<{\hspost}@{}}%
\>[B]{}\Varid{bisimulation}\;\mathbin{:}\;\Conid{Bisimulation}\;{\dummy\xrightarrow[\Varid{CES}]{}\dummy}\;{\dummy\xrightarrow[\Varid{CESH}]{}\dummy}\;\Varid{R}_{\Varid{Cfg}}{}\<[E]%
\ColumnHook
\end{hscode}\resethooks
where
\begin{hscode}\SaveRestoreHook
\column{B}{@{}>{\hspre}l<{\hspost}@{}}%
\column{E}{@{}>{\hspre}l<{\hspost}@{}}%
\>[B]{}\Conid{Bisimulation}\;\Varid{⟶}\;\Varid{⟶'}\;\Conid{R}\;\mathrel{=}\;\Conid{Simulation}\;\Varid{⟶}\;\Varid{⟶'}\;\Conid{R}\;\Varid{×}\;\Conid{Simulation}\;\Varid{⟶'}\;\Varid{⟶}\;(\Conid{R}\;\Varid{⁻¹}){}\<[E]%
\ColumnHook
\end{hscode}\resethooks
\fi
\end{theorem}
\begin{proof}
Theorem {\textsmaller[.5]{\ensuremath{\Varid{presimulation-to-simulation}}}} applied to
{\textsmaller[.5]{\ensuremath{\Varid{determinism}_{\Varid{CESH}}}}} and {\textsmaller[.5]{\ensuremath{\Varid{simulation}}}} implies that {\textsmaller[.5]{\ensuremath{\Varid{R}_{\Varid{Cfg}}\;\Varid{⁻¹}}}}
is a simulation, which together with {\textsmaller[.5]{\ensuremath{\Varid{simulation}}}} shows that
{\textsmaller[.5]{\ensuremath{\Varid{R}_{\Varid{Cfg}}}}} is a bisimulation.
\end{proof}

\begin{corollary}[{\textsmaller[.5]{\ensuremath{\Varid{termination-agrees}}}}, {\textsmaller[.5]{\ensuremath{\Varid{divergence-agrees}}}}]
\iffullversion
In particular, a CES configuration terminates with a natural number {\textsmaller[.5]{\ensuremath{\Varid{n}}}}
(diverges) if and only if a related CESH configuration terminates
with the same number (diverges):
\begin{hscode}\SaveRestoreHook
\column{B}{@{}>{\hspre}l<{\hspost}@{}}%
\column{3}{@{}>{\hspre}l<{\hspost}@{}}%
\column{E}{@{}>{\hspre}l<{\hspost}@{}}%
\>[B]{}\Varid{termination-agrees}\;\mathbin{:}\;\Varid{∀}\;\Varid{cfg₁}\;\Varid{cfg₂}\;\Varid{n}\;\Varid{→}\;{}\<[E]%
\\
\>[B]{}\hsindent{3}{}\<[3]%
\>[3]{}\Varid{R}_{\Varid{Cfg}}\;\Varid{cfg₁}\;\Varid{cfg₂}\;\Varid{→}\;\Varid{cfg₁}\;\downarrow_{\Varid{CES}}\;\MyConid{nat}\;\Varid{n}\;\Varid{↔}\;\Varid{cfg₂}\;\downarrow_{\Varid{CESH}}\;\MyConid{nat}\;\Varid{n}{}\<[E]%
\ColumnHook
\end{hscode}\resethooks
\fi
\iffullversion
\begin{hscode}\SaveRestoreHook
\column{B}{@{}>{\hspre}l<{\hspost}@{}}%
\column{3}{@{}>{\hspre}l<{\hspost}@{}}%
\column{E}{@{}>{\hspre}l<{\hspost}@{}}%
\>[B]{}\Varid{divergence-agrees}\;\mathbin{:}\;\Varid{∀}\;\Varid{cfg₁}\;\Varid{cfg₂}\;\Varid{→}\;{}\<[E]%
\\
\>[B]{}\hsindent{3}{}\<[3]%
\>[3]{}\Varid{R}_{\Varid{Cfg}}\;\Varid{cfg₁}\;\Varid{cfg₂}\;\Varid{→}\;\Varid{cfg₁}\;\uparrow_{\Varid{CES}}\;\Varid{↔}\;\Varid{cfg₂}\;\uparrow_{\Varid{CESH}}{}\<[E]%
\ColumnHook
\end{hscode}\resethooks
\fi
\end{corollary}
These results are of course not useful until we can show that there
are configurations in {\textsmaller[.5]{\ensuremath{\Varid{R}_{\Varid{Cfg}}}}}. One such example is the ``initial''
(mostly empty) configuration for a fragment of code:
\iffullversion
\begin{hscode}\SaveRestoreHook
\column{B}{@{}>{\hspre}l<{\hspost}@{}}%
\column{E}{@{}>{\hspre}l<{\hspost}@{}}%
\>[B]{}\Varid{initial-related}\;\mathbin{:}\;\Varid{∀}\;\Varid{c}\;\Varid{→}\;\Varid{R}_{\Varid{Cfg}}\;(\Varid{c},[\mskip1.5mu \mskip1.5mu],[\mskip1.5mu \mskip1.5mu])\;(\Varid{c},[\mskip1.5mu \mskip1.5mu],[\mskip1.5mu \mskip1.5mu],\Varid{∅}){}\<[E]%
\\
\>[B]{}\Varid{initial-related}\;\Varid{c}\;\mathrel{=}\;\Varid{refl},\Varid{tt},\Varid{tt}{}\<[E]%
\ColumnHook
\end{hscode}\resethooks
\fi

\section{Synchronous and asynchronous networks} \label{section:Networks}
Since we are later going to define two distributed abstract machines, it would
save us some work if we could make a network model that is general enough to be
used for both. In this section we will define models for synchronous
and asynchronous networks, that are parameterised by an underlying labelled
transition system. Both kinds of networks are modelled by two-level transition
systems, which is common in operational semantics for concurrent and parallel
languages. The idea is that the global level describes the transitions of the
system as a whole, and the low level the local transitions of the nodes in the
system.
Synchronous communication is modelled by \emph{rendezvous}, i.e.\ that two
nodes have to be ready to send and receive a message at a single point in time.
Asynchronous communication is modelled using a ``message soup'', representing
messages currently in transit, that nodes can add and remove messages from,
reminiscent of the Chemical Abstract Machine~\cite{DBLP:conf/popl/BerryB90}.

We construct an Agda module {\textsmaller[.5]{\ensuremath{\Conid{Network}}}}, parameterised by the underlying
transition relation, {\textsmaller[.5]{\ensuremath{{\dummy ⊢ \dummy \xrightarrow[\Varid{Machine}]{\dummy} \dummy}\;\mathbin{:}\;\Conid{Node}\;\Varid{→}\;\Conid{Machine}\;\Varid{→}\;\Conid{Tagged}\;\Conid{Msg}\;\Varid{→}\;\Conid{Machine}\;\Varid{→}\;\star}}}.  The sets {\textsmaller[.5]{\ensuremath{\Conid{Node}}}}, {\textsmaller[.5]{\ensuremath{\Conid{Machine}}}}, and {\textsmaller[.5]{\ensuremath{\Conid{Msg}}}} are additional
parameters.  Elements of {\textsmaller[.5]{\ensuremath{\Conid{Node}}}} will act as node identifiers, and we
assume that these enjoy decidable equality. If we were using MPI, they
would correspond to the so called node ranks, which are just machine
integers.  The type {\textsmaller[.5]{\ensuremath{\Conid{Machine}}}} is the type of the nodes'
configurations, and {\textsmaller[.5]{\ensuremath{\Conid{Msg}}}} the type of messages that the machines can
send. The {\textsmaller[.5]{\ensuremath{\Conid{Node}}}} in the type of {\textsmaller[.5]{\ensuremath{{\dummy ⊢ \dummy \xrightarrow[\Varid{Machine}]{\dummy} \dummy}}}} means, intuitively, that
the configuration of a node knows about and can depend on its own
identifier. The type constructor {\textsmaller[.5]{\ensuremath{\Conid{Tagged}}}} is used to separate
different kinds of local transitions: A {\textsmaller[.5]{\ensuremath{\Conid{Tagged}\;\Conid{Msg}}}} can be {\textsmaller[.5]{\ensuremath{\MyConid{silent}}}} (i.e. a $\tau$
transition), {\textsmaller[.5]{\ensuremath{\MyConid{send}\;\Varid{msg}}}}, or {\textsmaller[.5]{\ensuremath{\MyConid{receive}\;\Varid{msg}}}} (for {\textsmaller[.5]{\ensuremath{\Varid{msg}\;\mathbin{:}\;\Conid{Msg}}}}).
\iffullversion
\begin{hscode}\SaveRestoreHook
\column{B}{@{}>{\hspre}l<{\hspost}@{}}%
\column{3}{@{}>{\hspre}l<{\hspost}@{}}%
\column{11}{@{}>{\hspre}l<{\hspost}@{}}%
\column{16}{@{}>{\hspre}l<{\hspost}@{}}%
\column{E}{@{}>{\hspre}l<{\hspost}@{}}%
\>[B]{}\Keyword{module}\;\Conid{Network}{}\<[E]%
\\
\>[B]{}\hsindent{3}{}\<[3]%
\>[3]{}(\Conid{Node}\;\mathbin{:}\;\star){}\<[E]%
\\
\>[B]{}\hsindent{3}{}\<[3]%
\>[3]{}(\Varid{\char95 ≟\char95 }\;\mathbin{:}\;{}\<[11]%
\>[11]{}(\Varid{n}\;\Varid{n'}\;\mathbin{:}\;\Conid{Node})\;\Varid{→}\;\Conid{Dec}\;(\Varid{n}\;\Varid{≡}\;\Varid{n'})){}\<[E]%
\\
\>[B]{}\hsindent{3}{}\<[3]%
\>[3]{}\{\mskip1.5mu \Conid{Machine}\;\Conid{Msg}\;\mathbin{:}\;\star\mskip1.5mu\}{}\<[E]%
\\
\>[B]{}\hsindent{3}{}\<[3]%
\>[3]{}({\dummy ⊢ \dummy \xrightarrow[\Varid{Machine}]{\dummy} \dummy}\;\mathbin{:}\;{}\<[16]%
\>[16]{}\Conid{Node}\;\Varid{→}\;\Conid{Machine}\;\Varid{→}\;{}\<[E]%
\\
\>[16]{}\Conid{Tagged}\;\Conid{Msg}\;\Varid{→}\;\Conid{Machine}\;\Varid{→}\;\star){}\<[E]%
\\
\>[B]{}\hsindent{3}{}\<[3]%
\>[3]{}\Keyword{where}{}\<[E]%
\ColumnHook
\end{hscode}\resethooks
\fi

A synchronous network ({\textsmaller[.5]{\ensuremath{\Conid{SyncNetwork}}}}) is an indexed family of machines,
{\textsmaller[.5]{\ensuremath{\Conid{Node}\;\Varid{→}\;\Conid{Machine}}}}, representing the nodes of the system.
An asynchronous network ({\textsmaller[.5]{\ensuremath{\Conid{AsyncNetwork}}}}) is an indexed family of machines together with
a list of messages representing the messages currently
in transit, {\textsmaller[.5]{\ensuremath{(\Conid{Node}\;\Varid{→}\;\Conid{Machine})\;\Varid{×}\;\Conid{List}\;\Conid{Msg}}}}.

\iffullversion
The following function updates an element in a set indexed by node
identifiers, and will be used in defining the transition relations for
networks:
\begin{hscode}\SaveRestoreHook
\column{B}{@{}>{\hspre}l<{\hspost}@{}}%
\column{30}{@{}>{\hspre}l<{\hspost}@{}}%
\column{34}{@{}>{\hspre}l<{\hspost}@{}}%
\column{E}{@{}>{\hspre}l<{\hspost}@{}}%
\>[B]{}\Varid{update}\;\mathbin{:}\;\{\mskip1.5mu \Conid{A}\;\mathbin{:}\;\star\mskip1.5mu\}\;\Varid{→}\;(\Conid{Node}\;\Varid{→}\;\Conid{A})\;{}\<[34]%
\>[34]{}\Varid{→}\;\Conid{Node}\;\Varid{→}\;\Conid{A}\;\Varid{→}\;\Conid{Node}\;\Varid{→}\;\Conid{A}{}\<[E]%
\\
\>[B]{}\Varid{update}\;\Varid{nodes}\;\Varid{n}\;\Varid{m}\;\Varid{n'}\;\Keyword{with}\;\Varid{n'}\;\Varid{≟}\;\Varid{n}{}\<[E]%
\\
\>[B]{}\Varid{update}\;\Varid{nodes}\;\Varid{n}\;\Varid{m}\;\Varid{n'}\;\mid \;\Varid{yes}\;\Varid{p}\;{}\<[30]%
\>[30]{}\mathrel{=}\;\Varid{m}{}\<[E]%
\\
\>[B]{}\Varid{update}\;\Varid{nodes}\;\Varid{n}\;\Varid{m}\;\Varid{n'}\;\mid \;\Varid{no}\;\Varid{¬p}\;{}\<[30]%
\>[30]{}\mathrel{=}\;\Varid{nodes}\;\Varid{n'}{}\<[E]%
\ColumnHook
\end{hscode}\resethooks
\fi
\iffullversion
\begin{sidewaysfigure}
\else
\begin{figure*}[!t]
\fi
\centering
\begin{varwidth}[t]{\textwidth}
\iffullversion
\begin{hscode}\SaveRestoreHook
\column{B}{@{}>{\hspre}l<{\hspost}@{}}%
\column{24}{@{}>{\hspre}l<{\hspost}@{}}%
\column{E}{@{}>{\hspre}l<{\hspost}@{}}%
\>[B]{}\Keyword{data}\;{\dummy\xrightarrow[\Varid{Sync}]{}\dummy}\;(\Varid{nodes}\;\mathbin{:}\;{}\<[24]%
\>[24]{}\Conid{SyncNetwork})\;\mathbin{:}\;\Conid{SyncNetwork}\;\Varid{→}\;\star\;\Keyword{where}{}\<[E]%
\ColumnHook
\end{hscode}\resethooks
\removecodespace
\indentcolumn{3}
\fi
\begin{hscode}\SaveRestoreHook
\column{B}{@{}>{\hspre}l<{\hspost}@{}}%
\column{3}{@{}>{\hspre}l<{\hspost}@{}}%
\column{5}{@{}>{\hspre}l<{\hspost}@{}}%
\column{16}{@{}>{\hspre}l<{\hspost}@{}}%
\column{E}{@{}>{\hspre}l<{\hspost}@{}}%
\>[3]{}\MyConid{silent-step}\;{}\<[16]%
\>[16]{}\mathbin{:}\;\Varid{∀}\;\{\mskip1.5mu \Varid{i}\;\Varid{m'}\mskip1.5mu\}\;\Varid{→}\;\Varid{i}⊢\Varid{nodes}\;\Varid{i}\xrightarrow[\Varid{Machine}]{\MyConid{silent}}\Varid{m'}\;\Varid{→}\;\Varid{nodes}\;\xrightarrow[\Varid{Sync}]{}\;\Varid{update}\;\Varid{nodes}\;\Varid{i}\;\Varid{m'}{}\<[E]%
\\
\>[3]{}\MyConid{comm-step}\;{}\<[16]%
\>[16]{}\mathbin{:}\;\Varid{∀}\;\{\mskip1.5mu \Varid{s}\;\Varid{r}\;\Varid{msg}\;\Varid{sender'}\;\Varid{receiver'}\mskip1.5mu\}\;\Varid{→}\;\Keyword{let}\;\Varid{nodes'}\;\mathrel{=}\;\Varid{update}\;\Varid{nodes}\;\Varid{s}\;\Varid{sender'}\;\Keyword{in}{}\<[E]%
\\
\>[3]{}\hsindent{2}{}\<[5]%
\>[5]{}\Varid{s}⊢\Varid{nodes}\;\Varid{s}\xrightarrow[\Varid{Machine}]{\MyConid{send}\;\Varid{msg}}\Varid{sender'}\;\Varid{→}\;\Varid{r}⊢\Varid{nodes'}\;\Varid{r}\xrightarrow[\Varid{Machine}]{\MyConid{receive}\;\Varid{msg}}\Varid{receiver'}\;\Varid{→}{}\<[E]%
\\
\>[3]{}\hsindent{2}{}\<[5]%
\>[5]{}\Varid{nodes}\;\xrightarrow[\Varid{Sync}]{}\;\Varid{update}\;\Varid{nodes'}\;\Varid{r}\;\Varid{receiver'}{}\<[E]%
\ColumnHook
\end{hscode}\resethooks
\iffullversion
\begin{hscode}\SaveRestoreHook
\column{B}{@{}>{\hspre}l<{\hspost}@{}}%
\column{18}{@{}>{\hspre}l<{\hspost}@{}}%
\column{E}{@{}>{\hspre}l<{\hspost}@{}}%
\>[B]{}\Keyword{data}\;{\dummy\xrightarrow[\Varid{Async}]{}\dummy}\;\mathbin{:}\;{}\<[18]%
\>[18]{}\Conid{AsyncNetwork}\;\Varid{→}\;\Conid{AsyncNetwork}\;\Varid{→}\;\star\;\Keyword{where}{}\<[E]%
\ColumnHook
\end{hscode}\resethooks
\removecodespace
\indentcolumn{3}
\fi
\begin{hscode}\SaveRestoreHook
\column{B}{@{}>{\hspre}l<{\hspost}@{}}%
\column{3}{@{}>{\hspre}l<{\hspost}@{}}%
\column{5}{@{}>{\hspre}l<{\hspost}@{}}%
\column{41}{@{}>{\hspre}l<{\hspost}@{}}%
\column{E}{@{}>{\hspre}l<{\hspost}@{}}%
\>[3]{}\MyConid{step}\;\mathbin{:}\;\Varid{∀}\;\{\mskip1.5mu \Varid{nodes}\mskip1.5mu\}\;\Varid{msgs}_{\Varid{l}}\;\Varid{msgs}_{\Varid{r}}\;\{\mskip1.5mu \Varid{tmsg}\;\Varid{m'}\;\Varid{i}\mskip1.5mu\}\;\Varid{→}\;\Keyword{let}\;(\Varid{msgs}_{\Varid{in}},\Varid{msgs}_{\Varid{out}})\;\mathrel{=}\;\Varid{detag}\;\Varid{tmsg}\;\Keyword{in}{}\<[E]%
\\
\>[3]{}\hsindent{2}{}\<[5]%
\>[5]{}\Varid{i}⊢\Varid{nodes}\;\Varid{i}\xrightarrow[\Varid{Machine}]{\Varid{tmsg}}\Varid{m'}\;\Varid{→}{}\<[E]%
\\
\>[3]{}\hsindent{2}{}\<[5]%
\>[5]{}(\Varid{nodes},\Varid{msgs}_{\Varid{l}}\;\plus \;\Varid{msgs}_{\Varid{in}}\;\plus \;\Varid{msgs}_{\Varid{r}})\;{}\<[41]%
\>[41]{}\xrightarrow[\Varid{Async}]{}\;(\Varid{update}\;\Varid{nodes}\;\Varid{i}\;\Varid{m'},\Varid{msgs}_{\Varid{l}}\;\plus \;\Varid{msgs}_{\Varid{out}}\;\plus \;\Varid{msgs}_{\Varid{r}}){}\<[E]%
\ColumnHook
\end{hscode}\resethooks
\end{varwidth}
\caption{The definition of the transition relations for synchronous and asynchronous networks.}
\label{figure:Networks-step}
\iffullversion
\end{sidewaysfigure}
\else
\end{figure*}
\fi

Fig.~\ref{figure:Networks-step} shows the definition of the transition
relation for synchronous and asynchronous networks.
\ifnotfullversion
It uses
the {\textsmaller[.5]{\ensuremath{\Varid{update}}}} function which updates an element of a family indexed by
nodes (using decidable equality).
\fi

There are two ways for a synchronous network to make a transition. The
first, {\textsmaller[.5]{\ensuremath{\MyConid{silent-step}}}}, occurs when a machine in the network makes a transition
tagged with {\textsmaller[.5]{\ensuremath{\MyConid{silent}}}}, and is allowed at any time.
The second, {\textsmaller[.5]{\ensuremath{\MyConid{comm-step}}}}, is the aforementioned rendezvous. A node {\textsmaller[.5]{\ensuremath{\Varid{s}}}}
first takes a step sending a message, and afterwards a node {\textsmaller[.5]{\ensuremath{\Varid{r}}}} takes a step
receiving the same message. Note that {\textsmaller[.5]{\ensuremath{\Varid{s}}}} and {\textsmaller[.5]{\ensuremath{\Varid{r}}}} are not necessarily
different, i.e. nodes can send messages to themselves.
Asynchronous networks only have one rule, {\textsmaller[.5]{\ensuremath{\MyConid{step}}}}, which can be used if
a node steps with a tagged message that ``agrees''
with the list of messages in transit. The definition is fairly involved, but
the intuition is that if the node \emph{receives} a message, the message has to
be in the list \emph{before} the transition. If the node \emph{sends}
a message, it has to be there \emph{after}. If the node takes a
\emph{silent} step, the list stays the same before and after.
This is what the usage of the {\textsmaller[.5]{\ensuremath{\Varid{detag}}}} function, which creates lists of
input and output messages from a tagged message, achieves:

\iffullversion
\begin{hscode}\SaveRestoreHook
\column{B}{@{}>{\hspre}l<{\hspost}@{}}%
\column{10}{@{}>{\hspre}l<{\hspost}@{}}%
\column{E}{@{}>{\hspre}l<{\hspost}@{}}%
\>[B]{}\Varid{detag}\;\mathbin{:}\;{}\<[10]%
\>[10]{}\{\mskip1.5mu \Conid{A}\;\mathbin{:}\;\star\mskip1.5mu\}\;\Varid{→}\;\Conid{Tagged}\;\Conid{A}\;\Varid{→}\;\Conid{List}\;\Conid{A}\;\Varid{×}\;\Conid{List}\;\Conid{A}{}\<[E]%
\ColumnHook
\end{hscode}\resethooks
\removecodespace
\fi
\begin{hscode}\SaveRestoreHook
\column{B}{@{}>{\hspre}l<{\hspost}@{}}%
\column{20}{@{}>{\hspre}l<{\hspost}@{}}%
\column{29}{@{}>{\hspre}l<{\hspost}@{}}%
\column{E}{@{}>{\hspre}l<{\hspost}@{}}%
\>[B]{}\Varid{detag}\;\MyConid{silent}\;{}\<[20]%
\>[20]{}\mathrel{=}\;[\mskip1.5mu \mskip1.5mu]{}\<[29]%
\>[29]{},[\mskip1.5mu \mskip1.5mu]{}\<[E]%
\\
\>[B]{}\Varid{detag}\;(\MyConid{send}\;\Varid{x})\;{}\<[20]%
\>[20]{}\mathrel{=}\;[\mskip1.5mu \mskip1.5mu]{}\<[29]%
\>[29]{},[\mskip1.5mu \Varid{x}\mskip1.5mu]{}\<[E]%
\\
\>[B]{}\Varid{detag}\;(\MyConid{receive}\;\Varid{x})\;{}\<[20]%
\>[20]{}\mathrel{=}\;[\mskip1.5mu \Varid{x}\mskip1.5mu]{}\<[29]%
\>[29]{},[\mskip1.5mu \mskip1.5mu]{}\<[E]%
\ColumnHook
\end{hscode}\resethooks

\begin{lemma}
\ifnotfullversion
If {\textsmaller[.5]{\ensuremath{\Varid{a}\;\xrightarrow[\Varid{Sync}]{}\;\Varid{b}}}}, then {\textsmaller[.5]{\ensuremath{(\Varid{a},[\mskip1.5mu \mskip1.5mu])\;{\xrightarrow[\Varid{Async}]{}\Varid{⁺}}\;(\Varid{b},[\mskip1.5mu \mskip1.5mu])}}}, where {\textsmaller[.5]{\ensuremath{\Varid{\char95 ⁺}}}} takes the
transitive closure of a relation.
\else
If we have a synchronous transition from {\textsmaller[.5]{\ensuremath{\Varid{a}}}} to {\textsmaller[.5]{\ensuremath{\Varid{b}}}}, then we have one or more
asynchronous transitions from {\textsmaller[.5]{\ensuremath{(\Varid{a},[\mskip1.5mu \mskip1.5mu])}}} to {\textsmaller[.5]{\ensuremath{(\Varid{b},[\mskip1.5mu \mskip1.5mu])}}}, as follows:
\begin{hscode}\SaveRestoreHook
\column{B}{@{}>{\hspre}l<{\hspost}@{}}%
\column{31}{@{}>{\hspre}l<{\hspost}@{}}%
\column{37}{@{}>{\hspre}l<{\hspost}@{}}%
\column{61}{@{}>{\hspre}l<{\hspost}@{}}%
\column{E}{@{}>{\hspre}l<{\hspost}@{}}%
\>[B]{}{\xrightarrow[\Varid{Sync}]{}\Varid{-to-}{\xrightarrow[\Varid{Async}]{}\Varid{⁺}}}\;\mathbin{:}\;\Varid{∀}\;\{\mskip1.5mu \Varid{a}\;\Varid{b}\mskip1.5mu\}\;\Varid{→}\;{}\<[31]%
\>[31]{}\Varid{a}\;\xrightarrow[\Varid{Sync}]{}\;\Varid{b}\;\Varid{→}\;(\Varid{a},[\mskip1.5mu \mskip1.5mu])\;{\xrightarrow[\Varid{Async}]{}\Varid{⁺}}\;{}\<[61]%
\>[61]{}(\Varid{b},[\mskip1.5mu \mskip1.5mu]){}\<[E]%
\\
\>[B]{}{\xrightarrow[\Varid{Sync}]{}\Varid{-to-}{\xrightarrow[\Varid{Async}]{}\Varid{⁺}}}\;(\MyConid{silent-step}\;\Varid{s})\;{}\<[37]%
\>[37]{}\mathrel{=}\;[\mskip1.5mu \MyConid{step}\;[\mskip1.5mu \mskip1.5mu]\;[\mskip1.5mu \mskip1.5mu]\;\Varid{s}\mskip1.5mu]{}\<[E]%
\\
\>[B]{}{\xrightarrow[\Varid{Sync}]{}\Varid{-to-}{\xrightarrow[\Varid{Async}]{}\Varid{⁺}}}\;(\MyConid{comm-step}\;\Varid{s₁}\;\Varid{s₂})\;{}\<[37]%
\>[37]{}\mathrel{=}\;\MyConid{step}\;[\mskip1.5mu \mskip1.5mu]\;[\mskip1.5mu \mskip1.5mu]\;\Varid{s₁}\;\Varid{∷}\;[\mskip1.5mu \MyConid{step}\;[\mskip1.5mu \mskip1.5mu]\;[\mskip1.5mu \mskip1.5mu]\;\Varid{s₂}\mskip1.5mu]{}\<[E]%
\ColumnHook
\end{hscode}\resethooks
where {\textsmaller[.5]{\ensuremath{\Varid{\char95 ⁺}}}} is defined as follows:
\begin{hscode}\SaveRestoreHook
\column{B}{@{}>{\hspre}l<{\hspost}@{}}%
\column{3}{@{}>{\hspre}l<{\hspost}@{}}%
\column{9}{@{}>{\hspre}l<{\hspost}@{}}%
\column{10}{@{}>{\hspre}l<{\hspost}@{}}%
\column{22}{@{}>{\hspre}l<{\hspost}@{}}%
\column{E}{@{}>{\hspre}l<{\hspost}@{}}%
\>[B]{}\Keyword{data}\;\Varid{\char95 ⁺}\;{}\<[10]%
\>[10]{}\{\mskip1.5mu \Conid{A}\;\mathbin{:}\;\star\mskip1.5mu\}\;(\Conid{R}\;\mathbin{:}\;\Conid{Rel}\;\Conid{A}\;\Conid{A})\;(\Varid{a}\;\mathbin{:}\;\Conid{A})\;\mathbin{:}\;\Conid{A}\;\Varid{→}\;\star\;\Keyword{where}{}\<[E]%
\\
\>[B]{}\hsindent{3}{}\<[3]%
\>[3]{}[\mskip1.5mu \anonymous \mskip1.5mu]\;{}\<[9]%
\>[9]{}\mathbin{:}\;\{\mskip1.5mu \Varid{b}\;\mathbin{:}\;\Conid{A}\mskip1.5mu\}\;{}\<[22]%
\>[22]{}\Varid{→}\;\Conid{R}\;\Varid{a}\;\Varid{b}\;\Varid{→}\;(\Conid{R}\;\Varid{⁺})\;\Varid{a}\;\Varid{b}{}\<[E]%
\\
\>[B]{}\hsindent{3}{}\<[3]%
\>[3]{}\Varid{\char95 ∷\char95 }\;{}\<[9]%
\>[9]{}\mathbin{:}\;\{\mskip1.5mu \Varid{b}\;\Varid{c}\;\mathbin{:}\;\Conid{A}\mskip1.5mu\}\;{}\<[22]%
\>[22]{}\Varid{→}\;\Conid{R}\;\Varid{a}\;\Varid{b}\;\Varid{→}\;(\Conid{R}\;\Varid{⁺})\;\Varid{b}\;\Varid{c}\;\Varid{→}\;(\Conid{R}\;\Varid{⁺})\;\Varid{a}\;\Varid{c}{}\<[E]%
\ColumnHook
\end{hscode}\resethooks
\fi
\end{lemma}

We can thus say that asynchronous networks subsume synchronous
networks.  Going in the other direction is not possible in general,
but for some specific instances of the underlying transition relation
it is, as we will see later.

\section{\DCESHi: A degenerate distributed machine} \label{section:ADCESH}

In higher-order distributed programs containing locus specifiers, we
will sometimes encounter situations where a function is not available
locally. For example, when evaluating the function {\textsmaller[.5]{\ensuremath{\Varid{f}}}} in the term {\textsmaller[.5]{\ensuremath{(\Varid{f}\;\MyConid{$\boldsymbol{@}$}\;\Conid{A})\;(\Varid{g}\;\MyConid{$\boldsymbol{@}$}\;\Conid{B})}}}, we may need to apply the remotely available function
{\textsmaller[.5]{\ensuremath{\Varid{g}}}}.  As stated in the introduction, our general idea is to do this by
decomposing some instructions into communication. In the example, the
function {\textsmaller[.5]{\ensuremath{\Varid{f}}}} may send a message requesting the evaluation of {\textsmaller[.5]{\ensuremath{\Varid{g}}}},
meaning that the {\textsmaller[.5]{\ensuremath{\MyConid{APPL}}}} instruction is split into a pair of
instructions: {\textsmaller[.5]{\ensuremath{\MyConid{APPL-send}}}} and {\textsmaller[.5]{\ensuremath{\MyConid{APPL-receive}}}}.

This section outlines an abstract machine, called \DCESHi{}, which
decomposes all application and return instructions into communication.
The machine is degenerate, because it runs as the sole node in a
network and sends messages to itself, but illustrates this
decomposition, which will be used in the fully distributed system.

A configuration of the \DCESHi{} machine ({\textsmaller[.5]{\ensuremath{\Conid{Machine}}}}) is a tuple
consisting of a possibly running thread ({\textsmaller[.5]{\ensuremath{\Conid{Maybe}\;\Conid{Thread}}}}), a closure
heap ({\textsmaller[.5]{\ensuremath{\Conid{Heap}\;\Conid{Closure}}}}), and a ``continuation heap'' ({\textsmaller[.5]{\ensuremath{\Conid{Heap}\;(\Conid{Closure}\;\Varid{×}\;\Conid{Stack})}}}).  Since the current work does not support parallelism, we
have at most one thread running at once.  The thread resembles a CES
configuration, {\textsmaller[.5]{\ensuremath{\Conid{Thread}\;\mathrel{=}\;\Conid{Code}\;\Varid{×}\;\Conid{Env}\;\Varid{×}\;\Conid{Stack}}}}, but stacks are defined
differently. A stack is now a list of values paired with an optional
pointer (pointing into the continuation heap), {\textsmaller[.5]{\ensuremath{\Conid{Stack}\;\mathrel{=}\;\Conid{List}\;\Conid{Val}\;\Varid{×}\;\Conid{Maybe}\;\Conid{ContPtr}}}} (where {\textsmaller[.5]{\ensuremath{\Conid{ContPtr}}}} is a more descriptive synonym for
{\textsmaller[.5]{\ensuremath{\Conid{Ptr}}}}).  The intuition here is that when performing an application,
when CES would push a continuation on the stack, the \DCESHi{} machine
is going to stop the current thread and send a message, which means
that it has to save the continuation and the remainder of the stack
in the heap for them to persist the thread's lifetime.

The optional pointer in {\textsmaller[.5]{\ensuremath{\Conid{Stack}}}} is to be thought of as being an
element at the \emph{bottom} of the list of values.  Comparing it to
the definition of the CES machine, where stacks are lists of either
values or continuations (which were just closures), we can picture their
relation: Whereas the CES machine stores the values and
continuations in a single, contiguous stack, the \DCESHi{} machine
stores first a contiguous block of values until reaching a
continuation, at which point it stores ({\textsmaller[.5]{\ensuremath{\MyConid{just}}}}) a pointer to the
continuation closure and the rest of the stack.

The definition of closures, values, and environments are otherwise just like
in the CESH machine.
\iffullversion
\begin{hscode}\SaveRestoreHook
\column{B}{@{}>{\hspre}l<{\hspost}@{}}%
\column{3}{@{}>{\hspre}l<{\hspost}@{}}%
\column{5}{@{}>{\hspre}l<{\hspost}@{}}%
\column{11}{@{}>{\hspre}l<{\hspost}@{}}%
\column{22}{@{}>{\hspre}l<{\hspost}@{}}%
\column{E}{@{}>{\hspre}l<{\hspost}@{}}%
\>[B]{}\Conid{ClosPtr}\;\mathrel{=}\;\Conid{Ptr}{}\<[E]%
\\
\>[B]{}\Keyword{mutual}{}\<[E]%
\\
\>[B]{}\hsindent{3}{}\<[3]%
\>[3]{}\Conid{Closure}\;\mathrel{=}\;\Conid{Code}\;\Varid{×}\;\Conid{Env}{}\<[E]%
\\
\>[B]{}\hsindent{3}{}\<[3]%
\>[3]{}\Keyword{data}\;\Conid{Val}\;\mathbin{:}\;\star\;\Keyword{where}{}\<[E]%
\\
\>[3]{}\hsindent{2}{}\<[5]%
\>[5]{}\MyConid{nat}\;{}\<[11]%
\>[11]{}\mathbin{:}\;\Conid{ℕ}\;{}\<[22]%
\>[22]{}\Varid{→}\;\Conid{Val}{}\<[E]%
\\
\>[3]{}\hsindent{2}{}\<[5]%
\>[5]{}\MyConid{clos}\;{}\<[11]%
\>[11]{}\mathbin{:}\;\Conid{ClosPtr}\;{}\<[22]%
\>[22]{}\Varid{→}\;\Conid{Val}{}\<[E]%
\\
\>[B]{}\hsindent{3}{}\<[3]%
\>[3]{}\Conid{Env}\;\mathrel{=}\;\Conid{List}\;\Conid{Val}{}\<[E]%
\\
\>[B]{}\Conid{ClosHeap}\;\mathrel{=}\;\Conid{Heap}\;\Conid{Closure}{}\<[E]%
\\
\>[B]{}\Conid{ContPtr}\;{}\<[11]%
\>[11]{}\mathrel{=}\;\Conid{Ptr}{}\<[E]%
\\
\>[B]{}\Conid{Stack}\;{}\<[11]%
\>[11]{}\mathrel{=}\;\Conid{List}\;\Conid{Val}\;\Varid{×}\;\Conid{Maybe}\;\Conid{ContPtr}{}\<[E]%
\\
\>[B]{}\Conid{ContHeap}\;{}\<[11]%
\>[11]{}\mathrel{=}\;\Conid{Heap}\;(\Conid{Closure}\;\Varid{×}\;\Conid{Stack}){}\<[E]%
\\
\>[B]{}\Conid{Thread}\;{}\<[11]%
\>[11]{}\mathrel{=}\;\Conid{Code}\;\Varid{×}\;\Conid{Env}\;\Varid{×}\;\Conid{Stack}{}\<[E]%
\\
\>[B]{}\Conid{Machine}\;{}\<[11]%
\>[11]{}\mathrel{=}\;\Conid{Maybe}\;\Conid{Thread}\;\Varid{×}\;\Conid{ClosHeap}\;\Varid{×}\;\Conid{ContHeap}{}\<[E]%
\ColumnHook
\end{hscode}\resethooks
\fi
The machine communicates with itself using two kinds of messages,
{\textsmaller[.5]{\ensuremath{\MyConid{APPL}}}} and {\textsmaller[.5]{\ensuremath{\MyConid{RET}}}}, corresponding to the instructions that we are
replacing with communication.
\iffullversion
\begin{hscode}\SaveRestoreHook
\column{B}{@{}>{\hspre}l<{\hspost}@{}}%
\column{3}{@{}>{\hspre}l<{\hspost}@{}}%
\column{11}{@{}>{\hspre}l<{\hspost}@{}}%
\column{22}{@{}>{\hspre}l<{\hspost}@{}}%
\column{E}{@{}>{\hspre}l<{\hspost}@{}}%
\>[B]{}\Keyword{data}\;\Conid{Msg}\;\mathbin{:}\;\star\;\Keyword{where}{}\<[E]%
\\
\>[B]{}\hsindent{3}{}\<[3]%
\>[3]{}\MyConid{APPL}\;{}\<[11]%
\>[11]{}\mathbin{:}\;\Conid{ClosPtr}\;{}\<[22]%
\>[22]{}\Varid{→}\;\Conid{Val}\;\Varid{→}\;\Conid{ContPtr}\;\Varid{→}\;\Conid{Msg}{}\<[E]%
\\
\>[B]{}\hsindent{3}{}\<[3]%
\>[3]{}\MyConid{RET}\;{}\<[11]%
\>[11]{}\mathbin{:}\;\Conid{ContPtr}\;{}\<[22]%
\>[22]{}\Varid{→}\;\Conid{Val}\;\Varid{→}\;\Conid{Msg}{}\<[E]%
\ColumnHook
\end{hscode}\resethooks
\fi

\iffullversion
\begin{sidewaysfigure}
\else
\begin{figure*}[!t]
\fi
\centering
\begin{varwidth}[t]{\textwidth}
\iffullversion
\begin{hscode}\SaveRestoreHook
\column{B}{@{}>{\hspre}l<{\hspost}@{}}%
\column{E}{@{}>{\hspre}l<{\hspost}@{}}%
\>[B]{}\Keyword{data}\;\dummy\xrightarrow{\dummy}\dummy\;\mathbin{:}\;\Conid{Machine}\;\Varid{→}\;\Conid{Tagged}\;\Conid{Msg}\;\Varid{→}\;\Conid{Machine}\;\Varid{→}\;\star\;\Keyword{where}{}\<[E]%
\ColumnHook
\end{hscode}\resethooks
\fi
\removecodespace
\iffullversion
\indentcolumn{3}
\rightaligncolumn{5}
\centeraligncolumn{83}
\savecolumns
\begin{hscode}\SaveRestoreHook
\column{B}{@{}>{\hspre}l<{\hspost}@{}}%
\column{3}{@{}>{\hspre}l<{\hspost}@{}}%
\column{5}{@{}>{\hspre}l<{\hspost}@{}}%
\column{19}{@{}>{\hspre}l<{\hspost}@{}}%
\column{83}{@{}>{\hspre}l<{\hspost}@{}}%
\column{143}{@{}>{\hspre}l<{\hspost}@{}}%
\column{E}{@{}>{\hspre}l<{\hspost}@{}}%
\>[3]{}\MyConid{VAR}\;{}\<[19]%
\>[19]{}\mathbin{:}\;\Varid{∀}\;\{\mskip1.5mu \Varid{n}\;\Varid{c}\;\Varid{e}\;\Varid{s}\;\Varid{v}\;\Varid{r}\;\Varid{h}_{\Varid{cl}}\;\Varid{h}_{\Varid{cnt}}\mskip1.5mu\}\;\Varid{→}\;\Varid{lookup}\;\Varid{n}\;\Varid{e}\;\Varid{≡}\;\MyConid{just}\;\Varid{v}\;\Varid{→}\;{}\<[E]%
\\
\>[3]{}\hsindent{2}{}\<[5]%
\>[5]{}(\MyConid{just}\;(\MyConid{VAR}\;\Varid{n}\;\MyConid{;}\;\Varid{c},\Varid{e},\Varid{s},\Varid{r}),\Varid{h}_{\Varid{cl}},\Varid{h}_{\Varid{cnt}})\;\;{}\<[83]%
\>[83]{}\xrightarrow{\MyConid{silent}}\;\;{}\<[143]%
\>[143]{}(\MyConid{just}\;(\Varid{c},\Varid{e},\Varid{v}\;\Varid{∷}\;\Varid{s},\Varid{r}),\Varid{h}_{\Varid{cl}},\Varid{h}_{\Varid{cnt}}){}\<[E]%
\\
\>[3]{}\MyConid{CLOS}\;{}\<[19]%
\>[19]{}\mathbin{:}\;\Varid{∀}\;\{\mskip1.5mu \Varid{c'}\;\Varid{c}\;\Varid{e}\;\Varid{s}\;\Varid{r}\;\Varid{h}_{\Varid{cl}}\;\Varid{h}_{\Varid{cnt}}\mskip1.5mu\}\;\Varid{→}\;\Keyword{let}\;(\Varid{h'}_{\Varid{cl}},\Varid{ptr}_{\Varid{cl}})\;\mathrel{=}\;\Varid{h}_{\Varid{cl}}\;\Varid{▸}\;(\Varid{c'},\Varid{e})\;\Keyword{in}{}\<[E]%
\\
\>[3]{}\hsindent{2}{}\<[5]%
\>[5]{}(\MyConid{just}\;(\MyConid{CLOS}\;\Varid{c'}\;\MyConid{;}\;\Varid{c},\Varid{e},\Varid{s},\Varid{r}),\Varid{h}_{\Varid{cl}},\Varid{h}_{\Varid{cnt}})\;\;{}\<[83]%
\>[83]{}\xrightarrow{\MyConid{silent}}\;\;{}\<[143]%
\>[143]{}(\MyConid{just}\;(\Varid{c},\Varid{e},\MyConid{clos}\;\Varid{ptr}_{\Varid{cl}}\;\Varid{∷}\;\Varid{s},\Varid{r}),\Varid{h'}_{\Varid{cl}},\Varid{h}_{\Varid{cnt}}){}\<[E]%
\ColumnHook
\end{hscode}\resethooks
\removecodespace
\restorecolumns
\fi
\rightaligncolumn{5}
\centeraligncolumn{83}
\indentcolumn{3}
\begin{hscode}\SaveRestoreHook
\column{B}{@{}>{\hspre}l<{\hspost}@{}}%
\column{3}{@{}>{\hspre}l<{\hspost}@{}}%
\column{5}{@{}>{\hspre}l<{\hspost}@{}}%
\column{19}{@{}>{\hspre}l<{\hspost}@{}}%
\column{83}{@{}>{\hspre}l<{\hspost}@{}}%
\column{143}{@{}>{\hspre}l<{\hspost}@{}}%
\column{E}{@{}>{\hspre}l<{\hspost}@{}}%
\>[3]{}\MyConid{APPL-send}\;{}\<[19]%
\>[19]{}\mathbin{:}\;\Varid{∀}\;\{\mskip1.5mu \Varid{c}\;\Varid{e}\;\Varid{v}\;\Varid{ptr}_{\Varid{cl}}\;\Varid{s}\;\Varid{r}\;\Varid{h}_{\Varid{cl}}\;\Varid{h}_{\Varid{cnt}}\mskip1.5mu\}\;\Varid{→}\;\Keyword{let}\;(\Varid{h'}_{\Varid{cnt}},\Varid{ptr}_{\Varid{cnt}})\;\mathrel{=}\;\Varid{h}_{\Varid{cnt}}\;\Varid{▸}\;((\Varid{c},\Varid{e}),\Varid{s},\Varid{r})\;\Keyword{in}{}\<[E]%
\\
\>[3]{}\hsindent{2}{}\<[5]%
\>[5]{}(\MyConid{just}\;(\MyConid{APPL}\;\MyConid{;}\;\Varid{c},\Varid{e},\Varid{v}\;\Varid{∷}\;\MyConid{clos}\;\Varid{ptr}_{\Varid{cl}}\;\Varid{∷}\;\Varid{s},\Varid{r}),\Varid{h}_{\Varid{cl}},\Varid{h}_{\Varid{cnt}})\;\;{}\<[83]%
\>[83]{}\xrightarrow{\MyConid{send}\;(\MyConid{APPL}\;\Varid{ptr}_{\Varid{cl}}\;\Varid{v}\;\Varid{ptr}_{\Varid{cnt}})}\;\;{}\<[143]%
\>[143]{}(\MyConid{nothing},\Varid{h}_{\Varid{cl}},\Varid{h'}_{\Varid{cnt}}){}\<[E]%
\\
\>[3]{}\MyConid{APPL-receive}\;{}\<[19]%
\>[19]{}\mathbin{:}\;\Varid{∀}\;\{\mskip1.5mu \Varid{h}_{\Varid{cl}}\;\Varid{h}_{\Varid{cnt}}\;\Varid{ptr}_{\Varid{cl}}\;\Varid{v}\;\Varid{ptr}_{\Varid{cnt}}\;\Varid{c}\;\Varid{e}\mskip1.5mu\}\;\Varid{→}\;\Varid{h}_{\Varid{cl}}\;\mathbin{!}\;\Varid{ptr}_{\Varid{cl}}\;\Varid{≡}\;\MyConid{just}\;(\Varid{c},\Varid{e})\;\Varid{→}\;{}\<[E]%
\\
\>[3]{}\hsindent{2}{}\<[5]%
\>[5]{}(\MyConid{nothing},\Varid{h}_{\Varid{cl}},\Varid{h}_{\Varid{cnt}})\;\;{}\<[83]%
\>[83]{}\xrightarrow{\MyConid{receive}\;(\MyConid{APPL}\;\Varid{ptr}_{\Varid{cl}}\;\Varid{v}\;\Varid{ptr}_{\Varid{cnt}})}\;\;{}\<[143]%
\>[143]{}(\MyConid{just}\;(\Varid{c},\Varid{v}\;\Varid{∷}\;\Varid{e},[\mskip1.5mu \mskip1.5mu],\MyConid{just}\;\Varid{ptr}_{\Varid{cnt}}),\Varid{h}_{\Varid{cl}},\Varid{h}_{\Varid{cnt}}){}\<[E]%
\\
\>[3]{}\MyConid{RET-send}\;{}\<[19]%
\>[19]{}\mathbin{:}\;\Varid{∀}\;\{\mskip1.5mu \Varid{e}\;\Varid{v}\;\Varid{ptr}_{\Varid{cnt}}\;\Varid{h}_{\Varid{cl}}\;\Varid{h}_{\Varid{cnt}}\mskip1.5mu\}\;\Varid{→}\;{}\<[E]%
\\
\>[3]{}\hsindent{2}{}\<[5]%
\>[5]{}(\MyConid{just}\;(\MyConid{RET},\Varid{e},\Varid{v}\;\Varid{∷}\;[\mskip1.5mu \mskip1.5mu],\MyConid{just}\;\Varid{ptr}_{\Varid{cnt}}),\Varid{h}_{\Varid{cl}},\Varid{h}_{\Varid{cnt}})\;\;{}\<[83]%
\>[83]{}\xrightarrow{\MyConid{send}\;(\MyConid{RET}\;\Varid{ptr}_{\Varid{cnt}}\;\Varid{v})}\;\;{}\<[143]%
\>[143]{}(\MyConid{nothing},\Varid{h}_{\Varid{cl}},\Varid{h}_{\Varid{cnt}}){}\<[E]%
\\
\>[3]{}\MyConid{RET-receive}\;{}\<[19]%
\>[19]{}\mathbin{:}\;\Varid{∀}\;\{\mskip1.5mu \Varid{h}_{\Varid{cl}}\;\Varid{h}_{\Varid{cnt}}\;\Varid{ptr}_{\Varid{cnt}}\;\Varid{v}\;\Varid{c}\;\Varid{e}\;\Varid{s}\;\Varid{r}\mskip1.5mu\}\;\Varid{→}\;\Varid{h}_{\Varid{cnt}}\;\mathbin{!}\;\Varid{ptr}_{\Varid{cnt}}\;\Varid{≡}\;\MyConid{just}\;((\Varid{c},\Varid{e}),\Varid{s},\Varid{r})\;\Varid{→}\;{}\<[E]%
\\
\>[3]{}\hsindent{2}{}\<[5]%
\>[5]{}(\MyConid{nothing},\Varid{h}_{\Varid{cl}},\Varid{h}_{\Varid{cnt}})\;\;{}\<[83]%
\>[83]{}\xrightarrow{\MyConid{receive}\;(\MyConid{RET}\;\Varid{ptr}_{\Varid{cnt}}\;\Varid{v})}\;\;{}\<[143]%
\>[143]{}(\MyConid{just}\;(\Varid{c},\Varid{e},\Varid{v}\;\Varid{∷}\;\Varid{s},\Varid{r}),\Varid{h}_{\Varid{cl}},\Varid{h}_{\Varid{cnt}}){}\<[E]%
\ColumnHook
\end{hscode}\resethooks
\iffullversion
\removecodespace
\restorecolumns
\rightaligncolumn{5}
\centeraligncolumn{83}
\indentcolumn{3}
\begin{hscode}\SaveRestoreHook
\column{B}{@{}>{\hspre}l<{\hspost}@{}}%
\column{3}{@{}>{\hspre}l<{\hspost}@{}}%
\column{5}{@{}>{\hspre}l<{\hspost}@{}}%
\column{6}{@{}>{\hspre}l<{\hspost}@{}}%
\column{19}{@{}>{\hspre}l<{\hspost}@{}}%
\column{83}{@{}>{\hspre}l<{\hspost}@{}}%
\column{143}{@{}>{\hspre}l<{\hspost}@{}}%
\column{E}{@{}>{\hspre}l<{\hspost}@{}}%
\>[3]{}\MyConid{COND-0}\;{}\<[19]%
\>[19]{}\mathbin{:}\;\Varid{∀}\;\{\mskip1.5mu \Varid{c}\;\Varid{c'}\;\Varid{e}\;\Varid{s}\;\Varid{r}\;\Varid{h}_{\Varid{cl}}\;\Varid{h}_{\Varid{cnt}}\mskip1.5mu\}\;\Varid{→}\;{}\<[E]%
\\
\>[3]{}\hsindent{2}{}\<[5]%
\>[5]{}(\MyConid{just}\;(\MyConid{COND}\;\Varid{c}\;\Varid{c'},\Varid{e},\MyConid{nat}\;\Varid{0}\;\Varid{∷}\;\Varid{s},\Varid{r}),\Varid{h}_{\Varid{cl}},\Varid{h}_{\Varid{cnt}})\;\;{}\<[83]%
\>[83]{}\xrightarrow{\MyConid{silent}}\;\;{}\<[143]%
\>[143]{}(\MyConid{just}\;(\Varid{c},\Varid{e},\Varid{s},\Varid{r}),\Varid{h}_{\Varid{cl}},\Varid{h}_{\Varid{cnt}}){}\<[E]%
\\
\>[3]{}\MyConid{COND-1+n}\;{}\<[19]%
\>[19]{}\mathbin{:}\;\Varid{∀}\;\{\mskip1.5mu \Varid{c}\;\Varid{c'}\;\Varid{e}\;\Varid{n}\;\Varid{s}\;\Varid{r}\;\Varid{h}_{\Varid{cl}}\;\Varid{h}_{\Varid{cnt}}\mskip1.5mu\}\;\Varid{→}\;{}\<[E]%
\\
\>[3]{}\hsindent{2}{}\<[5]%
\>[5]{}(\MyConid{just}\;(\MyConid{COND}\;\Varid{c}\;\Varid{c'},\Varid{e},\MyConid{nat}\;(1+\;\Varid{n})\;\Varid{∷}\;\Varid{s},\Varid{r}),\Varid{h}_{\Varid{cl}},\Varid{h}_{\Varid{cnt}})\;\;{}\<[83]%
\>[83]{}\xrightarrow{\MyConid{silent}}\;\;{}\<[143]%
\>[143]{}(\MyConid{just}\;(\Varid{c'},\Varid{e},\Varid{s},\Varid{r}),\Varid{h}_{\Varid{cl}},\Varid{h}_{\Varid{cnt}}){}\<[E]%
\\
\>[3]{}\MyConid{LIT}\;{}\<[19]%
\>[19]{}\mathbin{:}\;\Varid{∀}\;\{\mskip1.5mu \Varid{l}\;\Varid{c}\;\Varid{e}\;\Varid{s}\;\Varid{r}\;\Varid{h}_{\Varid{cl}}\;\Varid{h}_{\Varid{cnt}}\mskip1.5mu\}\;\Varid{→}\;{}\<[E]%
\\
\>[3]{}\hsindent{2}{}\<[5]%
\>[5]{}(\MyConid{just}\;(\MyConid{LIT}\;\Varid{l}\;\MyConid{;}\;\Varid{c},\Varid{e},\Varid{s},\Varid{r}),\Varid{h}_{\Varid{cl}},\Varid{h}_{\Varid{cnt}})\;\;{}\<[83]%
\>[83]{}\xrightarrow{\MyConid{silent}}\;\;{}\<[143]%
\>[143]{}(\MyConid{just}\;(\Varid{c},\Varid{e},\MyConid{nat}\;\Varid{l}\;\Varid{∷}\;\Varid{s},\Varid{r}),\Varid{h}_{\Varid{cl}},\Varid{h}_{\Varid{cnt}}){}\<[E]%
\\
\>[3]{}\MyConid{OP}\;{}\<[19]%
\>[19]{}\mathbin{:}\;\Varid{∀}\;\{\mskip1.5mu \Varid{f}\;\Varid{c}\;\Varid{e}\;\Varid{l₁}\;\Varid{l₂}\;\Varid{s}\;\Varid{r}\;\Varid{h}_{\Varid{cl}}\;\Varid{h}_{\Varid{cnt}}\mskip1.5mu\}\;\Varid{→}\;{}\<[E]%
\\
\>[3]{}\hsindent{3}{}\<[6]%
\>[6]{}(\MyConid{just}\;(\MyConid{OP}\;\Varid{f}\;\MyConid{;}\;\Varid{c},\Varid{e},\MyConid{nat}\;\Varid{l₁}\;\Varid{∷}\;\MyConid{nat}\;\Varid{l₂}\;\Varid{∷}\;\Varid{s},\Varid{r}),\Varid{h}_{\Varid{cl}},\Varid{h}_{\Varid{cnt}})\;\;{}\<[83]%
\>[83]{}\xrightarrow{\MyConid{silent}}\;\;{}\<[143]%
\>[143]{}(\MyConid{just}\;(\Varid{c},\Varid{e},\MyConid{nat}\;(\Varid{f}\;\Varid{l₁}\;\Varid{l₂})\;\Varid{∷}\;\Varid{s},\Varid{r}),\Varid{h}_{\Varid{cl}},\Varid{h}_{\Varid{cnt}}){}\<[E]%
\ColumnHook
\end{hscode}\resethooks
\fi
\removecodespace
\iffullversion
\caption{The definition of the transition relation of the \DCESHi{} machine.}
\fi
\label{figure:ADCESH-step}
\end{varwidth}
\iffullversion
\end{sidewaysfigure}
\else
\end{figure*}
\fi

Fig.~\ref{figure:ADCESH-step} defines the transition relation for the
\DCESHi{} machine, written {\textsmaller[.5]{\ensuremath{\Varid{m}\;\;\xrightarrow{\Varid{tmsg}}\;\;\Varid{m'}}}}
for a tagged message {\textsmaller[.5]{\ensuremath{\Varid{tmsg}}}} and machine configurations {\textsmaller[.5]{\ensuremath{\Varid{m}}}} and {\textsmaller[.5]{\ensuremath{\Varid{m'}}}}.
Most transitions are the same as in the CESH machine, just
framed with the additional heaps and the {\textsmaller[.5]{\ensuremath{\MyConid{just}}}} meaning that
the thread is running.
\ifnotfullversion
  We elide them for brevity.
\fi
The interesting rules are the decomposed application and return rules.  When an
application is performed, an {\textsmaller[.5]{\ensuremath{\MyConid{APPL}}}} message containing a pointer
to the closure to apply, the argument value and a pointer to a return continuation (which
is first allocated) is sent, and the thread is stopped (represented by the {\textsmaller[.5]{\ensuremath{\MyConid{nothing}}}}).
The machine can receive an application message if the thread is not running.
When that happens, the closure pointer is dereferenced and
entered, adding the received argument to the
environment. The stack is left empty apart from the continuation
pointer of the received message.
When returning from a function application, the machine sends a return
message containing the continuation pointer and the value to return.
On the receiving end of that communication, it dereferences the
continuation pointer and enters it, putting the result value on top of
the stack.

\begin{example}
We show what happens when we have instantiated the asynchronous networks of
the {\textsmaller[.5]{\ensuremath{\Conid{Network}}}} module with this transition relation, using the unit (one-element)
type for the {\textsmaller[.5]{\ensuremath{\Conid{Node}}}} set. Once again we trace the execution of our
running example, {\textsmaller[.5]{\ensuremath{\Varid{codeExample}}}}.
For readability, we write heaps with pointer mappings like {\textsmaller[.5]{\ensuremath{\{\mskip1.5mu \Varid{ptr}\;\Varid{↦}\;\Varid{element}\mskip1.5mu\}}}}.  The last list shown in each step is the message list of
the asynchronous network.
\begin{hscode}\SaveRestoreHook
\column{B}{@{}>{\hspre}l<{\hspost}@{}}%
\column{6}{@{}>{\hspre}l<{\hspost}@{}}%
\column{10}{@{}>{\hspre}l<{\hspost}@{}}%
\column{13}{@{}>{\hspre}l<{\hspost}@{}}%
\column{E}{@{}>{\hspre}l<{\hspost}@{}}%
\>[B]{}\Keyword{let}\;{}\<[6]%
\>[6]{}\Varid{h}_{\Varid{cl}}\;{}\<[13]%
\>[13]{}\mathrel{=}\;\{\mskip1.5mu \Varid{ptr₁}\;\Varid{↦}\;(\Varid{c₁},[\mskip1.5mu \mskip1.5mu])\mskip1.5mu\}{}\<[E]%
\\
\>[6]{}\Varid{h'}_{\Varid{cl}}\;{}\<[13]%
\>[13]{}\mathrel{=}\;\{\mskip1.5mu \Varid{ptr₁}\;\Varid{↦}\;(\Varid{c₁},[\mskip1.5mu \mskip1.5mu]),\Varid{ptr₂}\;\Varid{↦}\;(\Varid{c₂},[\mskip1.5mu \mskip1.5mu])\mskip1.5mu\}{}\<[E]%
\\
\>[6]{}\Varid{h}_{\Varid{cnt}}\;{}\<[13]%
\>[13]{}\mathrel{=}\;\{\mskip1.5mu \Varid{ptr}_{\Varid{cnt}}\;\Varid{↦}\;((\MyConid{END},[\mskip1.5mu \mskip1.5mu]),[\mskip1.5mu \mskip1.5mu],\MyConid{nothing})\mskip1.5mu\}{}\<[E]%
\\
\>[B]{}\Keyword{in}\;(\MyConid{just}\;({}\<[13]%
\>[13]{}\MyConid{CLOS}\;\Varid{c₁}\;\MyConid{;}\;\MyConid{CLOS}\;\Varid{c₂}\;\MyConid{;}\;\MyConid{APPL}\;\MyConid{;}\;\MyConid{END},[\mskip1.5mu \mskip1.5mu],[\mskip1.5mu \mskip1.5mu],{}\<[E]%
\\
\>[13]{}\MyConid{nothing}),\Varid{∅},\Varid{∅}),[\mskip1.5mu \mskip1.5mu]{}\<[E]%
\\
\>[B]{}\Varid{⟶⟨}\;\MyConid{step}\;\MyConid{CLOS}\;\Varid{⟩}{}\<[E]%
\\
\>[B]{}(\MyConid{just}\;({}\<[10]%
\>[10]{}\MyConid{CLOS}\;\Varid{c₂}\;\MyConid{;}\;\MyConid{APPL}\;\MyConid{;}\;\MyConid{END},[\mskip1.5mu \mskip1.5mu],[\mskip1.5mu \MyConid{clos}\;\Varid{ptr₁}\mskip1.5mu],{}\<[E]%
\\
\>[10]{}\MyConid{nothing}),\Varid{h}_{\Varid{cl}},\Varid{∅}),[\mskip1.5mu \mskip1.5mu]{}\<[E]%
\\
\>[B]{}\Varid{⟶⟨}\;\MyConid{step}\;\MyConid{CLOS}\;\Varid{⟩}{}\<[E]%
\\
\>[B]{}(\MyConid{just}\;({}\<[10]%
\>[10]{}\MyConid{APPL}\;\MyConid{;}\;\MyConid{END},[\mskip1.5mu \mskip1.5mu],[\mskip1.5mu \MyConid{clos}\;\Varid{ptr₂},\MyConid{clos}\;\Varid{ptr₁}\mskip1.5mu],{}\<[E]%
\\
\>[10]{}\MyConid{nothing}),\Varid{h'}_{\Varid{cl}},\Varid{∅}),[\mskip1.5mu \mskip1.5mu]{}\<[E]%
\\
\>[B]{}\Varid{⟶⟨}\;\MyConid{step}\;\MyConid{APPL-send}\;\Varid{⟩}{}\<[E]%
\\
\>[B]{}(\MyConid{nothing},\Varid{h'}_{\Varid{cl}},\Varid{h}_{\Varid{cnt}}),[\mskip1.5mu \MyConid{APPL}\;\Varid{ptr₁}\;(\MyConid{clos}\;\Varid{ptr₂})\;\Varid{ptr}_{\Varid{cnt}}\mskip1.5mu]{}\<[E]%
\\
\>[B]{}\Varid{⟶⟨}\;\MyConid{step}\;\MyConid{APPL-receive}\;\Varid{⟩}{}\<[E]%
\\
\>[B]{}(\MyConid{just}\;({}\<[10]%
\>[10]{}\MyConid{VAR}\;\Varid{0}\;\MyConid{;}\;\MyConid{RET},[\mskip1.5mu \MyConid{clos}\;\Varid{ptr₂}\mskip1.5mu],[\mskip1.5mu \mskip1.5mu],{}\<[E]%
\\
\>[10]{}\MyConid{just}\;\Varid{ptr}_{\Varid{cnt}}),\Varid{h'}_{\Varid{cl}},\Varid{h}_{\Varid{cnt}}),[\mskip1.5mu \mskip1.5mu]{}\<[E]%
\\
\>[B]{}\Varid{⟶⟨}\;\MyConid{step}\;(\MyConid{VAR}\;\Varid{refl})\;\Varid{⟩}{}\<[E]%
\\
\>[B]{}(\MyConid{just}\;({}\<[10]%
\>[10]{}\MyConid{RET},[\mskip1.5mu \MyConid{clos}\;\Varid{ptr₂}\mskip1.5mu],[\mskip1.5mu \MyConid{clos}\;\Varid{ptr₂}\mskip1.5mu],{}\<[E]%
\\
\>[10]{}\MyConid{just}\;\Varid{ptr}_{\Varid{cnt}}),\Varid{h'}_{\Varid{cl}},\Varid{h}_{\Varid{cnt}}),[\mskip1.5mu \mskip1.5mu]{}\<[E]%
\\
\>[B]{}\Varid{⟶⟨}\;\MyConid{step}\;\MyConid{RET-send}\;\Varid{⟩}{}\<[E]%
\\
\>[B]{}(\MyConid{nothing},\Varid{h'}_{\Varid{cl}},\Varid{h}_{\Varid{cnt}}),[\mskip1.5mu \MyConid{RET}\;\Varid{ptr}_{\Varid{cnt}}\;(\MyConid{clos}\;\Varid{ptr₂})\mskip1.5mu]{}\<[E]%
\\
\>[B]{}\Varid{⟶⟨}\;\MyConid{step}\;\MyConid{RET-receive}\;\Varid{⟩}{}\<[E]%
\\
\>[B]{}(\MyConid{just}\;(\MyConid{END},[\mskip1.5mu \mskip1.5mu],[\mskip1.5mu \MyConid{clos}\;\Varid{ptr₂}\mskip1.5mu],\MyConid{nothing}),\Varid{h'}_{\Varid{cl}},\Varid{h}_{\Varid{cnt}}),[\mskip1.5mu \mskip1.5mu]{}\<[E]%
\ColumnHook
\end{hscode}\resethooks
  Comparing this to Example~\ref{example:CES} we can see that an
  {\textsmaller[.5]{\ensuremath{\MyConid{APPL-send}}}} followed by an {\textsmaller[.5]{\ensuremath{\MyConid{APPL-receive}}}} amounts to the same
  thing as the {\textsmaller[.5]{\ensuremath{\MyConid{APPL}}}} rule in the CES machine, and similarly for the
  {\textsmaller[.5]{\ensuremath{\MyConid{RET}}}} instruction.
\end{example}

\section{\DCESHn: The distributed CESH machine} \label{section:DCESH}

We have so far seen two extensions of the CES machine. We have seen
CESH, that adds heaps, and \DCESHi{}, that decomposes instructions
into communication in a degenerate network of only one node.
Our final extension is a machine, \DCESHn{}, that supports multiple nodes.
The main problem that we now face is that there is no centralised heap,
but each node has its own local heap. This means that, for
supporting higher-order functions across node boundaries, we have to
somehow keep references to closures in the heaps of \emph{other} nodes.
Another problem is efficiency; we would like a system where we do
not pay the higher price of communication for locally running code.
The main idea for solving these two problems is to use
\emph{remote pointers}, {\textsmaller[.5]{\ensuremath{\Conid{RPtr}\;\mathrel{=}\;\Conid{Ptr}\;\Varid{×}\;\Conid{Node}}}}, pointers paired with node
identifiers signifying on what node's heap the pointer is located.
This solves the heap problem because we always know
where a pointer comes from. It can also be used to solve the
efficiency problem since we can choose what instructions to run based on
whether a pointer is local or remote. If it is local, we run the
rules of the CESH machine. If it is remote, we run the
decomposed rules of the \DCESHi{} machine.

The final extension to the term language and bytecode will
add support for locus specifiers.
\ifnotfullversion
We add a term construct {\textsmaller[.5]{\ensuremath{\Varid{t}\;\MyConid{$\boldsymbol{@}$}\;\Varid{i}}}},
and an instruction {\textsmaller[.5]{\ensuremath{\MyConid{REMOTE}\;\Varid{c}\;\Varid{i}}}} for its compilation.
\else
\begin{hscode}\SaveRestoreHook
\column{B}{@{}>{\hspre}l<{\hspost}@{}}%
\column{3}{@{}>{\hspre}l<{\hspost}@{}}%
\column{9}{@{}>{\hspre}l<{\hspost}@{}}%
\column{12}{@{}>{\hspre}l<{\hspost}@{}}%
\column{E}{@{}>{\hspre}l<{\hspost}@{}}%
\>[B]{}\Keyword{data}\;\Conid{Term}\;\mathbin{:}\;\star\;\Keyword{where}{}\<[E]%
\\
\>[B]{}\hsindent{3}{}\<[3]%
\>[3]{}\Varid{...}{}\<[E]%
\\
\>[B]{}\hsindent{3}{}\<[3]%
\>[3]{}{\dummy\MyConid{$\boldsymbol{@}$}\dummy}\;{}\<[9]%
\>[9]{}\mathbin{:}\;\Conid{Term}\;\Varid{→}\;\Conid{Node}\;\Varid{→}\;\Conid{Term}{}\<[E]%
\\
\>[B]{}\Keyword{data}\;\Conid{Instr}\;\mathbin{:}\;\star\;\Keyword{where}{}\<[E]%
\\
\>[B]{}\hsindent{3}{}\<[3]%
\>[3]{}\Varid{...}{}\<[E]%
\\
\>[B]{}\hsindent{3}{}\<[3]%
\>[3]{}\MyConid{REMOTE}\;{}\<[12]%
\>[12]{}\mathbin{:}\;\Conid{Code}\;\Varid{→}\;\Conid{Node}\;\Varid{→}\;\Conid{Instr}{}\<[E]%
\ColumnHook
\end{hscode}\resethooks
\fi
The locus specifiers, {\textsmaller[.5]{\ensuremath{\Varid{t}\;\MyConid{$\boldsymbol{@}$}\;\Varid{i}}}}, are taken to mean that the
term {\textsmaller[.5]{\ensuremath{\Varid{t}}}} should be evaluated on node {\textsmaller[.5]{\ensuremath{\Varid{i}}}}.  For simplicity, we assume
that the terms {\textsmaller[.5]{\ensuremath{\Varid{t}}}} in all locus specification sub-terms {\textsmaller[.5]{\ensuremath{\Varid{t}\;\MyConid{$\boldsymbol{@}$}\;\Varid{i}}}} are
\emph{closed}.  This is a reasonable
assumption, since a term where this does not hold can be transformed
into one where it does with roughly similar behaviour, using
e.g. lambda lifting~\cite{DBLP:conf/fpca/Johnsson85}
\iffullversion
\footnote{Transform every sub-term {\textsmaller[.5]{\ensuremath{\Varid{t}\;\MyConid{$\boldsymbol{@}$}\;\Varid{i}}}} to {\textsmaller[.5]{\ensuremath{\Varid{t'}\;\mathrel{=}\;((\Varid{λ}\;\Varid{fv}\;\Varid{t.}\;\Varid{t})\;\MyConid{$\boldsymbol{@}$}\;\Varid{i})\;(\Varid{fv}\;\Varid{t})}}}.
These have ``roughly similar behaviour'' in that the semantics of
{\textsmaller[.5]{\ensuremath{\Varid{t}}}} and {\textsmaller[.5]{\ensuremath{\Varid{t'}}}} are identical under the assumption that locus specifiers
do not change the meaning of a program.
}.
\fi
The {\textsmaller[.5]{\ensuremath{\MyConid{REMOTE}\;\Varid{c}\;\Varid{i}}}}
instruction will be used to start running a code fragment {\textsmaller[.5]{\ensuremath{\Varid{c}}}} on node
{\textsmaller[.5]{\ensuremath{\Varid{i}}}} in the network.  We also extend the {\textsmaller[.5]{\ensuremath{\Varid{compile'}}}} function to handle
the new term construct:
\iffullversion
\begin{hscode}\SaveRestoreHook
\column{B}{@{}>{\hspre}l<{\hspost}@{}}%
\column{E}{@{}>{\hspre}l<{\hspost}@{}}%
\>[B]{}\Varid{compile'}\;\mathbin{:}\;\Conid{Term}\;\Varid{→}\;\Conid{Code}\;\Varid{→}\;\Conid{Code}{}\<[E]%
\\
\>[B]{}\Varid{...}{}\<[E]%
\ColumnHook
\end{hscode}\resethooks
\removecodespace
\fi
\begin{hscode}\SaveRestoreHook
\column{B}{@{}>{\hspre}l<{\hspost}@{}}%
\column{20}{@{}>{\hspre}l<{\hspost}@{}}%
\column{23}{@{}>{\hspre}l<{\hspost}@{}}%
\column{E}{@{}>{\hspre}l<{\hspost}@{}}%
\>[B]{}\Varid{compile'}\;(\Varid{t}\;\MyConid{$\boldsymbol{@}$}\;\Varid{i})\;{}\<[20]%
\>[20]{}\Varid{c}\;{}\<[23]%
\>[23]{}\mathrel{=}\;\MyConid{REMOTE}\;(\Varid{compile'}\;\Varid{t}\;\MyConid{RET})\;\Varid{i}\;\MyConid{;}\;\Varid{c}{}\<[E]%
\ColumnHook
\end{hscode}\resethooks
Note that we reuse the {\textsmaller[.5]{\ensuremath{\MyConid{RET}}}} instruction to return from a remote
computation.

\iffullversion
Once again we assume that we are given a set {\textsmaller[.5]{\ensuremath{\Conid{Node}}}} with decidable
equality:
\begin{hscode}\SaveRestoreHook
\column{B}{@{}>{\hspre}l<{\hspost}@{}}%
\column{3}{@{}>{\hspre}l<{\hspost}@{}}%
\column{E}{@{}>{\hspre}l<{\hspost}@{}}%
\>[B]{}\Keyword{module}\;\Conid{DCESH}{}\<[E]%
\\
\>[B]{}\hsindent{3}{}\<[3]%
\>[3]{}(\Conid{Node}\;\mathbin{:}\;\star){}\<[E]%
\\
\>[B]{}\hsindent{3}{}\<[3]%
\>[3]{}(\Varid{\char95 ≟\char95 }\;\mathbin{:}\;(\Varid{n}\;\Varid{n'}\;\mathbin{:}\;\Conid{Node})\;\Varid{→}\;\Conid{Dec}\;(\Varid{n}\;\Varid{≡}\;\Varid{n'})){}\<[E]%
\\
\>[B]{}\hsindent{3}{}\<[3]%
\>[3]{}\Keyword{where}{}\<[E]%
\ColumnHook
\end{hscode}\resethooks
\fi

\iffullversion
The intended meaning of a remote pointer {\textsmaller[.5]{\ensuremath{\Conid{RPtr}}}} is that it is a
pointer located in the heap of the given node. We assume once again
that the set {\textsmaller[.5]{\ensuremath{\Conid{Node}}}} has decidable equality meaning that we can, for
instance, determine if an {\textsmaller[.5]{\ensuremath{\Conid{RPtr}}}} is remote or local.  This generalises
the \DCESHi{} machine, since we can now hold pointers
pointing to something in the heap of \emph{another} node's machine.
\begin{hscode}\SaveRestoreHook
\column{B}{@{}>{\hspre}l<{\hspost}@{}}%
\column{10}{@{}>{\hspre}l<{\hspost}@{}}%
\column{E}{@{}>{\hspre}l<{\hspost}@{}}%
\>[B]{}\Conid{RPtr}\;\mathrel{=}\;\Conid{Ptr}\;\Varid{×}\;\Conid{Node}{}\<[E]%
\\
\>[B]{}\Conid{ClosPtr}\;{}\<[10]%
\>[10]{}\mathrel{=}\;\Conid{RPtr}{}\<[E]%
\ColumnHook
\end{hscode}\resethooks
\fi
The definition of closures, values, environments and closure heaps are
the same as in the CESH machine, but using {\textsmaller[.5]{\ensuremath{\Conid{RPtr}}}} instead of {\textsmaller[.5]{\ensuremath{\Conid{Ptr}}}} for closure pointers.
\iffullversion
\begin{hscode}\SaveRestoreHook
\column{B}{@{}>{\hspre}l<{\hspost}@{}}%
\column{3}{@{}>{\hspre}l<{\hspost}@{}}%
\column{5}{@{}>{\hspre}l<{\hspost}@{}}%
\column{11}{@{}>{\hspre}l<{\hspost}@{}}%
\column{22}{@{}>{\hspre}l<{\hspost}@{}}%
\column{E}{@{}>{\hspre}l<{\hspost}@{}}%
\>[B]{}\Keyword{mutual}{}\<[E]%
\\
\>[B]{}\hsindent{3}{}\<[3]%
\>[3]{}\Conid{Closure}\;\mathrel{=}\;\Conid{Code}\;\Varid{×}\;\Conid{Env}{}\<[E]%
\\
\>[B]{}\hsindent{3}{}\<[3]%
\>[3]{}\Keyword{data}\;\Conid{Value}\;\mathbin{:}\;\star\;\Keyword{where}{}\<[E]%
\\
\>[3]{}\hsindent{2}{}\<[5]%
\>[5]{}\MyConid{nat}\;{}\<[11]%
\>[11]{}\mathbin{:}\;\Conid{ℕ}\;{}\<[22]%
\>[22]{}\Varid{→}\;\Conid{Value}{}\<[E]%
\\
\>[3]{}\hsindent{2}{}\<[5]%
\>[5]{}\MyConid{clos}\;{}\<[11]%
\>[11]{}\mathbin{:}\;\Conid{ClosPtr}\;{}\<[22]%
\>[22]{}\Varid{→}\;\Conid{Value}{}\<[E]%
\\
\>[B]{}\hsindent{3}{}\<[3]%
\>[3]{}\Conid{Env}\;\mathrel{=}\;\Conid{List}\;\Conid{Value}{}\<[E]%
\\
\>[B]{}\Conid{ClosHeap}\;\mathrel{=}\;\Conid{Heap}\;\Conid{Closure}{}\<[E]%
\ColumnHook
\end{hscode}\resethooks
\fi

The stack combines the functionality of the CES(H) machine, permitting
local continuations, with that of the \DCESHi{} machine, making it
possible for a stack to end with a continuation on another node. A
stack element is a value or a (local) continuation signified by the
{\textsmaller[.5]{\ensuremath{\MyConid{val}}}} and {\textsmaller[.5]{\ensuremath{\MyConid{cont}}}} constructors.  A stack ({\textsmaller[.5]{\ensuremath{\Conid{Stack}}}}) is a list of stack elements,
possibly ending with a (remote) pointer to a
\ifnotfullversion
continuation,
{\textsmaller[.5]{\ensuremath{\Conid{List}\;\Conid{StackElem}\;\Varid{×}\;\Conid{Maybe}\;\Conid{ContPtr}}}} (where {\textsmaller[.5]{\ensuremath{\Conid{ContPtr}\;\mathrel{=}\;\Conid{RPtr}}}}).
\else
continuation.
\begin{hscode}\SaveRestoreHook
\column{B}{@{}>{\hspre}l<{\hspost}@{}}%
\column{3}{@{}>{\hspre}l<{\hspost}@{}}%
\column{9}{@{}>{\hspre}l<{\hspost}@{}}%
\column{10}{@{}>{\hspre}l<{\hspost}@{}}%
\column{11}{@{}>{\hspre}l<{\hspost}@{}}%
\column{20}{@{}>{\hspre}l<{\hspost}@{}}%
\column{E}{@{}>{\hspre}l<{\hspost}@{}}%
\>[B]{}\Keyword{data}\;\Conid{StackElem}\;\mathbin{:}\;\star\;\Keyword{where}{}\<[E]%
\\
\>[B]{}\hsindent{3}{}\<[3]%
\>[3]{}\MyConid{val}\;{}\<[9]%
\>[9]{}\mathbin{:}\;\Conid{Value}\;{}\<[20]%
\>[20]{}\Varid{→}\;\Conid{StackElem}{}\<[E]%
\\
\>[B]{}\hsindent{3}{}\<[3]%
\>[3]{}\MyConid{cont}\;{}\<[9]%
\>[9]{}\mathbin{:}\;\Conid{Closure}\;{}\<[20]%
\>[20]{}\Varid{→}\;\Conid{StackElem}{}\<[E]%
\\
\>[B]{}\Conid{ContPtr}\;{}\<[10]%
\>[10]{}\mathrel{=}\;\Conid{RPtr}{}\<[E]%
\\
\>[B]{}\Conid{Stack}\;{}\<[11]%
\>[11]{}\mathrel{=}\;\Conid{List}\;\Conid{StackElem}\;\Varid{×}\;\Conid{Maybe}\;\Conid{ContPtr}{}\<[E]%
\\
\>[B]{}\Conid{ContHeap}\;{}\<[11]%
\>[11]{}\mathrel{=}\;\Conid{Heap}\;(\Conid{Closure}\;\Varid{×}\;\Conid{Stack}){}\<[E]%
\ColumnHook
\end{hscode}\resethooks
\fi
Threads and machines are defined like in the \DCESHi{} machine.
\iffullversion
\begin{hscode}\SaveRestoreHook
\column{B}{@{}>{\hspre}l<{\hspost}@{}}%
\column{10}{@{}>{\hspre}l<{\hspost}@{}}%
\column{E}{@{}>{\hspre}l<{\hspost}@{}}%
\>[B]{}\Conid{Thread}\;{}\<[10]%
\>[10]{}\mathrel{=}\;\Conid{Code}\;\Varid{×}\;\Conid{Env}\;\Varid{×}\;\Conid{Stack}{}\<[E]%
\\
\>[B]{}\Conid{Machine}\;{}\<[10]%
\>[10]{}\mathrel{=}\;\Conid{Maybe}\;\Conid{Thread}\;\Varid{×}\;\Conid{ClosHeap}\;\Varid{×}\;\Conid{ContHeap}{}\<[E]%
\ColumnHook
\end{hscode}\resethooks
\fi
The messages that \DCESHn{} can send are those of the \DCESHi{}
machine but using remote pointers instead of plain pointers, plus a
message for starting a remote computation, {\textsmaller[.5]{\ensuremath{\MyConid{REMOTE}\;\Varid{c}\;\Varid{i}\;\Varid{rptr}_{\Varid{cnt}}}}}.
\iffullversion
\begin{hscode}\SaveRestoreHook
\column{B}{@{}>{\hspre}l<{\hspost}@{}}%
\column{3}{@{}>{\hspre}l<{\hspost}@{}}%
\column{11}{@{}>{\hspre}l<{\hspost}@{}}%
\column{22}{@{}>{\hspre}l<{\hspost}@{}}%
\column{31}{@{}>{\hspre}l<{\hspost}@{}}%
\column{E}{@{}>{\hspre}l<{\hspost}@{}}%
\>[B]{}\Keyword{data}\;\Conid{Msg}\;\mathbin{:}\;\star\;\Keyword{where}{}\<[E]%
\\
\>[B]{}\hsindent{3}{}\<[3]%
\>[3]{}\MyConid{REMOTE}\;{}\<[11]%
\>[11]{}\mathbin{:}\;\Conid{Code}\;{}\<[22]%
\>[22]{}\Varid{→}\;\Conid{Node}\;{}\<[31]%
\>[31]{}\Varid{→}\;\Conid{ContPtr}\;\Varid{→}\;\Conid{Msg}{}\<[E]%
\\
\>[B]{}\hsindent{3}{}\<[3]%
\>[3]{}\MyConid{RET}\;{}\<[11]%
\>[11]{}\mathbin{:}\;\Conid{ContPtr}\;{}\<[22]%
\>[22]{}\Varid{→}\;\Conid{Value}\;{}\<[31]%
\>[31]{}\Varid{→}\;\Conid{Msg}{}\<[E]%
\\
\>[B]{}\hsindent{3}{}\<[3]%
\>[3]{}\MyConid{APPL}\;{}\<[11]%
\>[11]{}\mathbin{:}\;\Conid{ClosPtr}\;{}\<[22]%
\>[22]{}\Varid{→}\;\Conid{Value}\;{}\<[31]%
\>[31]{}\Varid{→}\;\Conid{ContPtr}\;\Varid{→}\;\Conid{Msg}{}\<[E]%
\ColumnHook
\end{hscode}\resethooks
\fi
Note that sending a {\textsmaller[.5]{\ensuremath{\MyConid{REMOTE}}}} message amounts to sending code in our
formalisation, which is something that we said that it would not
do. However, because no code is generated at run-time, every machine
can be ``pre-loaded'' with all the bytecode it needs, and the message
only needs to contain a \emph{reference} to a fragment of code.

\iffullversion
\begin{sidewaysfigure}
\else
\begin{figure*}[!t]
\fi
\centering
\begin{varwidth}[t]{\textwidth}
\iffullversion
\begin{hscode}\SaveRestoreHook
\column{B}{@{}>{\hspre}l<{\hspost}@{}}%
\column{E}{@{}>{\hspre}l<{\hspost}@{}}%
\>[B]{}\Keyword{data}\;\dummy\Varid{⊢}\dummy\xrightarrow{\dummy}\dummy\;(\Varid{i}\;\mathbin{:}\;\Conid{Node})\;\mathbin{:}\;\Conid{Machine}\;\Varid{→}\;\Conid{Tagged}\;\Conid{Msg}\;\Varid{→}\;\Conid{Machine}\;\Varid{→}\;\star\;\Keyword{where}{}\<[E]%
\ColumnHook
\end{hscode}\resethooks
\removecodespace
\fi
\iffullversion
\indentcolumn{3}
\rightaligncolumn{10}
\centeraligncolumn{101}
\savecolumns
\begin{hscode}\SaveRestoreHook
\column{B}{@{}>{\hspre}l<{\hspost}@{}}%
\column{3}{@{}>{\hspre}l<{\hspost}@{}}%
\column{5}{@{}>{\hspre}l<{\hspost}@{}}%
\column{10}{@{}>{\hspre}l<{\hspost}@{}}%
\column{19}{@{}>{\hspre}l<{\hspost}@{}}%
\column{101}{@{}>{\hspre}l<{\hspost}@{}}%
\column{170}{@{}>{\hspre}l<{\hspost}@{}}%
\column{E}{@{}>{\hspre}l<{\hspost}@{}}%
\>[3]{}\MyConid{VAR}\;{}\<[19]%
\>[19]{}\mathbin{:}\;\Varid{∀}\;\{\mskip1.5mu \Varid{n}\;\Varid{c}\;\Varid{e}\;\Varid{s}\;\Varid{v}\;\Varid{r}\;\Varid{h}_{\Varid{cl}}\;\Varid{h}_{\Varid{cnt}}\mskip1.5mu\}\;\Varid{→}\;\Varid{lookup}\;\Varid{n}\;\Varid{e}\;\Varid{≡}\;\MyConid{just}\;\Varid{v}\;\Varid{→}\;{}\<[E]%
\\
\>[3]{}\hsindent{2}{}\<[5]%
\>[5]{}\Varid{i}\;\Varid{⊢}\;{}\<[10]%
\>[10]{}(\MyConid{just}\;(\MyConid{VAR}\;\Varid{n}\;\MyConid{;}\;\Varid{c},\Varid{e},\Varid{s},\Varid{r}),\Varid{h}_{\Varid{cl}},\Varid{h}_{\Varid{cnt}})\;\;{}\<[101]%
\>[101]{}\xrightarrow{\MyConid{silent}}\;\;{}\<[170]%
\>[170]{}(\MyConid{just}\;(\Varid{c},\Varid{e},\MyConid{val}\;\Varid{v}\;\Varid{∷}\;\Varid{s},\Varid{r}),\Varid{h}_{\Varid{cl}},\Varid{h}_{\Varid{cnt}}){}\<[E]%
\\
\>[3]{}\MyConid{CLOS}\;{}\<[19]%
\>[19]{}\mathbin{:}\;\Varid{∀}\;\{\mskip1.5mu \Varid{c'}\;\Varid{c}\;\Varid{e}\;\Varid{s}\;\Varid{r}\;\Varid{h}_{\Varid{cl}}\;\Varid{h}_{\Varid{cnt}}\mskip1.5mu\}\;\Varid{→}\;\Keyword{let}\;(\Varid{h'}_{\Varid{cl}},\Varid{rptr}_{\Varid{cl}})\;\mathrel{=}\;\Varid{i}\;\Varid{⊢}\;\Varid{h}_{\Varid{cl}}\;\Varid{▸}\;(\Varid{c'},\Varid{e})\;\Keyword{in}{}\<[E]%
\\
\>[3]{}\hsindent{2}{}\<[5]%
\>[5]{}\Varid{i}\;\Varid{⊢}\;{}\<[10]%
\>[10]{}(\MyConid{just}\;(\MyConid{CLOS}\;\Varid{c'}\;\MyConid{;}\;\Varid{c},\Varid{e},\Varid{s},\Varid{r}),\Varid{h}_{\Varid{cl}},\Varid{h}_{\Varid{cnt}})\;\;{}\<[101]%
\>[101]{}\xrightarrow{\MyConid{silent}}\;\;{}\<[170]%
\>[170]{}(\MyConid{just}\;(\Varid{c},\Varid{e},\MyConid{val}\;(\MyConid{clos}\;\Varid{rptr}_{\Varid{cl}})\;\Varid{∷}\;\Varid{s},\Varid{r}),\Varid{h'}_{\Varid{cl}},\Varid{h}_{\Varid{cnt}}){}\<[E]%
\\
\>[3]{}\MyConid{APPL}\;{}\<[19]%
\>[19]{}\mathbin{:}\;\Varid{∀}\;\{\mskip1.5mu \Varid{c}\;\Varid{e}\;\Varid{v}\;\Varid{c'}\;\Varid{e'}\;\Varid{s}\;\Varid{r}\;\Varid{ptr}_{\Varid{cl}}\;\Varid{h}_{\Varid{cl}}\;\Varid{h}_{\Varid{cnt}}\mskip1.5mu\}\;\Varid{→}\;\Varid{h}_{\Varid{cl}}\;\mathbin{!}\;\Varid{ptr}_{\Varid{cl}}\;\Varid{≡}\;\MyConid{just}\;(\Varid{c'},\Varid{e'})\;\Varid{→}\;{}\<[E]%
\\
\>[3]{}\hsindent{2}{}\<[5]%
\>[5]{}\Varid{i}\;\Varid{⊢}\;{}\<[10]%
\>[10]{}(\MyConid{just}\;(\MyConid{APPL}\;\MyConid{;}\;\Varid{c},\Varid{e},\MyConid{val}\;\Varid{v}\;\Varid{∷}\;\MyConid{val}\;(\MyConid{clos}\;(\Varid{ptr}_{\Varid{cl}},\Varid{i}))\;\Varid{∷}\;\Varid{s},\Varid{r}),\Varid{h}_{\Varid{cl}},\Varid{h}_{\Varid{cnt}})\;\;{}\<[101]%
\>[101]{}\xrightarrow{\MyConid{silent}}\;\;{}\<[170]%
\>[170]{}(\MyConid{just}\;(\Varid{c'},\Varid{v}\;\Varid{∷}\;\Varid{e'},\MyConid{cont}\;(\Varid{c},\Varid{e})\;\Varid{∷}\;\Varid{s},\Varid{r}),\Varid{h}_{\Varid{cl}},\Varid{h}_{\Varid{cnt}}){}\<[E]%
\\
\>[3]{}\MyConid{RET}\;{}\<[19]%
\>[19]{}\mathbin{:}\;\Varid{∀}\;\{\mskip1.5mu \Varid{e}\;\Varid{v}\;\Varid{c}\;\Varid{e'}\;\Varid{s}\;\Varid{r}\;\Varid{h}_{\Varid{cl}}\;\Varid{h}_{\Varid{cnt}}\mskip1.5mu\}\;\Varid{→}\;{}\<[E]%
\\
\>[3]{}\hsindent{2}{}\<[5]%
\>[5]{}\Varid{i}\;\Varid{⊢}\;{}\<[10]%
\>[10]{}(\MyConid{just}\;(\MyConid{RET},\Varid{e},\MyConid{val}\;\Varid{v}\;\Varid{∷}\;\MyConid{cont}\;(\Varid{c},\Varid{e'})\;\Varid{∷}\;\Varid{s},\Varid{r}),\Varid{h}_{\Varid{cl}},\Varid{h}_{\Varid{cnt}})\;\;{}\<[101]%
\>[101]{}\xrightarrow{\MyConid{silent}}\;\;{}\<[170]%
\>[170]{}(\MyConid{just}\;(\Varid{c},\Varid{e'},\MyConid{val}\;\Varid{v}\;\Varid{∷}\;\Varid{s},\Varid{r}),\Varid{h}_{\Varid{cl}},\Varid{h}_{\Varid{cnt}}){}\<[E]%
\\
\>[3]{}\MyConid{LIT}\;{}\<[19]%
\>[19]{}\mathbin{:}\;\Varid{∀}\;\{\mskip1.5mu \Varid{n}\;\Varid{c}\;\Varid{e}\;\Varid{s}\;\Varid{r}\;\Varid{h}_{\Varid{cl}}\;\Varid{h}_{\Varid{cnt}}\mskip1.5mu\}\;\Varid{→}\;{}\<[E]%
\\
\>[3]{}\hsindent{2}{}\<[5]%
\>[5]{}\Varid{i}\;\Varid{⊢}\;{}\<[10]%
\>[10]{}(\MyConid{just}\;(\MyConid{LIT}\;\Varid{n}\;\MyConid{;}\;\Varid{c},\Varid{e},\Varid{s},\Varid{r}),\Varid{h}_{\Varid{cl}},\Varid{h}_{\Varid{cnt}})\;\;{}\<[101]%
\>[101]{}\xrightarrow{\MyConid{silent}}\;\;{}\<[170]%
\>[170]{}(\MyConid{just}\;(\Varid{c},\Varid{e},\MyConid{val}\;(\MyConid{nat}\;\Varid{n})\;\Varid{∷}\;\Varid{s},\Varid{r}),\Varid{h}_{\Varid{cl}},\Varid{h}_{\Varid{cnt}}){}\<[E]%
\\
\>[3]{}\MyConid{OP}\;{}\<[19]%
\>[19]{}\mathbin{:}\;\Varid{∀}\;\{\mskip1.5mu \Varid{f}\;\Varid{c}\;\Varid{e}\;\Varid{n₁}\;\Varid{n₂}\;\Varid{s}\;\Varid{r}\;\Varid{h}_{\Varid{cl}}\;\Varid{h}_{\Varid{cnt}}\mskip1.5mu\}\;\Varid{→}\;{}\<[E]%
\\
\>[3]{}\hsindent{2}{}\<[5]%
\>[5]{}\Varid{i}\;\Varid{⊢}\;{}\<[10]%
\>[10]{}(\MyConid{just}\;(\MyConid{OP}\;\Varid{f}\;\MyConid{;}\;\Varid{c},\Varid{e},\MyConid{val}\;(\MyConid{nat}\;\Varid{n₁})\;\Varid{∷}\;\MyConid{val}\;(\MyConid{nat}\;\Varid{n₂})\;\Varid{∷}\;\Varid{s},\Varid{r}),\Varid{h}_{\Varid{cl}},\Varid{h}_{\Varid{cnt}})\;\;{}\<[101]%
\>[101]{}\xrightarrow{\MyConid{silent}}\;\;{}\<[170]%
\>[170]{}(\MyConid{just}\;(\Varid{c},\Varid{e},\MyConid{val}\;(\MyConid{nat}\;(\Varid{f}\;\Varid{n₁}\;\Varid{n₂}))\;\Varid{∷}\;\Varid{s},\Varid{r}),\Varid{h}_{\Varid{cl}},\Varid{h}_{\Varid{cnt}}){}\<[E]%
\\
\>[3]{}\MyConid{COND-0}\;{}\<[19]%
\>[19]{}\mathbin{:}\;\Varid{∀}\;\{\mskip1.5mu \Varid{c}\;\Varid{c'}\;\Varid{e}\;\Varid{s}\;\Varid{r}\;\Varid{h}_{\Varid{cl}}\;\Varid{h}_{\Varid{cnt}}\mskip1.5mu\}\;\Varid{→}\;{}\<[E]%
\\
\>[3]{}\hsindent{2}{}\<[5]%
\>[5]{}\Varid{i}\;\Varid{⊢}\;{}\<[10]%
\>[10]{}(\MyConid{just}\;(\MyConid{COND}\;\Varid{c}\;\Varid{c'},\Varid{e},\MyConid{val}\;(\MyConid{nat}\;\Varid{0})\;\Varid{∷}\;\Varid{s},\Varid{r}),\Varid{h}_{\Varid{cl}},\Varid{h}_{\Varid{cnt}})\;\;{}\<[101]%
\>[101]{}\xrightarrow{\MyConid{silent}}\;\;{}\<[170]%
\>[170]{}(\MyConid{just}\;(\Varid{c},\Varid{e},\Varid{s},\Varid{r}),\Varid{h}_{\Varid{cl}},\Varid{h}_{\Varid{cnt}}){}\<[E]%
\\
\>[3]{}\MyConid{COND-1+n}\;{}\<[19]%
\>[19]{}\mathbin{:}\;\Varid{∀}\;\{\mskip1.5mu \Varid{c}\;\Varid{c'}\;\Varid{e}\;\Varid{n}\;\Varid{s}\;\Varid{r}\;\Varid{h}_{\Varid{cl}}\;\Varid{h}_{\Varid{cnt}}\mskip1.5mu\}\;\Varid{→}\;{}\<[E]%
\\
\>[3]{}\hsindent{2}{}\<[5]%
\>[5]{}\Varid{i}\;\Varid{⊢}\;{}\<[10]%
\>[10]{}(\MyConid{just}\;(\MyConid{COND}\;\Varid{c}\;\Varid{c'},\Varid{e},\MyConid{val}\;(\MyConid{nat}\;(1+\;\Varid{n}))\;\Varid{∷}\;\Varid{s},\Varid{r}),\Varid{h}_{\Varid{cl}},\Varid{h}_{\Varid{cnt}})\;\;{}\<[101]%
\>[101]{}\xrightarrow{\MyConid{silent}}\;\;{}\<[170]%
\>[170]{}(\MyConid{just}\;(\Varid{c'},\Varid{e},\Varid{s},\Varid{r}),\Varid{h}_{\Varid{cl}},\Varid{h}_{\Varid{cnt}}){}\<[E]%
\ColumnHook
\end{hscode}\resethooks
\removecodespace
\restorecolumns
\fi
\indentcolumn{3}
\rightaligncolumn{10}
\centeraligncolumn{101}
\begin{hscode}\SaveRestoreHook
\column{B}{@{}>{\hspre}l<{\hspost}@{}}%
\column{3}{@{}>{\hspre}l<{\hspost}@{}}%
\column{5}{@{}>{\hspre}l<{\hspost}@{}}%
\column{10}{@{}>{\hspre}l<{\hspost}@{}}%
\column{19}{@{}>{\hspre}l<{\hspost}@{}}%
\column{101}{@{}>{\hspre}l<{\hspost}@{}}%
\column{170}{@{}>{\hspre}l<{\hspost}@{}}%
\column{E}{@{}>{\hspre}l<{\hspost}@{}}%
\>[3]{}\MyConid{REMOTE-send}\;{}\<[19]%
\>[19]{}\mathbin{:}\;\Varid{∀}\;\{\mskip1.5mu \Varid{c'}\;\Varid{i'}\;\Varid{c}\;\Varid{e}\;\Varid{s}\;\Varid{r}\;\Varid{h}_{\Varid{cl}}\;\Varid{h}_{\Varid{cnt}}\mskip1.5mu\}\;\Varid{→}\;\Keyword{let}\;(\Varid{h'}_{\Varid{cnt}},\Varid{rptr})\;\mathrel{=}\;\Varid{i}\;\Varid{⊢}\;\Varid{h}_{\Varid{cnt}}\;\Varid{▸}\;((\Varid{c},\Varid{e}),\Varid{s},\Varid{r})\;\Keyword{in}{}\<[E]%
\\
\>[3]{}\hsindent{2}{}\<[5]%
\>[5]{}\Varid{i}\;\Varid{⊢}\;{}\<[10]%
\>[10]{}(\MyConid{just}\;(\MyConid{REMOTE}\;\Varid{c'}\;\Varid{i'}\;\MyConid{;}\;\Varid{c},\Varid{e},\Varid{s},\Varid{r}),\Varid{h}_{\Varid{cl}},\Varid{h}_{\Varid{cnt}})\;\;{}\<[101]%
\>[101]{}\xrightarrow{\MyConid{send}\;(\MyConid{REMOTE}\;\Varid{c'}\;\Varid{i'}\;\Varid{rptr})}\;\;{}\<[170]%
\>[170]{}(\MyConid{nothing},\Varid{h}_{\Varid{cl}},\Varid{h'}_{\Varid{cnt}}){}\<[E]%
\\
\>[3]{}\MyConid{REMOTE-receive}\;{}\<[19]%
\>[19]{}\mathbin{:}\;\Varid{∀}\;\{\mskip1.5mu \Varid{h}_{\Varid{cl}}\;\Varid{h}_{\Varid{cnt}}\;\Varid{c}\;\Varid{rptr}_{\Varid{cnt}}\mskip1.5mu\}\;\Varid{→}\;{}\<[E]%
\\
\>[3]{}\hsindent{2}{}\<[5]%
\>[5]{}\Varid{i}\;\Varid{⊢}\;{}\<[10]%
\>[10]{}(\MyConid{nothing},\Varid{h}_{\Varid{cl}},\Varid{h}_{\Varid{cnt}})\;\;{}\<[101]%
\>[101]{}\xrightarrow{\MyConid{receive}\;(\MyConid{REMOTE}\;\Varid{c}\;\Varid{i}\;\Varid{rptr}_{\Varid{cnt}})}\;\;{}\<[170]%
\>[170]{}(\MyConid{just}\;(\Varid{c},[\mskip1.5mu \mskip1.5mu],[\mskip1.5mu \mskip1.5mu],\MyConid{just}\;\Varid{rptr}_{\Varid{cnt}}),\Varid{h}_{\Varid{cl}},\Varid{h}_{\Varid{cnt}}){}\<[E]%
\\
\>[3]{}\MyConid{APPL-send}\;{}\<[19]%
\>[19]{}\mathbin{:}\;\Varid{∀}\;\{\mskip1.5mu \Varid{c}\;\Varid{e}\;\Varid{v}\;\Varid{ptr}_{\Varid{cl}}\;\Varid{j}\;\Varid{s}\;\Varid{r}\;\Varid{h}_{\Varid{cl}}\;\Varid{h}_{\Varid{cnt}}\mskip1.5mu\}\;\Varid{→}\;\Varid{i}\;\Varid{≢}\;\Varid{j}\;\Varid{→}\;\Keyword{let}\;(\Varid{h'}_{\Varid{cnt}},\Varid{rptr}_{\Varid{cnt}})\;\mathrel{=}\;\Varid{i}\;\Varid{⊢}\;\Varid{h}_{\Varid{cnt}}\;\Varid{▸}\;((\Varid{c},\Varid{e}),\Varid{s},\Varid{r})\;\Keyword{in}{}\<[E]%
\\
\>[3]{}\hsindent{2}{}\<[5]%
\>[5]{}\Varid{i}\;\Varid{⊢}\;{}\<[10]%
\>[10]{}(\MyConid{just}\;(\MyConid{APPL}\;\MyConid{;}\;\Varid{c},\Varid{e},\MyConid{val}\;\Varid{v}\;\Varid{∷}\;\MyConid{val}\;(\MyConid{clos}\;(\Varid{ptr}_{\Varid{cl}},\Varid{j}))\;\Varid{∷}\;\Varid{s},\Varid{r}),\Varid{h}_{\Varid{cl}},\Varid{h}_{\Varid{cnt}})\;\;\xrightarrow{\MyConid{send}\;(\MyConid{APPL}\;(\Varid{ptr}_{\Varid{cl}},\Varid{j})\;\Varid{v}\;\Varid{rptr}_{\Varid{cnt}})}\;\;(\MyConid{nothing},\Varid{h}_{\Varid{cl}},\Varid{h'}_{\Varid{cnt}}){}\<[E]%
\\
\>[3]{}\MyConid{APPL-receive}\;{}\<[19]%
\>[19]{}\mathbin{:}\;\Varid{∀}\;\{\mskip1.5mu \Varid{h}_{\Varid{cl}}\;\Varid{h}_{\Varid{cnt}}\;\Varid{ptr}_{\Varid{cl}}\;\Varid{v}\;\Varid{rptr}_{\Varid{cnt}}\;\Varid{c}\;\Varid{e}\mskip1.5mu\}\;\Varid{→}\;\Varid{h}_{\Varid{cl}}\;\mathbin{!}\;\Varid{ptr}_{\Varid{cl}}\;\Varid{≡}\;\MyConid{just}\;(\Varid{c},\Varid{e})\;\Varid{→}\;{}\<[E]%
\\
\>[3]{}\hsindent{2}{}\<[5]%
\>[5]{}\Varid{i}\;\Varid{⊢}\;{}\<[10]%
\>[10]{}(\MyConid{nothing},\Varid{h}_{\Varid{cl}},\Varid{h}_{\Varid{cnt}})\;\;{}\<[101]%
\>[101]{}\xrightarrow{\MyConid{receive}\;(\MyConid{APPL}\;(\Varid{ptr}_{\Varid{cl}},\Varid{i})\;\Varid{v}\;\Varid{rptr}_{\Varid{cnt}})}\;\;{}\<[170]%
\>[170]{}(\MyConid{just}\;(\Varid{c},\Varid{v}\;\Varid{∷}\;\Varid{e},[\mskip1.5mu \mskip1.5mu],\MyConid{just}\;\Varid{rptr}_{\Varid{cnt}}),\Varid{h}_{\Varid{cl}},\Varid{h}_{\Varid{cnt}}){}\<[E]%
\\
\>[3]{}\MyConid{RET-send}\;{}\<[19]%
\>[19]{}\mathbin{:}\;\Varid{∀}\;\{\mskip1.5mu \Varid{e}\;\Varid{v}\;\Varid{rptr}_{\Varid{cnt}}\;\Varid{h}_{\Varid{cl}}\;\Varid{h}_{\Varid{cnt}}\mskip1.5mu\}\;\Varid{→}\;{}\<[E]%
\\
\>[3]{}\hsindent{2}{}\<[5]%
\>[5]{}\Varid{i}\;\Varid{⊢}\;{}\<[10]%
\>[10]{}(\MyConid{just}\;(\MyConid{RET},\Varid{e},\MyConid{val}\;\Varid{v}\;\Varid{∷}\;[\mskip1.5mu \mskip1.5mu],\MyConid{just}\;\Varid{rptr}_{\Varid{cnt}}),\Varid{h}_{\Varid{cl}},\Varid{h}_{\Varid{cnt}})\;\;{}\<[101]%
\>[101]{}\xrightarrow{\MyConid{send}\;(\MyConid{RET}\;\Varid{rptr}_{\Varid{cnt}}\;\Varid{v})}\;\;{}\<[170]%
\>[170]{}(\MyConid{nothing},\Varid{h}_{\Varid{cl}},\Varid{h}_{\Varid{cnt}}){}\<[E]%
\\
\>[3]{}\MyConid{RET-receive}\;{}\<[19]%
\>[19]{}\mathbin{:}\;\Varid{∀}\;\{\mskip1.5mu \Varid{h}_{\Varid{cl}}\;\Varid{h}_{\Varid{cnt}}\;\Varid{ptr}_{\Varid{cnt}}\;\Varid{v}\;\Varid{c}\;\Varid{e}\;\Varid{s}\;\Varid{r}\mskip1.5mu\}\;\Varid{→}\;\Varid{h}_{\Varid{cnt}}\;\mathbin{!}\;\Varid{ptr}_{\Varid{cnt}}\;\Varid{≡}\;\MyConid{just}\;((\Varid{c},\Varid{e}),\Varid{s},\Varid{r})\;\Varid{→}\;{}\<[E]%
\\
\>[3]{}\hsindent{2}{}\<[5]%
\>[5]{}\Varid{i}\;\Varid{⊢}\;{}\<[10]%
\>[10]{}(\MyConid{nothing},\Varid{h}_{\Varid{cl}},\Varid{h}_{\Varid{cnt}})\;\;{}\<[101]%
\>[101]{}\xrightarrow{\MyConid{receive}\;(\MyConid{RET}\;(\Varid{ptr}_{\Varid{cnt}},\Varid{i})\;\Varid{v})}\;\;{}\<[170]%
\>[170]{}(\MyConid{just}\;(\Varid{c},\Varid{e},\MyConid{val}\;\Varid{v}\;\Varid{∷}\;\Varid{s},\Varid{r}),\Varid{h}_{\Varid{cl}},\Varid{h}_{\Varid{cnt}}){}\<[E]%
\ColumnHook
\end{hscode}\resethooks
\end{varwidth}
\iffullversion
\caption{The definition of the transition relation of the \DCESHn{} machine.}
\fi
\label{figure:DCESH-step}
\iffullversion
\end{sidewaysfigure}
\else
\end{figure*}
\fi
Fig.~\ref{figure:DCESH-step} defines the transition relation of
the \DCESHn{} machine, written {\textsmaller[.5]{\ensuremath{\Varid{i}\;\Varid{⊢}\;\Varid{m}\;\;\xrightarrow{\Varid{tmsg}}\;\;\Varid{m'}}}} for
a node identifier {\textsmaller[.5]{\ensuremath{\Varid{i}}}}, a tagged message {\textsmaller[.5]{\ensuremath{\Varid{tmsg}}}} and machine
configurations {\textsmaller[.5]{\ensuremath{\Varid{m}}}} and {\textsmaller[.5]{\ensuremath{\Varid{m'}}}}. The parameter {\textsmaller[.5]{\ensuremath{\Varid{i}}}} is taken to be the
identifier of the node on which the transition is taking place.
\iffullversion
Most instructions are similar to those of the CESH machine
but adapted to this new setting, using remote pointers.
The following function is used to allocate a pointer in a heap on
a node {\textsmaller[.5]{\ensuremath{\Varid{i}}}}, yielding a new heap and a remote pointer (pointing to the
node {\textsmaller[.5]{\ensuremath{\Varid{i}}}}):
\begin{hscode}\SaveRestoreHook
\column{B}{@{}>{\hspre}l<{\hspost}@{}}%
\column{29}{@{}>{\hspre}l<{\hspost}@{}}%
\column{E}{@{}>{\hspre}l<{\hspost}@{}}%
\>[B]{}\Varid{\char95 ⊢\char95 ▸\char95 }\;\mathbin{:}\;\{\mskip1.5mu \Conid{A}\;\mathbin{:}\;\star\mskip1.5mu\}\;\Varid{→}\;\Conid{Node}\;\Varid{→}\;{}\<[29]%
\>[29]{}\Conid{Heap}\;\Conid{A}\;\Varid{→}\;\Conid{A}\;\Varid{→}\;\Conid{Heap}\;\Conid{A}\;\Varid{×}\;\Conid{RPtr}{}\<[E]%
\\
\>[B]{}\Varid{i}\;\Varid{⊢}\;\Varid{h}\;\Varid{▸}\;\Varid{x}\;\mathrel{=}\;\Keyword{let}\;(\Varid{h'},\Varid{ptr})\;\mathrel{=}\;\Varid{h}\;\Varid{▸}\;\Varid{x}\;\Keyword{in}\;\Varid{h'},(\Varid{ptr},\Varid{i}){}\<[E]%
\ColumnHook
\end{hscode}\resethooks
\fi
\iffullversion
When an application occurs and the closure pointer is on the
current node, {\textsmaller[.5]{\ensuremath{\Varid{i}}}}, the machine dereferences the pointer and enters it locally.
If there is a local continuation on the stack and the machine is to
run the return instruction, it also works just like the original CES
machine.
\fi
When starting a remote computation, the machine allocates a
continuation in the heap and sends a message containing the code and
continuation pointer to the remote node in question.
Afterwards the current thread is stopped. On the receiving end of such
a communication, a new thread is started, placing the continuation pointer
at the bottom of the stack for the later return to the caller node.
To run the apply instruction when the function closure is remote,
i.e. its location is \emph{not} equal to the current node, the machine
sends a message containing the closure pointer, argument value, and
continuation, like in the \DCESHi{} machine.
On the other end of such a communication, the machine dereferences the
pointer and enters the closure with the received value. The
bottom remote continuation pointer is set to the received continuation pointer.
After either a remote invocation or a remote application, the machine can return
if it has produced a value on the stack and has a remote continuation
at the bottom of the stack. To do this, a message containing the continuation pointer
and the return value is sent to the location of the continuation pointer.
When receiving a return message, the continuation pointer is dereferenced
and entered with the received value.

Now that we have defined the transition relation for machines we
instantiate the {\textsmaller[.5]{\ensuremath{\Conid{Network}}}} module with the {\textsmaller[.5]{\ensuremath{\Varid{⟶Machine}}}} relation.
\iffullversion
\begin{hscode}\SaveRestoreHook
\column{B}{@{}>{\hspre}l<{\hspost}@{}}%
\column{E}{@{}>{\hspre}l<{\hspost}@{}}%
\>[B]{}\Keyword{open}\;\Keyword{import}\;\Conid{Network}\;\Conid{Node}\;\Varid{\char95 ≟\char95 }\;\Varid{⟶Machine}\;\Keyword{public}{}\<[E]%
\ColumnHook
\end{hscode}\resethooks
\fi
From here on {\textsmaller[.5]{\ensuremath{\Conid{SyncNetwork}}}} and {\textsmaller[.5]{\ensuremath{\Conid{AsyncNetwork}}}} and their
transition relations will thus refer to the instantiated versions.

An initial network configuration, given a code fragment {\textsmaller[.5]{\ensuremath{\Varid{c}}}} and
a node identifier {\textsmaller[.5]{\ensuremath{\Varid{i}}}}, is a network where only node {\textsmaller[.5]{\ensuremath{\Varid{i}}}} is
active, ready to run the code fragment:
\begin{hscode}\SaveRestoreHook
\column{B}{@{}>{\hspre}l<{\hspost}@{}}%
\column{3}{@{}>{\hspre}l<{\hspost}@{}}%
\column{E}{@{}>{\hspre}l<{\hspost}@{}}%
\>[B]{}\Varid{initial-network}_{\Varid{Sync}}\;\mathbin{:}\;\Conid{Code}\;\Varid{→}\;\Conid{Node}\;\Varid{→}\;\Conid{SyncNetwork}{}\<[E]%
\\
\>[B]{}\Varid{initial-network}_{\Varid{Sync}}\;\Varid{c}\;\Varid{i}\;\mathrel{=}\;\Varid{update}\;(\Varid{λ}\;\Varid{i'}\;\Varid{→}\;(\MyConid{nothing},\Varid{∅},\Varid{∅}))\;{}\<[E]%
\\
\>[B]{}\hsindent{3}{}\<[3]%
\>[3]{}\Varid{i}\;(\MyConid{just}\;(\Varid{c},[\mskip1.5mu \mskip1.5mu],[\mskip1.5mu \mskip1.5mu],\MyConid{nothing}),\Varid{∅},\Varid{∅}){}\<[E]%
\ColumnHook
\end{hscode}\resethooks
An initial asynchronous network configuration is one where there are
no messages in message list:
\iffullversion
\begin{hscode}\SaveRestoreHook
\column{B}{@{}>{\hspre}l<{\hspost}@{}}%
\column{26}{@{}>{\hspre}l<{\hspost}@{}}%
\column{E}{@{}>{\hspre}l<{\hspost}@{}}%
\>[B]{}\Varid{initial-network}_{\Varid{Async}}\;\mathbin{:}\;{}\<[26]%
\>[26]{}\Conid{Code}\;\Varid{→}\;\Conid{Node}\;\Varid{→}\;{}\<[E]%
\\
\>[26]{}\Conid{AsyncNetwork}{}\<[E]%
\\
\>[B]{}\Varid{initial-network}_{\Varid{Async}}\;\Varid{c}\;\Varid{i}\;\mathrel{=}\;\Varid{initial-network}_{\Varid{Sync}}\;\Varid{c}\;\Varid{i},[\mskip1.5mu \mskip1.5mu]{}\<[E]%
\ColumnHook
\end{hscode}\resethooks
\fi

\iffullversion
\begin{lemma}
The communication that the local step relation enables is
\emph{point-to-point}: if two nodes can receive the same message, then
they are the same:
\begin{hscode}\SaveRestoreHook
\column{B}{@{}>{\hspre}l<{\hspost}@{}}%
\column{3}{@{}>{\hspre}l<{\hspost}@{}}%
\column{17}{@{}>{\hspre}l<{\hspost}@{}}%
\column{E}{@{}>{\hspre}l<{\hspost}@{}}%
\>[B]{}\Varid{point-to-point}\;{}\<[17]%
\>[17]{}\mathbin{:}\;{}\<[E]%
\\
\>[B]{}\hsindent{3}{}\<[3]%
\>[3]{}\Varid{∀}\;\Varid{i₁}\;\Varid{i₂}\;(\Varid{ms}\;\mathbin{:}\;\Conid{SyncNetwork})\;\Varid{msg}\;\{\mskip1.5mu \Varid{m₁}\;\Varid{m₂}\mskip1.5mu\}\;\Varid{→}\;{}\<[E]%
\\
\>[B]{}\hsindent{3}{}\<[3]%
\>[3]{}\Varid{i₁}\;\Varid{⊢}\;\Varid{ms}\;\Varid{i₁}\;\;\xrightarrow{\MyConid{receive}\;\Varid{msg}}\;\;\Varid{m₁}\;\Varid{→}\;{}\<[E]%
\\
\>[B]{}\hsindent{3}{}\<[3]%
\>[3]{}\Varid{i₂}\;\Varid{⊢}\;\Varid{ms}\;\Varid{i₂}\;\;\xrightarrow{\MyConid{receive}\;\Varid{msg}}\;\;\Varid{m₂}\;\Varid{→}\;{}\<[E]%
\\
\>[B]{}\hsindent{3}{}\<[3]%
\>[3]{}\Varid{i₁}\;\Varid{≡}\;\Varid{i₂}{}\<[E]%
\ColumnHook
\end{hscode}\resethooks
\end{lemma}
\fi
We call a machine \emph{inactive} if its thread is not running, i.e. it is equal to {\textsmaller[.5]{\ensuremath{\MyConid{nothing}}}}.
\iffullversion
\begin{hscode}\SaveRestoreHook
\column{B}{@{}>{\hspre}l<{\hspost}@{}}%
\column{E}{@{}>{\hspre}l<{\hspost}@{}}%
\>[B]{}\Varid{inactive}\;\mathbin{:}\;\Conid{Machine}\;\Varid{→}\;\star{}\<[E]%
\\
\>[B]{}\Varid{inactive}\;(\Varid{t},\anonymous )\;\mathrel{=}\;\Varid{t}\;\Varid{≡}\;\MyConid{nothing}{}\<[E]%
\ColumnHook
\end{hscode}\resethooks
\fi

\begin{lemma}[{\textsmaller[.5]{\ensuremath{\Varid{determinism}_{\Varid{Sync}}}}}]
If all nodes in a synchronous network except one are inactive, then
the next step is deterministic.
\iffullversion
\begin{hscode}\SaveRestoreHook
\column{B}{@{}>{\hspre}l<{\hspost}@{}}%
\column{3}{@{}>{\hspre}l<{\hspost}@{}}%
\column{E}{@{}>{\hspre}l<{\hspost}@{}}%
\>[B]{}\Varid{determinism}_{\Varid{Sync}}\;\mathbin{:}\;\Varid{∀}\;\Varid{nodes}\;\Varid{i}\;\Varid{→}\;{}\<[E]%
\\
\>[B]{}\hsindent{3}{}\<[3]%
\>[3]{}\Varid{all}\;\Varid{nodes}\;\Varid{except}\;\Varid{i}\;\Varid{are}\;\Varid{inactive}\;\Varid{→}\;{\dummy\xrightarrow[\Varid{Sync}]{}\dummy}\;\Varid{is-deterministic-at}\;\Varid{nodes}{}\<[E]%
\ColumnHook
\end{hscode}\resethooks
where
\begin{hscode}\SaveRestoreHook
\column{B}{@{}>{\hspre}l<{\hspost}@{}}%
\column{E}{@{}>{\hspre}l<{\hspost}@{}}%
\>[B]{}\Varid{all\char95 except\char95 are\char95 }\;\mathbin{:}\;\{\mskip1.5mu \Conid{A}\;\Conid{B}\;\mathbin{:}\;\star\mskip1.5mu\}\;\Varid{→}\;(\Conid{A}\;\Varid{→}\;\Conid{B})\;\Varid{→}\;\Conid{A}\;\Varid{→}\;(\Conid{B}\;\Varid{→}\;\star)\;\Varid{→}\;\star{}\<[E]%
\\
\>[B]{}\Varid{all}\;\Varid{f}\;\Varid{except}\;\Varid{i}\;\Varid{are}\;\Conid{P}\;\mathrel{=}\;\Varid{∀}\;\Varid{i'}\;\Varid{→}\;\Varid{i'}\;\Varid{≢}\;\Varid{i}\;\Varid{→}\;\Conid{P}\;(\Varid{f}\;\Varid{i'}){}\<[E]%
\ColumnHook
\end{hscode}\resethooks
\fi
\end{lemma}

\begin{lemma}[{\textsmaller[.5]{\ensuremath{{{\xrightarrow[\Varid{Async}]{}\Varid{⁺}}\Varid{-to-}\xrightarrow[\Varid{Sync}]{}\Varid{⁺}}}}}]
\ifnotfullversion
For a synchronous network {\textsmaller[.5]{\ensuremath{\Varid{nodes}}}}, if all nodes except one are inactive,
and {\textsmaller[.5]{\ensuremath{(\Varid{nodes},[\mskip1.5mu \mskip1.5mu])\;{\xrightarrow[\Varid{Async}]{}\Varid{⁺}}\;(\Varid{nodes'},[\mskip1.5mu \mskip1.5mu])}}}, then {\textsmaller[.5]{\ensuremath{\Varid{nodes}\;\xrightarrow[\Varid{Sync}]{}\Varid{⁺}\;\Varid{nodes'}}}}.
\else
If all nodes in a synchronous network except one are inactive and the network
takes one or more steps asynchronously from and to configurations
without any messages in the air, then that transition can also be done
synchronously.
\begin{hscode}\SaveRestoreHook
\column{B}{@{}>{\hspre}l<{\hspost}@{}}%
\column{3}{@{}>{\hspre}l<{\hspost}@{}}%
\column{20}{@{}>{\hspre}l<{\hspost}@{}}%
\column{39}{@{}>{\hspre}l<{\hspost}@{}}%
\column{76}{@{}>{\hspre}l<{\hspost}@{}}%
\column{E}{@{}>{\hspre}l<{\hspost}@{}}%
\>[B]{}{{\xrightarrow[\Varid{Async}]{}\Varid{⁺}}\Varid{-to-}\xrightarrow[\Varid{Sync}]{}\Varid{⁺}}\;{}\<[20]%
\>[20]{}\mathbin{:}\;\Varid{∀}\;\{\mskip1.5mu \Varid{nodes}\;\Varid{nodes'}\mskip1.5mu\}\;\Varid{i}\;\Varid{→}\;\Varid{all}\;\Varid{nodes}\;\Varid{except}\;\Varid{i}\;\Varid{are}\;\Varid{inactive}\;{}\<[76]%
\>[76]{}\Varid{→}\;{}\<[E]%
\\
\>[B]{}\hsindent{3}{}\<[3]%
\>[3]{}(\Varid{nodes},[\mskip1.5mu \mskip1.5mu])\;{\xrightarrow[\Varid{Async}]{}\Varid{⁺}}\;(\Varid{nodes'},[\mskip1.5mu \mskip1.5mu])\;{}\<[39]%
\>[39]{}\Varid{→}\;\Varid{nodes}\;\xrightarrow[\Varid{Sync}]{}\Varid{⁺}\;\Varid{nodes'}{}\<[E]%
\ColumnHook
\end{hscode}\resethooks
\fi
\end{lemma}
This is a key result because it means that it does not matter whether
we choose to look at synchronous or asynchronous networks for single
threaded computations.
With this result in place, we will from now on focus on the simpler
synchronous networks.

We define what it means for a synchronous \DCESHn{} network {\textsmaller[.5]{\ensuremath{\Varid{nodes}}}} to
\emph{terminate with a value {\textsmaller[.5]{\ensuremath{\Varid{v}}}}} ({\textsmaller[.5]{\ensuremath{\Varid{nodes}\;\downarrow_{\Varid{Sync}}\;\Varid{v}}}}), \emph{terminate}
({\textsmaller[.5]{\ensuremath{\Varid{nodes}\;\downarrow_{\Varid{Sync}}}}}), and \emph{diverge} ({\textsmaller[.5]{\ensuremath{\Varid{nodes}\;\uparrow_{\Varid{Sync}}}}}).  A network
terminates with a value {\textsmaller[.5]{\ensuremath{\Varid{v}}}} if it can step to a network where only one
node is active, and that node has reached the {\textsmaller[.5]{\ensuremath{\MyConid{END}}}} instruction with
the value {\textsmaller[.5]{\ensuremath{\Varid{v}}}} on top of its stack. The other definitions are analogous to
those of the CES(H) machine.
\iffullversion
\begin{hscode}\SaveRestoreHook
\column{B}{@{}>{\hspre}l<{\hspost}@{}}%
\column{E}{@{}>{\hspre}l<{\hspost}@{}}%
\>[B]{}{\dummy\downarrow_{\Varid{Sync}}\dummy}\;\mathbin{:}\;\Conid{SyncNetwork}\;\Varid{→}\;\Conid{Value}\;\Varid{→}\;\star{}\<[E]%
\ColumnHook
\end{hscode}\resethooks
\removecodespace
\fi
\begin{hscode}\SaveRestoreHook
\column{B}{@{}>{\hspre}l<{\hspost}@{}}%
\column{3}{@{}>{\hspre}l<{\hspost}@{}}%
\column{12}{@{}>{\hspre}l<{\hspost}@{}}%
\column{E}{@{}>{\hspre}l<{\hspost}@{}}%
\>[B]{}\Varid{nodes}\;\downarrow_{\Varid{Sync}}\;\Varid{v}\;\mathrel{=}\;\Varid{∃}\;\Varid{λ}\;\Varid{nodes'}\;\Varid{→}\;\Varid{nodes}\;{\xrightarrow[\Varid{Sync}]{}^*}\;\Varid{nodes'}\;\Varid{×}\;{}\<[E]%
\\
\>[B]{}\hsindent{3}{}\<[3]%
\>[3]{}\Varid{∃}\;\Varid{λ}\;\Varid{i}\;\Varid{→}\;{}\<[12]%
\>[12]{}\Varid{all}\;\Varid{nodes'}\;\Varid{except}\;\Varid{i}\;\Varid{are}\;\Varid{inactive}\;\Varid{×}\;\Varid{∃}\;\Varid{λ}\;\Varid{heaps}\;\Varid{→}\;{}\<[E]%
\\
\>[B]{}\hsindent{3}{}\<[3]%
\>[3]{}\Varid{nodes'}\;\Varid{i}\;\Varid{≡}\;(\MyConid{just}\;(\MyConid{END},[\mskip1.5mu \mskip1.5mu],\MyConid{val}\;\Varid{v}\;\Varid{∷}\;[\mskip1.5mu \mskip1.5mu],\MyConid{nothing}),\Varid{heaps}){}\<[E]%
\ColumnHook
\end{hscode}\resethooks
\iffullversion
\removecodespace
\begin{hscode}\SaveRestoreHook
\column{B}{@{}>{\hspre}l<{\hspost}@{}}%
\column{E}{@{}>{\hspre}l<{\hspost}@{}}%
\>[B]{}{\dummy\downarrow_{\Varid{Sync}}}\;\mathbin{:}\;\Varid{∀}\;\Varid{nodes}\;\Varid{→}\;\star{}\<[E]%
\\
\>[B]{}\Varid{nodes}\;\downarrow_{\Varid{Sync}}\;\mathrel{=}\;\Varid{∃}\;\Varid{λ}\;\Varid{v}\;\Varid{→}\;\Varid{nodes}\;\downarrow_{\Varid{Sync}}\;\Varid{v}{}\<[E]%
\\
\>[B]{}{\dummy\uparrow_{\Varid{Sync}}}\;\mathbin{:}\;\Varid{∀}\;\Varid{nodes}\;\Varid{→}\;\star{}\<[E]%
\\
\>[B]{}{\dummy\uparrow_{\Varid{Sync}}}\;\mathrel{=}\;\Varid{↑}\;{\dummy\xrightarrow[\Varid{Sync}]{}\dummy}{}\<[E]%
\ColumnHook
\end{hscode}\resethooks
\fi

\subsection{Correctness} \label{section:DCESH-bisim}
To prove the correctness of the machine, we will now establish a
bisimulation between the CESH and the \DCESHn{} machines.

To simplify this development, we extend the CESH machine with
a rule for the {\textsmaller[.5]{\ensuremath{\MyConid{REMOTE}\;\Varid{c}\;\Varid{i}}}} instruction so that both machines run the
same bytecode. This rule is almost a no-op, but since we are assuming
that the code we run remotely is closed, the environment is emptied,
and since the compiled code {\textsmaller[.5]{\ensuremath{\Varid{c}}}} will end in a {\textsmaller[.5]{\ensuremath{\MyConid{RET}}}} instruction
a return continuation is pushed on the stack.
\iffullversion
\begin{hscode}\SaveRestoreHook
\column{B}{@{}>{\hspre}l<{\hspost}@{}}%
\column{3}{@{}>{\hspre}l<{\hspost}@{}}%
\column{5}{@{}>{\hspre}l<{\hspost}@{}}%
\column{E}{@{}>{\hspre}l<{\hspost}@{}}%
\>[B]{}\Keyword{data}\;{\dummy\xrightarrow[\Varid{CESH}]{}\dummy}\;\mathbin{:}\;\Conid{Rel}\;\Conid{Config}\;\Conid{Config}\;\Keyword{where}{}\<[E]%
\\
\>[B]{}\hsindent{3}{}\<[3]%
\>[3]{}\Varid{...}{}\<[E]%
\\
\>[B]{}\hsindent{3}{}\<[3]%
\>[3]{}\MyConid{REMOTE}\;\mathbin{:}\;\Varid{∀}\;\{\mskip1.5mu \Varid{c'}\;\Varid{i}\;\Varid{c}\;\Varid{e}\;\Varid{s}\;\Varid{h}\mskip1.5mu\}\;\Varid{→}\;{}\<[E]%
\\
\>[3]{}\hsindent{2}{}\<[5]%
\>[5]{}(\MyConid{REMOTE}\;\Varid{c'}\;\Varid{i}\;\MyConid{;}\;\Varid{c},\Varid{e},\Varid{s},\Varid{h})\;\xrightarrow[\Varid{CESH}]{}\;(\Varid{c'},[\mskip1.5mu \mskip1.5mu],\MyConid{cont}\;(\Varid{c},\Varid{e})\;\Varid{∷}\;\Varid{s},\Varid{h}){}\<[E]%
\ColumnHook
\end{hscode}\resethooks
\fi

The intuition behind the relation that we are to construct should be
similar to the intuition for the relation between CES and CESH
configurations, i.e.\ that it is almost equality, but since
values may be pointers to closures, we need to parameterise it by
heaps. The problem now is that \emph{both} machines use pointers, and
the \DCESHn{} machine even uses \emph{remote} pointers and has two
heaps for each node. This means that we have to parameterise the
relations by all the heaps in the system.

\iffullversion
As before, two fragments of code are related if they are equal.
\begin{hscode}\SaveRestoreHook
\column{B}{@{}>{\hspre}l<{\hspost}@{}}%
\column{E}{@{}>{\hspre}l<{\hspost}@{}}%
\>[B]{}\Varid{R}_{\Varid{Code}}\;\mathbin{:}\;\Conid{Rel}\;\Conid{Code}\;\Conid{Code}{}\<[E]%
\\
\>[B]{}\Varid{R}_{\Varid{Code}}\;\Varid{c₁}\;\Varid{c₂}\;\mathrel{=}\;\Varid{c₁}\;\Varid{≡}\;\Varid{c₂}{}\<[E]%
\ColumnHook
\end{hscode}\resethooks
\fi

We define the type of the extra parameter that we need as a synonym
for an indexed family of the closure and continuation
\ifnotfullversion
heaps, {\textsmaller[.5]{\ensuremath{\Conid{Heaps}\;\mathrel{=}\;\Conid{Node}\;\Varid{→}\;\Conid{DCESH.ClosHeap}\;\Varid{×}\;\Conid{DCESH.ContHeap}}}}.
\else
heaps (here {\textsmaller[.5]{\ensuremath{\Conid{DCESH.ContHeap}\;\mathrel{=}\;\Conid{Heap}\;(\Conid{DCESH.Closure}\;\Varid{×}\;\Conid{DCESH.Stack})}}}):
\fi
\iffullversion
\begin{hscode}\SaveRestoreHook
\column{B}{@{}>{\hspre}l<{\hspost}@{}}%
\column{E}{@{}>{\hspre}l<{\hspost}@{}}%
\>[B]{}\Conid{Heaps}\;\mathrel{=}\;\Conid{Node}\;\Varid{→}\;\Conid{DCESH.ClosHeap}\;\Varid{×}\;\Conid{DCESH.ContHeap}{}\<[E]%
\ColumnHook
\end{hscode}\resethooks
\fi

Simply following the recipe that we used for the relation between the
CES and the CESH machines would not prove effective this time around. When we
constructed that, we could be sure that there would be no circularity,
since it was constructed inductively on the structure of the CES
configuration. But now both systems, CESH and \DCESHn{}, have heaps where
there is a potential for circular references (e.g.  a closure,
residing in a heap, whose environment contains a pointer to itself),
so a direct structural induction cannot work.  This is perhaps the
most mathematically (and formally) challenging point
of the paper. To fix this we parameterise the affected relation
definitions by a natural number {\textsmaller[.5]{\ensuremath{\Varid{rank}}}}, which records how many times
pointers are allowed to be dereferenced, in addition to the heap
parameters.

The relation for environments and closures is as before, but with the additional
parameters.
\iffullversion
\begin{hscode}\SaveRestoreHook
\column{B}{@{}>{\hspre}l<{\hspost}@{}}%
\column{12}{@{}>{\hspre}l<{\hspost}@{}}%
\column{15}{@{}>{\hspre}l<{\hspost}@{}}%
\column{E}{@{}>{\hspre}l<{\hspost}@{}}%
\>[B]{}\Varid{R}_{\Varid{Env}}\;{}\<[12]%
\>[12]{}\mathbin{:}\;{}\<[15]%
\>[15]{}\Conid{ℕ}\;\Varid{→}\;\Conid{CESH.ClosHeap}\;\Varid{→}\;\Conid{Heaps}\;\Varid{→}\;\Conid{Rel}\;\Conid{CESH.Env}\;\Conid{DCESH.Env}{}\<[E]%
\\
\>[B]{}\Varid{R}_{\Varid{Clos}}\;{}\<[12]%
\>[12]{}\mathbin{:}\;{}\<[15]%
\>[15]{}\Conid{ℕ}\;\Varid{→}\;\Conid{CESH.ClosHeap}\;\Varid{→}\;\Conid{Heaps}\;\Varid{→}\;\Conid{Rel}\;\Conid{CESH.Closure}\;\Conid{DCESH.Closure}{}\<[E]%
\\
\>[B]{}\Varid{R}_{\Varid{Clos}}\;\Varid{rank}\;\Varid{h}\;\Varid{hs}\;(\Varid{c₁},\Varid{e₁})\;(\Varid{c₂},\Varid{e₂})\;\mathrel{=}\;\Varid{R}_{\Varid{Code}}\;\Varid{c₁}\;\Varid{c₂}\;\Varid{×}\;\Varid{R}_{\Varid{Env}}\;\Varid{rank}\;\Varid{h}\;\Varid{hs}\;\Varid{e₁}\;\Varid{e₂}{}\<[E]%
\ColumnHook
\end{hscode}\resethooks
\fi
The relation for closure pointers is where the rank is used.
If the rank is zero, the relation is trivially fulfilled.
If the rank is non-zero, it makes sure that the CESH pointer points to
a closure in the CESH heap, that the remote pointer of the \DCESHn{} network
points to a closure in the heap of the location that the pointer refers to, and that
the two closures are related:
\begin{hscode}\SaveRestoreHook
\column{B}{@{}>{\hspre}l<{\hspost}@{}}%
\column{3}{@{}>{\hspre}l<{\hspost}@{}}%
\column{12}{@{}>{\hspre}l<{\hspost}@{}}%
\column{19}{@{}>{\hspre}l<{\hspost}@{}}%
\column{E}{@{}>{\hspre}l<{\hspost}@{}}%
\>[B]{}\Varid{R}_{\Varid{rptr}_{\Varid{cl}}}\;\mathbin{:}\;{}\<[12]%
\>[12]{}\Conid{ℕ}\;\Varid{→}\;\Conid{CESH.ClosHeap}\;\Varid{→}\;\Conid{Heaps}\;\Varid{→}\;{}\<[E]%
\\
\>[12]{}\Conid{Rel}\;\Conid{CESH.ClosPtr}\;\Conid{DCESH.ClosPtr}{}\<[E]%
\\
\>[B]{}\Varid{R}_{\Varid{rptr}_{\Varid{cl}}}\;\Varid{0}\;\anonymous \;\anonymous \;\anonymous \;\anonymous \;\mathrel{=}\;\Varid{⊤}{}\<[E]%
\\
\>[B]{}\Varid{R}_{\Varid{rptr}_{\Varid{cl}}}\;(1+\;\Varid{rank})\;\Varid{h}\;\Varid{hs}\;\Varid{ptr₁}\;(\Varid{ptr₂},\Varid{loc})\;\mathrel{=}\;{}\<[E]%
\\
\>[B]{}\hsindent{3}{}\<[3]%
\>[3]{}\Varid{∃₂}\;\Varid{λ}\;\Varid{cl₁}\;\Varid{cl₂}\;\Varid{→}\;{}\<[19]%
\>[19]{}\Varid{h}\;\mathbin{!}\;\Varid{ptr₁}\;\Varid{≡}\;\MyConid{just}\;\Varid{cl₁}\;\Varid{×}\;{}\<[E]%
\\
\>[19]{}\Varid{proj₁}\;(\Varid{hs}\;\Varid{loc})\;\mathbin{!}\;\Varid{ptr₂}\;\Varid{≡}\;\MyConid{just}\;\Varid{cl₂}\;\Varid{×}\;{}\<[E]%
\\
\>[19]{}\Varid{R}_{\Varid{Clos}}\;\Varid{rank}\;\Varid{h}\;\Varid{hs}\;\Varid{cl₁}\;\Varid{cl₂}{}\<[E]%
\ColumnHook
\end{hscode}\resethooks
The relation for values is also as before, but with the extra parameters.
\iffullversion
\begin{hscode}\SaveRestoreHook
\column{B}{@{}>{\hspre}l<{\hspost}@{}}%
\column{10}{@{}>{\hspre}l<{\hspost}@{}}%
\column{29}{@{}>{\hspre}l<{\hspost}@{}}%
\column{42}{@{}>{\hspre}l<{\hspost}@{}}%
\column{E}{@{}>{\hspre}l<{\hspost}@{}}%
\>[B]{}\Varid{R}_{\Varid{Val}}\;\mathbin{:}\;{}\<[10]%
\>[10]{}\Conid{ℕ}\;\Varid{→}\;\Conid{CESH.ClosHeap}\;\Varid{→}\;\Conid{Heaps}\;\Varid{→}\;\Conid{Rel}\;\Conid{CESH.Value}\;\Conid{DCESH.Value}{}\<[E]%
\\
\>[B]{}\Varid{R}_{\Varid{Val}}\;\Varid{rank}\;\Varid{h}\;\Varid{hs}\;(\MyConid{nat}\;\Varid{n₁})\;{}\<[29]%
\>[29]{}(\MyConid{nat}\;\Varid{n₂})\;{}\<[42]%
\>[42]{}\mathrel{=}\;\Varid{n₁}\;\Varid{≡}\;\Varid{n₂}{}\<[E]%
\\
\>[B]{}\Varid{R}_{\Varid{Val}}\;\Varid{rank}\;\Varid{h}\;\Varid{hs}\;(\MyConid{nat}\;\anonymous )\;{}\<[29]%
\>[29]{}(\MyConid{clos}\;\anonymous )\;{}\<[42]%
\>[42]{}\mathrel{=}\;\Varid{⊥}{}\<[E]%
\\
\>[B]{}\Varid{R}_{\Varid{Val}}\;\Varid{rank}\;\Varid{h}\;\Varid{hs}\;(\MyConid{clos}\;\anonymous )\;{}\<[29]%
\>[29]{}(\MyConid{nat}\;\anonymous )\;{}\<[42]%
\>[42]{}\mathrel{=}\;\Varid{⊥}{}\<[E]%
\\
\>[B]{}\Varid{R}_{\Varid{Val}}\;\Varid{rank}\;\Varid{h}\;\Varid{hs}\;(\MyConid{clos}\;\Varid{ptr})\;{}\<[29]%
\>[29]{}(\MyConid{clos}\;\Varid{rptr})\;{}\<[42]%
\>[42]{}\mathrel{=}\;\Varid{R}_{\Varid{rptr}_{\Varid{cl}}}\;\Varid{rank}\;\Varid{h}\;\Varid{hs}\;\Varid{ptr}\;\Varid{rptr}{}\<[E]%
\ColumnHook
\end{hscode}\resethooks
\fi
\iffullversion
The relation for environments is also as before, but included for completeness:
\begin{hscode}\SaveRestoreHook
\column{B}{@{}>{\hspre}l<{\hspost}@{}}%
\column{28}{@{}>{\hspre}l<{\hspost}@{}}%
\column{39}{@{}>{\hspre}l<{\hspost}@{}}%
\column{40}{@{}>{\hspre}l<{\hspost}@{}}%
\column{E}{@{}>{\hspre}l<{\hspost}@{}}%
\>[B]{}\Varid{R}_{\Varid{Env}}\;\Varid{rank}\;\Varid{h}\;\Varid{hs}\;[\mskip1.5mu \mskip1.5mu]\;{}\<[28]%
\>[28]{}[\mskip1.5mu \mskip1.5mu]\;{}\<[39]%
\>[39]{}\mathrel{=}\;\Varid{⊤}{}\<[E]%
\\
\>[B]{}\Varid{R}_{\Varid{Env}}\;\Varid{rank}\;\Varid{h}\;\Varid{hs}\;[\mskip1.5mu \mskip1.5mu]\;{}\<[28]%
\>[28]{}(\Varid{x}\;\Varid{∷}\;\Varid{e₂})\;{}\<[39]%
\>[39]{}\mathrel{=}\;\Varid{⊥}{}\<[E]%
\\
\>[B]{}\Varid{R}_{\Varid{Env}}\;\Varid{rank}\;\Varid{h}\;\Varid{hs}\;(\Varid{x₁}\;\Varid{∷}\;\Varid{e₁})\;{}\<[28]%
\>[28]{}[\mskip1.5mu \mskip1.5mu]\;{}\<[40]%
\>[40]{}\mathrel{=}\;\Varid{⊥}{}\<[E]%
\\
\>[B]{}\Varid{R}_{\Varid{Env}}\;\Varid{rank}\;\Varid{h}\;\Varid{hs}\;(\Varid{x₁}\;\Varid{∷}\;\Varid{e₁})\;{}\<[28]%
\>[28]{}(\Varid{x₂}\;\Varid{∷}\;\Varid{e₂})\;{}\<[40]%
\>[40]{}\mathrel{=}\;\Varid{R}_{\Varid{Val}}\;\Varid{rank}\;\Varid{h}\;\Varid{hs}\;\Varid{x₁}\;\Varid{x₂}\;\Varid{×}\;\Varid{R}_{\Varid{Env}}\;\Varid{rank}\;\Varid{h}\;\Varid{hs}\;\Varid{e₁}\;\Varid{e₂}{}\<[E]%
\ColumnHook
\end{hscode}\resethooks
\fi
The relation for stack elements is almost as before, but now requires that
for \emph{any} natural number {\textsmaller[.5]{\ensuremath{\Varid{rank}}}}, i.e. for any finite number of
pointer dereferencings, the relations hold:
\begin{hscode}\SaveRestoreHook
\column{B}{@{}>{\hspre}l<{\hspost}@{}}%
\column{3}{@{}>{\hspre}l<{\hspost}@{}}%
\column{16}{@{}>{\hspre}l<{\hspost}@{}}%
\column{30}{@{}>{\hspre}l<{\hspost}@{}}%
\column{42}{@{}>{\hspre}l<{\hspost}@{}}%
\column{E}{@{}>{\hspre}l<{\hspost}@{}}%
\>[B]{}\Varid{R}_{\Varid{StackElem}}\;\mathbin{:}\;{}\<[16]%
\>[16]{}\Conid{CESH.ClosHeap}\;\Varid{→}\;\Conid{Heaps}\;\Varid{→}\;{}\<[E]%
\\
\>[B]{}\hsindent{3}{}\<[3]%
\>[3]{}\Conid{Rel}\;\Conid{CESH.StackElem}\;\Conid{DCESH.StackElem}{}\<[E]%
\\
\>[B]{}\Varid{R}_{\Varid{StackElem}}\;\Varid{h}\;\Varid{hs}\;(\MyConid{val}\;\Varid{v₁})\;{}\<[30]%
\>[30]{}(\MyConid{val}\;\Varid{v₂})\;{}\<[42]%
\>[42]{}\mathrel{=}\;{}\<[E]%
\\
\>[B]{}\hsindent{3}{}\<[3]%
\>[3]{}\Varid{∀}\;\Varid{rank}\;\Varid{→}\;\Varid{R}_{\Varid{Val}}\;\Varid{rank}\;\Varid{h}\;\Varid{hs}\;\Varid{v₁}\;\Varid{v₂}{}\<[E]%
\\
\>[B]{}\Varid{R}_{\Varid{StackElem}}\;\Varid{h}\;\Varid{hs}\;(\MyConid{val}\;\anonymous )\;{}\<[30]%
\>[30]{}(\MyConid{cont}\;\anonymous )\;{}\<[42]%
\>[42]{}\mathrel{=}\;\Varid{⊥}{}\<[E]%
\\
\>[B]{}\Varid{R}_{\Varid{StackElem}}\;\Varid{h}\;\Varid{hs}\;(\MyConid{cont}\;\anonymous )\;{}\<[30]%
\>[30]{}(\MyConid{val}\;\anonymous )\;{}\<[42]%
\>[42]{}\mathrel{=}\;\Varid{⊥}{}\<[E]%
\\
\>[B]{}\Varid{R}_{\Varid{StackElem}}\;\Varid{h}\;\Varid{hs}\;(\MyConid{cont}\;\Varid{cl₁})\;{}\<[30]%
\>[30]{}(\MyConid{cont}\;\Varid{cl₂})\;{}\<[42]%
\>[42]{}\mathrel{=}\;{}\<[E]%
\\
\>[B]{}\hsindent{3}{}\<[3]%
\>[3]{}\Varid{∀}\;\Varid{rank}\;\Varid{→}\;\Varid{R}_{\Varid{Clos}}\;\Varid{rank}\;\Varid{h}\;\Varid{hs}\;\Varid{cl₁}\;\Varid{cl₂}{}\<[E]%
\ColumnHook
\end{hscode}\resethooks
The relation for stacks now takes into account that the \DCESHn{}
stacks may end in a pointer representing a remote continuation. It
makes sure that the pointer points to something in the
continuation heap of the location of the pointer, related to the CESH
stack element.
\begin{hscode}\SaveRestoreHook
\column{B}{@{}>{\hspre}l<{\hspost}@{}}%
\column{12}{@{}>{\hspre}l<{\hspost}@{}}%
\column{E}{@{}>{\hspre}l<{\hspost}@{}}%
\>[B]{}\Varid{R}_{\Varid{Stack}}\;\mathbin{:}\;{}\<[12]%
\>[12]{}\Conid{CESH.ClosHeap}\;\Varid{→}\;\Conid{Heaps}\;\Varid{→}\;{}\<[E]%
\\
\>[12]{}\Conid{Rel}\;\Conid{CESH.Stack}\;\Conid{DCESH.Stack}{}\<[E]%
\ColumnHook
\end{hscode}\resethooks
\removecodespace
\iffullversion
\savecolumns
\begin{hscode}\SaveRestoreHook
\column{B}{@{}>{\hspre}l<{\hspost}@{}}%
\column{29}{@{}>{\hspre}l<{\hspost}@{}}%
\column{52}{@{}>{\hspre}l<{\hspost}@{}}%
\column{55}{@{}>{\hspre}l<{\hspost}@{}}%
\column{E}{@{}>{\hspre}l<{\hspost}@{}}%
\>[B]{}\Varid{R}_{\Varid{Stack}}\;\Varid{h}\;\Varid{hs}\;[\mskip1.5mu \mskip1.5mu]\;{}\<[29]%
\>[29]{}([\mskip1.5mu \mskip1.5mu],\MyConid{nothing})\;{}\<[52]%
\>[52]{}\mathrel{=}\;\Varid{⊤}{}\<[E]%
\\
\>[B]{}\Varid{R}_{\Varid{Stack}}\;\Varid{h}\;\Varid{hs}\;[\mskip1.5mu \mskip1.5mu]\;{}\<[29]%
\>[29]{}(\Varid{x}\;\Varid{∷}\;\Varid{stack₂},\Varid{r})\;{}\<[52]%
\>[52]{}\mathrel{=}\;\Varid{⊥}{}\<[E]%
\\
\>[B]{}\Varid{R}_{\Varid{Stack}}\;\Varid{h}\;\Varid{hs}\;(\Varid{x₁}\;\Varid{∷}\;\Varid{stack₁})\;{}\<[29]%
\>[29]{}(\Varid{x₂}\;\Varid{∷}\;\Varid{stack₂},\Varid{r})\;{}\<[52]%
\>[52]{}\mathrel{=}\;{}\<[55]%
\>[55]{}\Varid{R}_{\Varid{StackElem}}\;\Varid{h}\;\Varid{hs}\;\Varid{x₁}\;\Varid{x₂}\;\Varid{×}\;{}\<[E]%
\\
\>[55]{}\Varid{R}_{\Varid{Stack}}\;\Varid{h}\;\Varid{hs}\;\Varid{stack₁}\;(\Varid{stack₂},\Varid{r}){}\<[E]%
\\
\>[B]{}\Varid{R}_{\Varid{Stack}}\;\Varid{h}\;\Varid{hs}\;(\Varid{x}\;\Varid{∷}\;\Varid{stack₁})\;{}\<[29]%
\>[29]{}([\mskip1.5mu \mskip1.5mu],\MyConid{nothing})\;{}\<[52]%
\>[52]{}\mathrel{=}\;\Varid{⊥}{}\<[E]%
\\
\>[B]{}\Varid{R}_{\Varid{Stack}}\;\Varid{h}\;\Varid{hs}\;[\mskip1.5mu \mskip1.5mu]\;{}\<[29]%
\>[29]{}([\mskip1.5mu \mskip1.5mu],\MyConid{just}\;\anonymous )\;{}\<[52]%
\>[52]{}\mathrel{=}\;\Varid{⊥}{}\<[E]%
\ColumnHook
\end{hscode}\resethooks
\fi
\removecodespace
\iffullversion
\restorecolumns
\fi
\begin{hscode}\SaveRestoreHook
\column{B}{@{}>{\hspre}l<{\hspost}@{}}%
\column{3}{@{}>{\hspre}l<{\hspost}@{}}%
\column{6}{@{}>{\hspre}l<{\hspost}@{}}%
\column{29}{@{}>{\hspre}l<{\hspost}@{}}%
\column{54}{@{}>{\hspre}l<{\hspost}@{}}%
\column{E}{@{}>{\hspre}l<{\hspost}@{}}%
\>[B]{}\Varid{R}_{\Varid{Stack}}\;\Varid{h}\;\Varid{hs}\;(\Varid{cont₁}\;\Varid{∷}\;\Varid{s₁})\;{}\<[29]%
\>[29]{}([\mskip1.5mu \mskip1.5mu],\MyConid{just}\;(\Varid{ptr},\Varid{loc}))\;{}\<[54]%
\>[54]{}\mathrel{=}\;{}\<[E]%
\\
\>[B]{}\hsindent{3}{}\<[3]%
\>[3]{}\Varid{∃₂}\;\Varid{λ}\;\Varid{cont₂}\;\Varid{s₂}\;\Varid{→}\;\Varid{proj₂}\;(\Varid{hs}\;\Varid{loc})\;\mathbin{!}\;\Varid{ptr}\;\Varid{≡}\;\MyConid{just}\;(\Varid{cont₂},\Varid{s₂})\;\Varid{×}\;{}\<[E]%
\\
\>[3]{}\hsindent{3}{}\<[6]%
\>[6]{}\Varid{R}_{\Varid{StackElem}}\;\Varid{h}\;\Varid{hs}\;\Varid{cont₁}\;(\MyConid{cont}\;\Varid{cont₂})\;\Varid{×}\;{}\<[E]%
\\
\>[3]{}\hsindent{3}{}\<[6]%
\>[6]{}\Varid{R}_{\Varid{Stack}}\;\Varid{h}\;\Varid{hs}\;\Varid{s₁}\;\Varid{s₂}{}\<[E]%
\ColumnHook
\end{hscode}\resethooks
Finally, a CESH configuration and a \DCESHn{} thread are related if the
thread is running and the constituents are pointwise related:
\begin{hscode}\SaveRestoreHook
\column{B}{@{}>{\hspre}l<{\hspost}@{}}%
\column{3}{@{}>{\hspre}l<{\hspost}@{}}%
\column{15}{@{}>{\hspre}l<{\hspost}@{}}%
\column{36}{@{}>{\hspre}l<{\hspost}@{}}%
\column{59}{@{}>{\hspre}l<{\hspost}@{}}%
\column{E}{@{}>{\hspre}l<{\hspost}@{}}%
\>[B]{}\Varid{R}_{\Varid{Thread}}\;\mathbin{:}\;\Conid{Heaps}\;\Varid{→}\;\Conid{Rel}\;\Conid{Config}\;(\Conid{Maybe}\;\Conid{Thread}){}\<[E]%
\\
\>[B]{}\Varid{R}_{\Varid{Thread}}\;\Varid{hs}\;{}\<[15]%
\>[15]{}\anonymous \;{}\<[36]%
\>[36]{}\MyConid{nothing}\;{}\<[59]%
\>[59]{}\mathrel{=}\;\Varid{⊥}{}\<[E]%
\\
\>[B]{}\Varid{R}_{\Varid{Thread}}\;\Varid{hs}\;{}\<[15]%
\>[15]{}(\Varid{c₁},\Varid{e₁},\Varid{s₁},\Varid{h₁})\;{}\<[36]%
\>[36]{}(\MyConid{just}\;(\Varid{c₂},\Varid{e₂},\Varid{s₂}))\;{}\<[59]%
\>[59]{}\mathrel{=}\;{}\<[E]%
\\
\>[B]{}\hsindent{3}{}\<[3]%
\>[3]{}\Varid{R}_{\Varid{Code}}\;\Varid{c₁}\;\Varid{c₂}\;\Varid{×}\;(\Varid{∀}\;\Varid{rank}\;\Varid{→}\;\Varid{R}_{\Varid{Env}}\;\Varid{rank}\;\Varid{h₁}\;\Varid{hs}\;\Varid{e₁}\;\Varid{e₂})\;\Varid{×}\;{}\<[E]%
\\
\>[B]{}\hsindent{3}{}\<[3]%
\>[3]{}\Varid{R}_{\Varid{Stack}}\;\Varid{h₁}\;\Varid{hs}\;\Varid{s₁}\;\Varid{s₂}{}\<[E]%
\ColumnHook
\end{hscode}\resethooks
\iffullversion
A configuration is related to an asynchronous \DCESHn{} network if the network
has exactly one running node, {\textsmaller[.5]{\ensuremath{\Varid{i}}}}, that is related to the configuration, and
there are no messages in the message soup:
\begin{hscode}\SaveRestoreHook
\column{B}{@{}>{\hspre}l<{\hspost}@{}}%
\column{3}{@{}>{\hspre}l<{\hspost}@{}}%
\column{29}{@{}>{\hspre}l<{\hspost}@{}}%
\column{E}{@{}>{\hspre}l<{\hspost}@{}}%
\>[B]{}\Varid{R}_{\Varid{Async}}\;\mathbin{:}\;\Conid{Rel}\;\Conid{Config}\;\Conid{AsyncNetwork}{}\<[E]%
\\
\>[B]{}\Varid{R}_{\Varid{Async}}\;\Varid{cfg}\;(\Varid{nodes},[\mskip1.5mu \mskip1.5mu])\;{}\<[29]%
\>[29]{}\mathrel{=}\;\Varid{∃}\;\Varid{λ}\;\Varid{i}\;\Varid{→}\;{}\<[E]%
\\
\>[B]{}\hsindent{3}{}\<[3]%
\>[3]{}\Varid{all}\;\Varid{nodes}\;\Varid{except}\;\Varid{i}\;\Varid{are}\;\Varid{inactive}\;\Varid{×}\;{}\<[E]%
\\
\>[B]{}\hsindent{3}{}\<[3]%
\>[3]{}\Varid{R}_{\Varid{Thread}}\;(\Varid{proj₂}\;\Varid{∘}\;\Varid{nodes})\;\Varid{cfg}\;(\Varid{proj₁}\;(\Varid{nodes}\;\Varid{i})){}\<[E]%
\\
\>[B]{}\Varid{R}_{\Varid{Async}}\;\Varid{cfg}\;(\Varid{nodes},\Varid{msgs})\;{}\<[29]%
\>[29]{}\mathrel{=}\;\Varid{⊥}{}\<[E]%
\ColumnHook
\end{hscode}\resethooks
A configuration is related to a synchronous \DCESHn{} network if it is related
to the asynchronous network gotten by pairing the synchronous network
with an empty list of messages:
\begin{hscode}\SaveRestoreHook
\column{B}{@{}>{\hspre}l<{\hspost}@{}}%
\column{E}{@{}>{\hspre}l<{\hspost}@{}}%
\>[B]{}\Varid{R}_{\Varid{Sync}}\;\mathbin{:}\;\Conid{Rel}\;\Conid{Config}\;\Conid{SyncNetwork}{}\<[E]%
\\
\>[B]{}\Varid{R}_{\Varid{Sync}}\;\Varid{cfg}\;\Varid{nodes}\;\mathrel{=}\;\Varid{R}_{\Varid{Async}}\;\Varid{cfg}\;(\Varid{nodes},[\mskip1.5mu \mskip1.5mu]){}\<[E]%
\ColumnHook
\end{hscode}\resethooks
\fi

We order heaps of a \DCESHn{} network pointwise as follows (called
{\textsmaller[.5]{\ensuremath{\Varid{⊆s}}}} since it is the ``plural'' of {\textsmaller[.5]{\ensuremath{\Varid{⊆}}}}):
\iffullversion
\begin{hscode}\SaveRestoreHook
\column{B}{@{}>{\hspre}l<{\hspost}@{}}%
\column{E}{@{}>{\hspre}l<{\hspost}@{}}%
\>[B]{}\Varid{\char95 ⊆s\char95 }\;\mathbin{:}\;(\Varid{hs}\;\Varid{hs'}\;\mathbin{:}\;\Conid{Heaps})\;\Varid{→}\;\star{}\<[E]%
\ColumnHook
\end{hscode}\resethooks
\removecodespace
\fi
\begin{hscode}\SaveRestoreHook
\column{B}{@{}>{\hspre}l<{\hspost}@{}}%
\column{20}{@{}>{\hspre}l<{\hspost}@{}}%
\column{23}{@{}>{\hspre}l<{\hspost}@{}}%
\column{25}{@{}>{\hspre}l<{\hspost}@{}}%
\column{42}{@{}>{\hspre}l<{\hspost}@{}}%
\column{E}{@{}>{\hspre}l<{\hspost}@{}}%
\>[B]{}\Varid{hs}\;\Varid{⊆s}\;\Varid{hs'}\;\mathrel{=}\;\Varid{∀}\;\Varid{i}\;\Varid{→}\;{}\<[20]%
\>[20]{}\Keyword{let}\;{}\<[25]%
\>[25]{}(\Varid{h}_{\Varid{cl}},\Varid{h}_{\Varid{cnt}})\;{}\<[42]%
\>[42]{}\mathrel{=}\;\Varid{hs}\;\Varid{i}{}\<[E]%
\\
\>[25]{}(\Varid{h'}_{\Varid{cl}},\Varid{h'}_{\Varid{cnt}})\;{}\<[42]%
\>[42]{}\mathrel{=}\;\Varid{hs'}\;\Varid{i}{}\<[E]%
\\
\>[20]{}\hsindent{3}{}\<[23]%
\>[23]{}\Keyword{in}\;\Varid{h}_{\Varid{cl}}\;\Varid{⊆}\;\Varid{h'}_{\Varid{cl}}\;\Varid{×}\;\Varid{h}_{\Varid{cnt}}\;\Varid{⊆}\;\Varid{h'}_{\Varid{cnt}}{}\<[E]%
\ColumnHook
\end{hscode}\resethooks
\iffullversion
\removecodespace
\begin{hscode}\SaveRestoreHook
\column{B}{@{}>{\hspre}l<{\hspost}@{}}%
\column{3}{@{}>{\hspre}l<{\hspost}@{}}%
\column{6}{@{}>{\hspre}l<{\hspost}@{}}%
\column{9}{@{}>{\hspre}l<{\hspost}@{}}%
\column{10}{@{}>{\hspre}l<{\hspost}@{}}%
\column{21}{@{}>{\hspre}l<{\hspost}@{}}%
\column{E}{@{}>{\hspre}l<{\hspost}@{}}%
\>[B]{}\Varid{⊆s-refl}\;\mathbin{:}\;(\Varid{hs}\;\mathbin{:}\;\Conid{Heaps})\;\Varid{→}\;\Varid{hs}\;\Varid{⊆s}\;\Varid{hs}{}\<[E]%
\\
\>[B]{}\Varid{⊆s-refl}\;\Varid{hs}\;\Varid{node}\;\mathrel{=}\;\Keyword{let}\;(\Varid{h}_{\Varid{cl}},\Varid{h}_{\Varid{cnt}})\;\mathrel{=}\;\Varid{hs}\;\Varid{node}{}\<[E]%
\\
\>[B]{}\hsindent{21}{}\<[21]%
\>[21]{}\Keyword{in}\;\Varid{⊆-refl}\;\Varid{h}_{\Varid{cl}},\Varid{⊆-refl}\;\Varid{h}_{\Varid{cnt}}{}\<[E]%
\\
\>[B]{}\Varid{⊆s-trans}\;\mathbin{:}\;\{\mskip1.5mu \Varid{hs₁}\;\Varid{hs₂}\;\Varid{hs₃}\;\mathbin{:}\;\Conid{Heaps}\mskip1.5mu\}\;\Varid{→}\;{}\<[E]%
\\
\>[B]{}\hsindent{3}{}\<[3]%
\>[3]{}\Varid{hs₁}\;\Varid{⊆s}\;\Varid{hs₂}\;\Varid{→}\;\Varid{hs₂}\;\Varid{⊆s}\;\Varid{hs₃}\;\Varid{→}\;\Varid{hs₁}\;\Varid{⊆s}\;\Varid{hs₃}{}\<[E]%
\\
\>[B]{}\Varid{⊆s-trans}\;\Varid{hs₁⊆shs₂}\;\Varid{hs₂⊆shs₃}\;\Varid{node}\;{}\<[E]%
\\
\>[B]{}\hsindent{3}{}\<[3]%
\>[3]{}\mathrel{=}\;\Keyword{let}\;(\Varid{clh₁⊆clh₂},\Varid{conth₁⊆conth₂})\;\mathrel{=}\;\Varid{hs₁⊆shs₂}\;\Varid{node}{}\<[E]%
\\
\>[3]{}\hsindent{6}{}\<[9]%
\>[9]{}(\Varid{clh₂⊆clh₃},\Varid{conth₂⊆conth₃})\;\mathrel{=}\;\Varid{hs₂⊆shs₃}\;\Varid{node}{}\<[E]%
\\
\>[3]{}\hsindent{3}{}\<[6]%
\>[6]{}\Keyword{in}\;{}\<[10]%
\>[10]{}\Varid{⊆-trans}\;\Varid{clh₁⊆clh₂}\;\Varid{clh₂⊆clh₃},{}\<[E]%
\\
\>[10]{}\Varid{⊆-trans}\;\Varid{conth₁⊆conth₂}\;\Varid{conth₂⊆conth₃}{}\<[E]%
\ColumnHook
\end{hscode}\resethooks
\fi
\begin{lemma}[{\textsmaller[.5]{\ensuremath{\Conid{HeapUpdate.env}}}}, {\textsmaller[.5]{\ensuremath{\Conid{HeapUpdate.stack}}}}]
Given CESH closure heaps {\textsmaller[.5]{\ensuremath{\Varid{h}}}} and {\textsmaller[.5]{\ensuremath{\Varid{h'}}}} such that {\textsmaller[.5]{\ensuremath{\Varid{h}\;\Varid{⊆}\;\Varid{h'}}}} and families of \DCESHn{} heaps {\textsmaller[.5]{\ensuremath{\Varid{hs}}}} and {\textsmaller[.5]{\ensuremath{\Varid{hs'}}}}
such that {\textsmaller[.5]{\ensuremath{\Varid{hs}\;\Varid{⊆s}\;\Varid{hs'}}}},
\ifnotfullversion
then {\textsmaller[.5]{\ensuremath{\Varid{R}_{\Varid{Env}}\;\Varid{n}\;\Varid{h}\;\Varid{hs}\;\Varid{e₁}\;\Varid{e₂}}}} implies {\textsmaller[.5]{\ensuremath{\Varid{R}_{\Varid{Env}}\;\Varid{n}\;\Varid{h'}\;\Varid{hs'}\;\Varid{e₁}\;\Varid{e₂}}}} and {\textsmaller[.5]{\ensuremath{\Varid{R}_{\Varid{Stack}}\;\Varid{h}\;\Varid{hs}\;\Varid{s₁}\;\Varid{s₂}}}} implies {\textsmaller[.5]{\ensuremath{\Varid{R}_{\Varid{Stack}}\;\Varid{h'}\;\Varid{hs'}\;\Varid{s₁}\;\Varid{s₂}}}}.
\else
then we can prove the following:
\begin{hscode}\SaveRestoreHook
\column{B}{@{}>{\hspre}l<{\hspost}@{}}%
\column{3}{@{}>{\hspre}l<{\hspost}@{}}%
\column{24}{@{}>{\hspre}l<{\hspost}@{}}%
\column{E}{@{}>{\hspre}l<{\hspost}@{}}%
\>[3]{}\Varid{env}\;\mathbin{:}\;\Varid{∀}\;\{\mskip1.5mu \Varid{n}\mskip1.5mu\}\;\Varid{e₁}\;\Varid{e₂}\;\Varid{→}\;{}\<[24]%
\>[24]{}\Varid{R}_{\Varid{Env}}\;\Varid{n}\;\Varid{h}\;\Varid{hs}\;\Varid{e₁}\;\Varid{e₂}\;\Varid{→}\;\Varid{R}_{\Varid{Env}}\;\Varid{n}\;\Varid{h'}\;\Varid{hs'}\;\Varid{e₁}\;\Varid{e₂}{}\<[E]%
\ColumnHook
\end{hscode}\resethooks
\removecodespace
\begin{hscode}\SaveRestoreHook
\column{B}{@{}>{\hspre}l<{\hspost}@{}}%
\column{3}{@{}>{\hspre}l<{\hspost}@{}}%
\column{22}{@{}>{\hspre}l<{\hspost}@{}}%
\column{E}{@{}>{\hspre}l<{\hspost}@{}}%
\>[3]{}\Varid{stack}\;\mathbin{:}\;\Varid{∀}\;\Varid{s₁}\;\Varid{s₂}\;\Varid{→}\;{}\<[22]%
\>[22]{}\Varid{R}_{\Varid{Stack}}\;\Varid{h}\;\Varid{hs}\;\Varid{s₁}\;\Varid{s₂}\;\Varid{→}\;\Varid{R}_{\Varid{Stack}}\;\Varid{h'}\;\Varid{hs'}\;\Varid{s₁}\;\Varid{s₂}{}\<[E]%
\ColumnHook
\end{hscode}\resethooks
\fi
\end{lemma}

\begin{theorem}[{\textsmaller[.5]{\ensuremath{\Varid{simulation}_{\Varid{Sync}}}}}]
{\textsmaller[.5]{\ensuremath{\Varid{R}_{\Varid{Sync}}}}} is a simulation relation.
\iffullversion
\begin{hscode}\SaveRestoreHook
\column{B}{@{}>{\hspre}l<{\hspost}@{}}%
\column{E}{@{}>{\hspre}l<{\hspost}@{}}%
\>[B]{}\Varid{simulation}_{\Varid{Sync}}\;\mathbin{:}\;\Conid{Simulation}\;{\dummy\xrightarrow[\Varid{CESH}]{}\dummy}\;{\dummy\xrightarrow[\Varid{Sync}]{}\dummy}\;\Varid{R}_{\Varid{Sync}}{}\<[E]%
\ColumnHook
\end{hscode}\resethooks
\fi
\end{theorem}
\begin{proof}
By cases on the CESH transition. In each case, the \DCESHn{}
network can make analogous transitions. Use the {\textsmaller[.5]{\ensuremath{\Conid{HeapUpdate}}}} lemmas to
show that {\textsmaller[.5]{\ensuremath{\Varid{R}_{\Varid{Sync}}}}} is preserved.
\end{proof}

\begin{theorem}[{\textsmaller[.5]{\ensuremath{\Varid{presimulation}_{\Varid{Sync}}}}}]
The inverse of {\textsmaller[.5]{\ensuremath{\Varid{R}_{\Varid{Sync}}}}} is a presimulation.
\iffullversion
\begin{hscode}\SaveRestoreHook
\column{B}{@{}>{\hspre}l<{\hspost}@{}}%
\column{37}{@{}>{\hspre}l<{\hspost}@{}}%
\column{E}{@{}>{\hspre}l<{\hspost}@{}}%
\>[B]{}\Varid{presimulation}_{\Varid{Sync}}\;\mathbin{:}\;\Conid{Presimulation}\;{}\<[37]%
\>[37]{}{\dummy\xrightarrow[\Varid{Sync}]{}\dummy}\;{\dummy\xrightarrow[\Varid{CESH}]{}\dummy}\;(\Varid{R}_{\Varid{Sync}}\;\Varid{⁻¹}){}\<[E]%
\ColumnHook
\end{hscode}\resethooks
\fi
\end{theorem}

\begin{theorem}[{\textsmaller[.5]{\ensuremath{\Varid{bisimulation}_{\Varid{Sync}}}}}]
{\textsmaller[.5]{\ensuremath{\Varid{R}_{\Varid{Sync}}}}} is a bisimulation.
\iffullversion
\begin{hscode}\SaveRestoreHook
\column{B}{@{}>{\hspre}l<{\hspost}@{}}%
\column{E}{@{}>{\hspre}l<{\hspost}@{}}%
\>[B]{}\Varid{bisimulation}_{\Varid{Sync}}\;\mathbin{:}\;\Conid{Bisimulation}\;{\dummy\xrightarrow[\Varid{CESH}]{}\dummy}\;{\dummy\xrightarrow[\Varid{Sync}]{}\dummy}\;\Varid{R}_{\Varid{Sync}}{}\<[E]%
\ColumnHook
\end{hscode}\resethooks
\fi
\end{theorem}
\begin{proof}
Theorem {\textsmaller[.5]{\ensuremath{\Varid{presimulation-to-simulation}}}} applied to
{\textsmaller[.5]{\ensuremath{\Varid{determinism}_{\Varid{Sync}}}}} and {\textsmaller[.5]{\ensuremath{\Varid{simulation}_{\Varid{Sync}}}}} implies that {\textsmaller[.5]{\ensuremath{\Varid{R}_{\Varid{Sync}}\;\Varid{⁻¹}}}} is a simulation, which together with {\textsmaller[.5]{\ensuremath{\Varid{simulation}_{\Varid{Sync}}}}} shows that
{\textsmaller[.5]{\ensuremath{\Varid{R}_{\Varid{Sync}}}}} is a bisimulation.
\end{proof}
\begin{corollary}[\mbox{{\textsmaller[.5]{\ensuremath{\Varid{termination-agrees}_{\Varid{Sync}}}}}, {\textsmaller[.5]{\ensuremath{\Varid{divergence-agrees}_{\Varid{Sync}}}}}}]
\ifnotfullversion
In particular, if {\textsmaller[.5]{\ensuremath{\Varid{R}_{\Varid{Sync}}\;\Varid{cfg}\;\Varid{nodes}}}} then {\textsmaller[.5]{\ensuremath{\Varid{cfg}\;\downarrow_{\Varid{CESH}}\;\MyConid{nat}\;\Varid{n}\;\Varid{↔}\;\Varid{nodes}\;\downarrow_{\Varid{Sync}}\;\MyConid{nat}\;\Varid{n}}}} and {\textsmaller[.5]{\ensuremath{\Varid{cfg}\;\uparrow_{\Varid{CESH}}\;\Varid{↔}\;\Varid{nodes}\;\uparrow_{\Varid{Sync}}}}}.
\else
In particular, a CESH configuration terminates with a natural number
{\textsmaller[.5]{\ensuremath{\Varid{n}}}} (diverges) if and only if a related synchronous DCESH network
terminates with a natural number {\textsmaller[.5]{\ensuremath{\Varid{n}}}} (diverges).
\begin{hscode}\SaveRestoreHook
\column{B}{@{}>{\hspre}l<{\hspost}@{}}%
\column{3}{@{}>{\hspre}l<{\hspost}@{}}%
\column{E}{@{}>{\hspre}l<{\hspost}@{}}%
\>[B]{}\Varid{termination-agrees}_{\Varid{Sync}}\;\mathbin{:}\;\Varid{∀}\;\Varid{cfg}\;\Varid{nodes}\;\Varid{n}\;\Varid{→}\;\Varid{R}_{\Varid{Sync}}\;\Varid{cfg}\;\Varid{nodes}\;\Varid{→}\;{}\<[E]%
\\
\>[B]{}\hsindent{3}{}\<[3]%
\>[3]{}\Varid{cfg}\;\downarrow_{\Varid{CESH}}\;\MyConid{nat}\;\Varid{n}\;\Varid{↔}\;\Varid{nodes}\;\downarrow_{\Varid{Sync}}\;\MyConid{nat}\;\Varid{n}{}\<[E]%
\ColumnHook
\end{hscode}\resethooks
\removecodespace
\begin{hscode}\SaveRestoreHook
\column{B}{@{}>{\hspre}l<{\hspost}@{}}%
\column{3}{@{}>{\hspre}l<{\hspost}@{}}%
\column{E}{@{}>{\hspre}l<{\hspost}@{}}%
\>[B]{}\Varid{divergence-agrees}_{\Varid{Sync}}\;\mathbin{:}\;\Varid{∀}\;\Varid{cfg₁}\;\Varid{cfg₂}\;\Varid{→}\;\Varid{R}_{\Varid{Sync}}\;\Varid{cfg₁}\;\Varid{cfg₂}\;\Varid{→}\;{}\<[E]%
\\
\>[B]{}\hsindent{3}{}\<[3]%
\>[3]{}\Varid{cfg₁}\;\uparrow_{\Varid{CESH}}\;\Varid{↔}\;\Varid{cfg₂}\;\uparrow_{\Varid{Sync}}{}\<[E]%
\ColumnHook
\end{hscode}\resethooks
\fi
\end{corollary}
We also have that initial configurations are in {\textsmaller[.5]{\ensuremath{\Varid{R}_{\Varid{Sync}}}}}:
\begin{hscode}\SaveRestoreHook
\column{B}{@{}>{\hspre}l<{\hspost}@{}}%
\column{42}{@{}>{\hspre}l<{\hspost}@{}}%
\column{E}{@{}>{\hspre}l<{\hspost}@{}}%
\>[B]{}\Varid{initial-related}_{\Varid{Sync}}\;\mathbin{:}\;\Varid{∀}\;\Varid{c}\;\Varid{i}\;\Varid{→}\;\Varid{R}_{\Varid{Sync}}\;{}\<[42]%
\>[42]{}(\Varid{c},[\mskip1.5mu \mskip1.5mu],[\mskip1.5mu \mskip1.5mu],\Varid{∅})\;{}\<[E]%
\\
\>[42]{}(\Varid{initial-network}_{\Varid{Sync}}\;\Varid{c}\;\Varid{i}){}\<[E]%
\ColumnHook
\end{hscode}\resethooks

These final results complete the picture for the \DCESHn{} machine. We
have established that we get the same final result regardless of
whether we choose to run a fragment of code using the CES, the CESH,
or the \DCESHn{} machine.

\section{Related work}

There is a multitude of programming languages and libraries for distributed
computing. We focus mostly on those with a functional flavour. For surveys, see
\cite{DBLP:journals/jfp/TrinderLP02, DBLP:journals/lisp/LoidlRSHHKLMPPPT03}.
Broadly speaking, we can divide them into those that use some form of explicit
message passing, and those that have more implicit mechanisms for distribution
and communication.

\paragraph*{Explicit}
A prime example of a language for distributed computing that uses explicit
message passing is Erlang~\cite{DBLP:books/daglib/0073501}. Erlang is a very
successful language used prominently in the telecommunication industry.
Conceptually similar solutions include MPI~\cite{gropp1999using} and Cloud
Haskell~\cite{DBLP:conf/haskell/EpsteinBJ11}.
The theoretically advanced projects Nomadic
Pict~\cite{DBLP:journals/ieeecc/WojciechowskiS00} and the distributed join
calculus~\cite{DBLP:conf/concur/FournetGLMR96} both support a notion of
mobility for distributed agents, which enables more expressivity for the
distribution of a program than the fairly static networks that our work uses.
In general, explicit languages are well-proven, but far away in the language
design-space from the seamless distributed computing that we envision because
they place the burden of explicit communication on the programmer.

\paragraph*{Implicit}
Our work can be seen as a generalised Remote Procedure
Call (RPC)~\cite{DBLP:journals/tocs/BirrelN84}.  In \emph{loc. cit.} it is argued
that emulating a shared address space is infeasible since it requires each
pointer to also contain location information, and that it is questionable
whether acceptable efficiency can be achieved.  These arguments certainly apply
to our work, where we do just this. With the goal of expressivity in mind,
however, we believe that we should \emph{enable} the programmer to write the
potentially inefficient programs that (internally) use remote pointers, because
not all programs are performance critical. Furthermore, using a tagged pointer
representation~\cite{DBLP:conf/icfp/MarlowYJ07} for closure pointers means that
we can tag pointers that are remote, and pay a very low, if any, performance
penalty for local pointers.

Remote Evaluation (REV)~\cite{DBLP:journals/toplas/StamosG90} is another
generalisation of RPC, siding with us on enabling the use of
higher-order functions across node boundaries. The main differences between REV
and our work is that REV relies on sending code and that it has a more general
distribution mechanism.

The well-researched project Eden~\cite{DBLP:journals/jfp/LoogenOP05}, which
builds on Haskell, is a semi-implicit language. Eden allows expressing
distributed algorithms at a high level of abstraction, and is mostly implicit
about communication, but explicit about process creation. Eden is specified
operationally using a two-level semantics similar to ours.

Hop~\cite{DBLP:conf/oopsla/SerranoGL06}, Links~\cite{DBLP:conf/fmco/CooperLWY06},
and ML5~\cite{DBLP:conf/tgc/VIICH07}
are examples of so called
\emph{tierless} languages that allow writing (for instance) the client and
server code of web applications in unified languages with more or less
seamless interoperability between them. We believe that our work shows
how a principled back-end and semantics can work for such languages.

\section{Conclusion and further work}
We have seen the definition and correctness proofs of \DCESHn{}, a distributed
abstract machine. Previously we have argued that distributed and heterogeneous
programming would benefit from languages that are architecture-independent,
using compilation based on the idea of seamless
computing~\cite{DBLP:conf/lics/FredrikssonG13}. This would allow the programmer
to focus on solving algorithmic problems without having to worry about the
low-level details of the underlying computational system. Our previous work
shows how to achieve this, but is very different from conventional compilation
techniques, relying on game semantics. This means that the vast literature on compiler optimisation does
not generally apply to it, and that it is difficult to interface with legacy
code.  We believe that the current work alleviates these issues, since it shows
a way to do distributed execution as a conservative extension of existing
abstract machines.  Additionally, \DCESHn{} adds very little overhead, if any,
for \emph{local} execution, while permitting any sub-terms to be seamlessly
distributed.


\paragraph*{Implementation}
An implementation of the \DCESHn{} machine can be constructed by 
\iffullversion
either a bytecode
interpreter or 
\fi
compiling the bytecode into a low-level language by macro expansion.  We have a
prototype implementation that does
\iffullversion
the latter,
\fi
illustrating the potential for
using \DCESHn{} as a basis for a usable 
\ifnotfullversion
compiler \cite{SourceCode}.
\else
compiler.
\fi

\paragraph*{Outstanding questions}
\begin{itemize}
  \item Do the proofs generalise to a language with parallelism?
  \item Can we efficiently do distributed garbage collection
  \cite{DBLP:conf/iwmm/PlainfosseS95}? This is necessary, since \DCESHn{}, in
  contrast to our previous work, never reclaims heap garbage. 
  \iffullversion
  It would also be interesting to find out if parts of programs can use
  \emph{local} garbage collection for better performance.
  \fi
  \item Can we find a way to express more complicated distribution patterns
  than those made possible by locus specifiers? From our experience, locus
  specifiers are excellent for simple programs (especially those with
  client-server disciplines), but due to the static nature of the specifiers,
  it is hard to express dynamic distributed algorithms.
  We believe that our work can be extended with dynamic locus specifiers to
  handle this.
  \iffullversion
  A simple first step would be to add support for compiling parts of a program
  for more than one node at a time, making it possible to pass (references to)
  functions already existing on some remote node to it.
  \item Can we add support for sending code code (like
  REV~\cite{DBLP:journals/toplas/StamosG90}) when the code is location-independent?
  \fi
\end{itemize}
\iffullversion
Two other language features that our abstract machines currently do not handle,
but that we would like to implement are abstract data types and mutable
references.
\fi

\section*{Acknowledgements}
\iffullversion
The author would like to thank
Mart\'{i}n Escard\'{o} for assistance with Agda,
Fredrik Nordvall Forsberg for rubber ducking,
Dan Ghica for fruitful discussions and supervision,
and Paul Blain Levy for simplifying some of the definitions.

\fi
This work was supported by Microsoft Research through its PhD Scholarship
Programme.


\bibliographystyle{IEEEtran}
\iffullversion
\bibliography{IEEEabrv,\jobname,dblprefs}

\begin{thebibliography}{10}
\providecommand{\url}[1]{#1}
\csname url@samestyle\endcsname
\providecommand{\newblock}{\relax}
\providecommand{\bibinfo}[2]{#2}
\providecommand{\BIBentrySTDinterwordspacing}{\spaceskip=0pt\relax}
\providecommand{\BIBentryALTinterwordstretchfactor}{4}
\providecommand{\BIBentryALTinterwordspacing}{\spaceskip=\fontdimen2\font plus
\BIBentryALTinterwordstretchfactor\fontdimen3\font minus
  \fontdimen4\font\relax}
\providecommand{\BIBforeignlanguage}[2]{{%
\expandafter\ifx\csname l@#1\endcsname\relax
\typeout{** WARNING: IEEEtran.bst: No hyphenation pattern has been}%
\typeout{** loaded for the language `#1'. Using the pattern for}%
\typeout{** the default language instead.}%
\else
\language=\csname l@#1\endcsname
\fi
#2}}
\providecommand{\BIBdecl}{\relax}
\BIBdecl

\bibitem{gropp1999using}
W.~D. Gropp, E.~L. Lusk, and A.~Skjellum, \emph{Using MPI: portable parallel
  programming with the message-passing interface}.\hskip 1em plus 0.5em minus
  0.4em\relax MIT Press, 1999, vol.~1.

\bibitem{DBLP:journals/tocs/BirrelN84}
A.~Birrell and B.~J. Nelson, ``{Implementing Remote Procedure Calls},''
  \emph{ACM Trans. Comput. Syst.}, vol.~2, no.~1, pp. 39--59, 1984.

\bibitem{DBLP:conf/lics/FredrikssonG13}
O.~Fredriksson and D.~R. Ghica, ``{Abstract Machines for Game Semantics,
  Revisited},'' in \emph{{28th Annual ACM/IEEE Symposium on Logic in Computer
  Science, LICS 2013, New Orleans, LA, USA, June 25-28, 2013}}.\hskip 1em plus
  0.5em minus 0.4em\relax {IEEE Computer Society}, 2013, pp. 560--569.

\bibitem{DBLP:conf/tgc/FredrikssonG12}
------, ``{Seamless Distributed Computing from the Geometry of Interaction},''
  in \emph{{Trustworthy Global Computing - 7th International Symposium, TGC
  2012, Newcastle upon Tyne, UK, September 7-8, 2012, Revised Selected
  Papers}}.\hskip 1em plus 0.5em minus 0.4em\relax {Springer}, 2012, pp.
  34--48.

\bibitem{DBLP:journals/toplas/StamosG90}
J.~W. Stamos and D.~K. Gifford, ``Remote evaluation,'' \emph{ACM Trans.
  Program. Lang. Syst.}, vol.~12, no.~4, pp. 537--565, 1990.

\bibitem{Landin64}
P.~J. Landin, ``The mechanical evaluation of expressions,'' \emph{Computer
  Journal}, vol.~6, no.~4, pp. 308--320, Jan. 1964.

\bibitem{norell:thesis}
U.~Norell, ``Towards a practical programming language based on dependent type
  theory,'' Ph.D. dissertation, Chalmers Uni. of Tech., 2007.

\bibitem{ModernSECD}
\BIBentryALTinterwordspacing
X.~Leroy, ``{MPRI} course 2-4-2, part {II}: abstract machines,'' 2013-2014.
  [Online]. Available: \url{http://gallium.inria.fr/~xleroy/mpri/progfunc/}
\BIBentrySTDinterwordspacing

\bibitem{DBLP:books/daglib/0068837}
P.~Henderson, \emph{{Functional programming - application and implementation}},
  ser. Prentice Hall International Series in Computer Science.\hskip 1em plus
  0.5em minus 0.4em\relax Prentice Hall, 1980.

\bibitem{Felleisen:1986:CEK}
M.~Felleisen and D.~P. Friedman, ``{Control operators, the SECD-machine, and
  the lambda-calculus},'' in \emph{IFIP TC 2/WG 2.2}, Aug. 1986.

\bibitem{DBLP:journals/tcs/Plotkin77}
G.~D. Plotkin, ``{LCF Considered as a Programming Language},'' \emph{Theor.
  Comput. Sci.}, vol.~5, no.~3, pp. 223--255, 1977.

\bibitem{DeBruijn}
N.~G. de~Bruijn, ``Lambda calculus notation with nameless dummies, a tool for
  automatic formula manipulation, with application to the church-rosser
  theorem,'' \emph{Indagationes Mathematicae}, pp. 381--392, 1972.

\bibitem{DBLP:conf/popl/BerryB90}
G.~Berry and G.~Boudol, ``{The Chemical Abstract Machine},'' in
  \emph{{Conference Record of the Seventeenth Annual ACM Symposium on
  Principles of Programming Languages, San Francisco, California, USA, January
  1990}}.\hskip 1em plus 0.5em minus 0.4em\relax {ACM Press}, 1990, pp. 81--94.

\bibitem{DBLP:conf/fpca/Johnsson85}
T.~Johnsson, ``{Lambda Lifting: Treansforming Programs to Recursive
  Equations},'' in \emph{FPCA}, 1985, pp. 190--203.

\bibitem{DBLP:journals/jfp/TrinderLP02}
P.~W. Trinder, H.-W. Loidl, and R.~F. Pointon, ``{Parallel and Distributed
  Haskells},'' \emph{J. Funct. Program.}, vol.~12, no. 4{\&}5, pp. 469--510,
  2002.

\bibitem{DBLP:journals/lisp/LoidlRSHHKLMPPPT03}
H.-W. Loidl, F.~Rubio, N.~Scaife, K.~Hammond, S.~Horiguchi, U.~Klusik,
  R.~Loogen, G.~Michaelson, R.~Pena, S.~Priebe, {\'A}.~J.~R. Portillo, and
  P.~W. Trinder, ``{Comparing Parallel Functional Languages: Programming and
  Performance},'' \emph{Higher-Order and Symbolic Computation}, vol.~16, no.~3,
  pp. 203--251, 2003.

\bibitem{DBLP:books/daglib/0073501}
J.~Armstrong, R.~Virding, and M.~Williams, \emph{{Concurrent programming in
  ERLANG}}.\hskip 1em plus 0.5em minus 0.4em\relax Prentice Hall, 1993.

\bibitem{DBLP:conf/haskell/EpsteinBJ11}
J.~Epstein, A.~P. Black, and S.~L.~P. Jones, ``{Towards Haskell in the
  cloud},'' in \emph{{Proceedings of the 4th ACM SIGPLAN Symposium on Haskell,
  Haskell 2011, Tokyo, Japan, 22 September 2011}}.\hskip 1em plus 0.5em minus
  0.4em\relax {ACM}, 2011, pp. 118--129.

\bibitem{DBLP:journals/ieeecc/WojciechowskiS00}
P.~T. Wojciechowski and P.~Sewell, ``Nomadic pict: language and infrastructure
  design for mobile agents,'' \emph{IEEE Concurrency}, vol.~8, no.~2, pp.
  42--52, 2000.

\bibitem{DBLP:conf/concur/FournetGLMR96}
C.~Fournet, G.~Gonthier, J.-J. L{\'e}vy, L.~Maranget, and D.~R{\'e}my, ``A
  calculus of mobile agents,'' in \emph{CONCUR}, ser. Lecture Notes in Computer
  Science, U.~Montanari and V.~Sassone, Eds., vol. 1119.\hskip 1em plus 0.5em
  minus 0.4em\relax Springer, 1996, pp. 406--421.

\bibitem{DBLP:conf/icfp/MarlowYJ07}
S.~Marlow, A.~R. Yakushev, and S.~L.~P. Jones, ``{Faster laziness using dynamic
  pointer tagging},'' in \emph{{Proceedings of the 12th ACM SIGPLAN
  International Conference on Functional Programming, ICFP 2007, Freiburg,
  Germany, October 1-3, 2007}}.\hskip 1em plus 0.5em minus 0.4em\relax {ACM},
  2007, pp. 277--288.

\bibitem{DBLP:journals/jfp/LoogenOP05}
R.~Loogen, Y.~Ortega-Mall{\'e}n, and R.~Pe{\~n}a-Mar\'{\i}, ``{Parallel
  functional programming in Eden},'' \emph{J. Funct. Program.}, vol.~15, no.~3,
  pp. 431--475, 2005.

\bibitem{DBLP:conf/oopsla/SerranoGL06}
M.~Serrano, E.~Gallesio, and F.~Loitsch, ``{Hop: a language for programming the
  web 2.0},'' in \emph{{Companion to the 21th Annual ACM SIGPLAN Conference on
  Object-Oriented Programming, Systems, Languages, and Applications, OOPSLA
  2006, October 22-26, 2006, Portland, Oregon, USA}}.\hskip 1em plus 0.5em
  minus 0.4em\relax {ACM}, 2006, pp. 975--985.

\bibitem{DBLP:conf/fmco/CooperLWY06}
E.~Cooper, S.~Lindley, P.~Wadler, and J.~Yallop, ``{Links: Web Programming
  Without Tiers},'' in \emph{{Formal Methods for Components and Objects, 5th
  International Symposium, FMCO 2006, Amsterdam, The Netherlands, November
  7-10, 2006, Revised Lectures}}.\hskip 1em plus 0.5em minus 0.4em\relax
  {Springer}, 2006, pp. 266--296.

\bibitem{DBLP:conf/tgc/VIICH07}
T.~M. VII, K.~Crary, and R.~Harper, ``{Type-Safe Distributed Programming with
  ML5},'' in \emph{{Trustworthy Global Computing, Third Symposium, TGC 2007,
  Sophia-Antipolis, France, November 5-6, 2007, Revised Selected
  Papers}}.\hskip 1em plus 0.5em minus 0.4em\relax {Springer}, 2007, pp.
  108--123.

\bibitem{DBLP:conf/iwmm/PlainfosseS95}
D.~Plainfoss{\'e} and M.~Shapiro, ``{A Survey of Distributed Garbage Collection
  Techniques},'' in \emph{{Memory Management, International Workshop IWMM 95,
  Kinross, UK, September 27-29, 1995, Proceedings}}.\hskip 1em plus 0.5em minus
  0.4em\relax {Springer}, 1995, pp. 211--249.

\end{thebibliography}
\fi

\end{document}